\documentclass[12pt, onecolumn, draftclsnofoot]{IEEEtran}

\usepackage{amsmath, amssymb, amsfonts, amsthm}
\usepackage{bm}
\usepackage{xcolor}
\usepackage{framed}

\usepackage{pgf}
\usepackage{tikz}
\usetikzlibrary{arrows,automata}
\usetikzlibrary{positioning}
\usetikzlibrary{decorations.pathmorphing}
\usetikzlibrary{shapes.geometric}
\usetikzlibrary{fit}
\usetikzlibrary{backgrounds}

\usepackage[caption=false,font=footnotesize]{subfig}

\newcommand{\mc}{\mathcal}
\newcommand{\mbb}{\mathbb}
\newcommand{\maxsub}{\mathrm{max}}
\newcommand{\sem}{\mathrm{sem}}
\newcommand{\str}{\mathrm{str}}

\DeclareMathOperator{\rank}{rank}
\DeclareMathOperator{\supp}{supp}

\newtheorem{lem}{Lemma}
\newtheorem{thm}[lem]{Theorem}
\newtheorem{cor}[lem]{Corollary}
\theoremstyle{definition}
\newtheorem{defn}[lem]{Definition}
\newtheorem{ex}[lem]{Example}
\theoremstyle{remark}
\newtheorem{rem}[lem]{Remark}

\usepackage{hyperref}

\title{Semantic Security via Seeded Modular Coding Schemes and Ramanujan Graphs}

\author{Moritz~Wiese and Holger Boche%
\thanks{M.~Wiese is with the Institute of Theoretical Information Technology, Technical University of Munich, 80333 M\"unchen, Germany and with the CASA -- Cyber Security in the Age of Large-Scale Adversaries -- Excellenzcluster, Ruhr Universit\"at Bochum, Bochum, Germany.
}
\thanks{H.~Boche is with the Institute of Theoretical Information Technology, Technical University of Munich, 80333 M\"unchen, Germany and with the Munich Center for Quantum Science and Technology (MCQST), Schellingstr. 4, 80799 M\"unchen, Germany. 
}
\thanks{M.~Wiese was supported by the Deutsche Forschungsgemeinschaft (DFG, German Research Foundation) under Germany's Excellence Strategy - EXC 2092 CASA - 390781972 and in part by the DFG within the Gottfried Wilhelm Leibniz Prize under Grant BO 1734/20-1.}
\thanks{H.~Boche was supported by the Deutsche Forschungsgemeinschaft (DFG, German Research Foundation) within Germany's Excellence Strategy EXC-2111 390814868 and within the Gottfried Wilhelm Leibniz-program under Grant
BO 1734/221.}
\thanks{This paper was presented in part at the 2019 IEEE International Symposium on Information Theory, Paris, July 2019, and at a meeting between Technical University of Munich and the German Federal Office for Information Security (BSI) on physical layer security.}
}

\date{\today}

\begin{document}

\maketitle

\begin{abstract}
    A novel type of functions called biregular irreducible functions is introduced and applied as security components (instead of, e.g., universal hash functions) in seeded modular wiretap coding schemes, whose second component is an error-correcting code. These schemes are called modular BRI schemes. An upper bound on the semantic security information leakage of modular BRI schemes in a one-shot setting is derived which separates the effects of the biregular irreducible function on the one hand and the error-correcting code plus the channel on the other hand. The effect of the biregular irreducible function is described by the second-largest eigenvalue of an associated stochastic matrix. A characterization of biregular irreducible functions is given in terms of connected edge-disjoint biregular graphs. It allows for the construction of new biregular irreducible functions from families of edge-disjoint Ramanujan graphs, which are shown to exist. A concrete and frequently used arithmetic universal hash function can be converted into a biregular irreducible function for certain parameters. Sequences of Ramanujan biregular irreducible functions are constructed which exhibit an optimal trade-off between the size of the regularity set and the rate of decrease of the associated second-largest eigenvalue. Together with the one-shot bound on the information leakage, the existence of these sequences implies an asymptotic coding result for modular BRI schemes applied to discrete and Gaussian wiretap channels. It shows that the separation of error correction and security as done in a modular BRI scheme is secrecy capacity-achieving for every discrete and Gaussian wiretap channel. The same holds for a derived construction where the seed is generated locally by the sender and reused several times. It is shown that the optimal sequences of biregular irreducible functions used in the above constructions must be nearly Ramanujan. 
\end{abstract}


\section{Introduction}

\subsection{Semantic security}

In the wiretap channel problem, a sender has a set of messages and would like to transmit one of these messages to an intended receiver. To this end, the message is encoded and then sent through a given noisy channel to the intended receiver, who decodes the channel output. An eavesdropper observes a different noisy version of the sent codeword. In a one-shot scenario, the goal is to find an encoding of the messages which allows for transmission to the intended message recipient with small error probability, whereas the eavesdropper obtains little information about the transmitted message. For a memoryless wiretap channel, both the probability of erroneous transmission to the intended recipient and the information leakage to the eavesdropper should tend to 0 with increasing blocklength.

The information leakage to the eavesdropper is measured with the help of a security measure. In this paper, we focus on the security measure of \textit{semantic security}. The semantic security information leakage is defined to be less than $\varepsilon$ if $\max_MI(M\wedge Z)\leq\varepsilon$, where the maximum is over all random variables on the message set, $Z$ is the eavesdropper's noisy observation of $M$, and $I(X\wedge Y)$ is the mutual information of random variables $X$ and $Y$. Semantic security was introduced in information theory by Bellare, Tessaro and Vardy\footnote{\cite{BT_Poly_time} and \cite{BTV_cryptWiretap} are unpublished extended versions of \cite{BTV_semSecWiretap}. We only cite the more detailed unpublished papers.} \cite{BTV_cryptWiretap}, \cite{BTV_semSecWiretap}. It is a stronger requirement than strong secrecy as defined by Maurer \cite{Maurer_StrongSec} and Ahlswede and Csisz\'ar \cite{AC_SecKey}, where the message is uniformly distributed. It is argued in \cite{BTV_cryptWiretap} that semantic security should be adopted as the standard secrecy measure in information-theoretic security, not least because it is the information-theoretic analog to the cryptographic definition of semantic security introduced by Goldwasser and Micali \cite{GoldwasserMicali} (see also Goldreich's book \cite{Goldreich}). One of the possible cryptographic formulations is the \textit{indistinguishability of encryptions}: There is no message pair for which an eavesdropper can computationally distinguish the two encrypted messages of this pair.

\subsection{Seeded modular coding schemes}

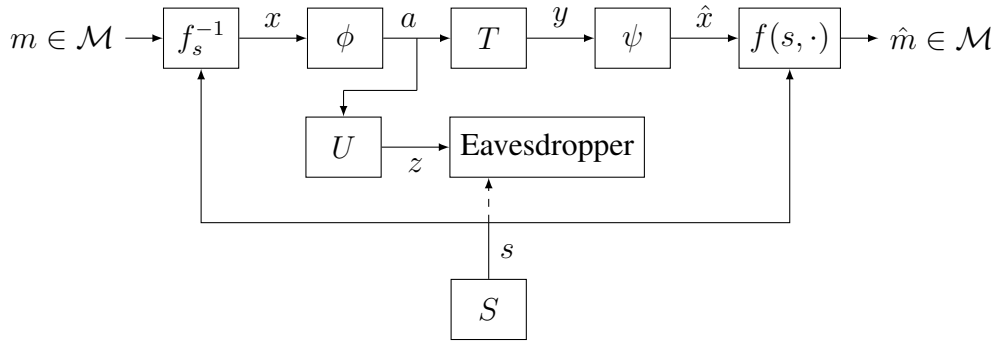
\begin{figure*}
	\centering
	\begin{tikzpicture}[kasten/.style={draw, minimum height = .8cm, minimum width = 1cm, inner sep=4pt}, pfeil/.style={->, >=latex}
	]
    \node (message) {$m\in\mc M$};
    \node[kasten, right = .5cm of message] (inv_uhf) {$f_s^{-1}$};
    \node[kasten, right = .9cm of inv_uhf] (encoder) {$\phi$};
    
    \node[right = .3cm of encoder] (aux1) {};
    \node[kasten, right = .3cm of aux1] (main_channel) {$T$};
    \node[kasten, right = .9cm of main_channel] (decoder) {$\psi$};
    \node[kasten, right = .9cm of decoder] (uhf) {$f(s,\cdot)$};
    \node[right = .5cm of uhf] (rec_message) {$\hat m\in\mc M$};
    
    \node[below = .4cm of aux1] (aux2) {};
    \node[left = .68cm of aux2] (aux3) {};
    \node[kasten, below = .2cm of aux3] (eve_channel) {$U$};
    \node[kasten, right = .9cm of eve_channel] (eavesdropper) {Eavesdropper};
    
    \node[below = 1.9cm of main_channel] (seed_aux) {};
    \node[kasten, below = .6cm of seed_aux] (seed) {$S$};
    
    \draw[pfeil] (message) -> (inv_uhf);
    \draw[pfeil] (inv_uhf) -> (encoder) node[above, midway] {$x$};
    \draw (encoder) -> (aux1.center) node[above, near end] {$a$};
    \draw[pfeil] (aux1.center) -> (main_channel);
    \draw[pfeil] (main_channel) -> (decoder) node[above, midway] {$y$};
    \draw[pfeil] (decoder) -> (uhf) node[above, midway] {$\hat x$};
    \draw[pfeil] (uhf) -> (rec_message);
    
    \draw (aux1.center) -- (aux2.center);
    \draw (aux2.center) -- (aux3.center);
    \draw[pfeil] (aux3.center) -> (eve_channel);
    \draw[pfeil] (eve_channel) -> (eavesdropper) node[below, midway] {$z$};
    
    \draw[pfeil, dashed] (seed_aux.center) -- (seed_aux.center|-eavesdropper.south);
    \draw[pfeil] (seed_aux.center) -| (inv_uhf);
    \draw[pfeil] (seed_aux.center) -| (uhf);
    \draw (seed) -- (seed_aux.center)  node[right, midway] {$s$};
 \end{tikzpicture}
	\caption{A seeded modular coding scheme for the wiretap channel $(T,U)$. $T$ denotes the physical channel between the sender and the intended receiver and $U$ the physical channel between the sender and the eavesdropper. $(\phi,\psi)$ is an error-correcting code and $f$ the security component. $f_s^{-1}$ denotes the randomized inverse of $f(s,\cdot)$. The seed $s$ is generated by a uniformly distributed random variable $S$ on the seed set and has to be known to sender and receiver beforehand. It is important that it may also be known to the eavesdropper.}\label{fig:BT_scheme}
\end{figure*}

This paper studies \textit{seeded modular coding schemes} consisting of two components which enhance ordinary error-correcting codes in order to provide semantic security against an eavesdropper. In addition to an error-correcting code, which we denote by its encoder-decoder pair $(\phi,\psi)$, the second component from which such a seeded modular coding scheme is constructed is a function $f:\mc S\times\mc X\to\mc N$, where $\mc X$ is a subset of the message set of $(\phi,\psi)$, the finite set $\mc S$ is called the \textit{seed set} and the message set $\mc M$ is a subset of $\mc N$. To transmit the message $m\in\mc M$, the seeded modular coding scheme constructed from $(\phi,\psi)$ and $f$ works as follows: First, a \textit{seed} $s\in\mc S$ is chosen uniformly at random. Then the \textit{randomized inverse} $f_s^{-1}(\cdot\vert m)$ uniformly at random picks an element $x$ of the set $\{x':f(s,x')=m\}$. This $x$ is encoded using $\phi$, transmitted over the channel to the intended receiver and decoded by $\psi$ into an $\hat x\in\mc X$. The final step at the decoder's side is to map $\hat x$ to $\hat m=f(s,\hat x)$ (see Fig. \ref{fig:BT_scheme}). Clearly, the seed must be known to the sender and the receiver, meaning that they must have access to sufficient common randomness. Reliable transmission of messages chosen from $\mc M$ is possible due to the error correction performed by $(\phi,\psi)$. Establishing security is left to $f$, hence we will sometimes call $f$ the ``security component'' of the seeded modular coding scheme. 

The information leakage is measured under the assumption that the eavesdropper knows the seed, i.e., one takes $\max_MI(M\wedge Z,S)$, where $M$ and $Z$ are as above and $S$ is uniformly distributed on $\mc S$ and independent of $M$. On memoryless wiretap channels, the sender can therefore generate the seed locally and transmit it to the intended receiver before the actual message transmission starts. This makes the use of common randomness unnecessary. For a sequence of seeded modular coding schemes whose error probability and information leakage tend to zero, the rate loss due to seed transmission can be made negligible while preserving security by reusing the seed not too often.

\subsection{Universal hash functions}

Various types of security components and corresponding seeded modular coding schemes have been investigated so far. One of them consists of \textit{universal hash functions} $f:\mc S\times\mc X\to\mc M$ satisfying $\mbb P[f(S,x)=f(S,x')]\leq\lvert\mc M\rvert^{-1}$ for $S$ uniformly distributed on $\mc S$ and for all $x\neq x'$ (in particular, every output of $f$ can be a message). We call a seeded modular coding scheme whose security component is a universal hash function a \textit{modular UHF scheme}. The concept of universal hash function is due to Carter and Wegman \cite{CW_hash}. Universal hash functions were first used in information theory by Bennett, Brassard and Robert \cite{BBR_priv_ampl}. Modular UHF schemes were proposed as a technique for wiretap coding by Hayashi \cite{Hay_expdecr}. 

It is shown by Bellare and Tessaro \cite{BT_Poly_time} and Tal and Vardy \cite{TalVardy_Upgrading} that the secrecy capacity of any discrete, degraded and symmetric wiretap channel as derived by Wyner \cite{Wyner} and Csisz\'ar and K\"orner \cite{CK_Wiretap} is achievable by modular UHF schemes such that semantic security is guaranteed. The proofs make heavy use of the symmetry of the wiretap channel. They also require the additional property that the sets $\{x:f(s,x)=m\}$ have the same size for all $s$ and $m$. Bellare and Tessaro give the example of such a universal hash function $\beta^o$ based on finite-field arithmetic which is efficiently computable and invertible (its definition seems to go back to the work of Bennett et al.~\cite{BBCM_genprivampl}). In combination with an efficient linear code, the resulting modular UHF scheme is efficient as well and provides a very practical and flexible way of achieving optimal rates and high security for discrete, degraded and symmetric wiretap channels.

A general universal hash function is tightly linked to the uniform distribution on the message set through the leftover hash lemma \cite{ILL_89} (see also \cite{BBR_priv_ampl,BBCM_genprivampl}). Thus it is not surprising that a general modular UHF scheme requires some symmetry of the channel in order to provide security for every possible distribution of the messages. 

A different type of security components, which also encompasses a subset of the universal hash functions, is used by Hayashi and Matsumoto \cite{HayMat_Multiplex}. Their \textit{inverses} are defined in terms of group homomorphisms. A basic lemma on the channel resolvability achievable with these security components (\cite[Lemma 21]{HayMat_Multiplex}) can be extended in order to show that they achieve semantic security for arbitrary discrete memoryless wiretap channels when applied in a seeded modular coding scheme. However, the seed they require is longer than that needed by, e.g., the function $\beta^o$ above. The details are worked out in Appendix \ref{app:HayMat}.

\subsection{Contributions}

\paragraph{Biregular irreducible functions}

In this paper, a novel type of security components, called \textit{biregular irreducible functions}, is introduced. These give rise to \textit{modular BRI schemes}. In a modular BRI scheme whose security component is the biregular irreducible function $f:\mc S\times\mc X\to\mc N$, in contrast to a modular UHF scheme, in general not the whole set $\mc N$ becomes the message set. Instead, $f$ comes with a \textit{regularity set} $\mc M\subset\mc N$ which is used as the set of confidential messages which can be transmitted to the intended message recipient. The condition for $f$ to be a biregular irreducible function is that for every $m\in\mc M$, the bipartite graph $G_{f,m}$ with bipartition $(\mc S,\mc X)$ (i.e., edges only go between $\mc S$ and $\mc X$) where $s$ is adjacent to $x$ if $f(s,x)=m$ is biregular, and that the second-largest eigenvalue modulus $\lambda_2(f,m)$ of the $\mc X\times\mc X$ stochastic matrix $P_{f,m}$ with $(x,x')$ entry proportional to $\lvert\{s:f(s,x)=f(s,x')=m\}\rvert$ is strictly smaller than 1. It will be seen that biregular irreducible functions establish an interesting connection between memoryless wiretap channels and the asymptotics of families of biregular graphs with small second-largest eigenvalues. As a connection with the previously mentioned security components, we show that biregular irreducible functions are universal hash functions on average.

\paragraph{Modular BRI schemes}

It will be shown that modular BRI schemes, and thus also codes without common randomness constructed from modular BRI schemes by seed reuse, can establish semantically secure message transmission up to the secrecy capacity of all memoryless discrete and Gaussian wiretap channels. This is a major improvement over the general modular UHF schemes of \cite{BT_Poly_time} and provides an alternative to the algebraic security components of \cite{HayMat_Multiplex}.

The main upper bound on the semantic security information leakage incurred by modular BRI schemes derived in this paper can be formulated in a one-shot setting. Recall that this leakage is given by $\max_MI(M\wedge Z,S)$. The independence from the message distribution is obtained by a reduction to an upper bound on a divergence term for every single message, which is nothing other than a channel resolvability bound \cite{HV_Approx} for every message. More precisely, let $f:\mc S\times\mc X\to\mc N$ be a biregular irreducible function with regularity set $\mc M$ and let $W:\mc X\to\mc Z$ be any channel to an eavesdropper, e.g., the concatenation of the encoder of an error-correcting code $\phi$ and the physical channel $U$ to the eavesdropper as in Fig.~\ref{fig:BT_scheme}. For a fixed seed $s$ and message $m$, assume that $K_m(\cdot\vert s)$ describes the eavesdropper's output distribution given that $m$ is first passed into the randomized inverse $f_s^{-1}$ of $f$ and then transmitted through $W$. The divergence term for $m$ which we upper-bound is the conditional Kullback-Leibler divergence $D(K_m\Vert P_{\mc X}W\vert P_{\mc S})$, where $P_{\mc X}W$ is the probability distribution on $\mc Z$ generated by the uniform distribution $P_{\mc X}$ on $\mc X$ via $W$ and $P_{\mc S}$ is the uniform distribution on $\mc S$ (in other words, the distribution of $S$). The proof of this upper bound is inspired by the proof of a ``leftover hash lemma'' for modular UHF schemes in the context of strong secrecy due to Tyagi and Vardy \cite[Lemma 4]{TV_UHF_preprint}.

The simple upper bound on $D(K_m\Vert P_{\mc X}W\vert P_{\mc S})$ which we obtain separates the influence of the biregular irreducible function from that of $W$. Up to some small terms, the upper bound is given by the product of the second-largest eigenvalue modulus $\lambda_2(f,m)$ of $P_{f,m}$ on the one hand and $\exp(D_2^\varepsilon(W\Vert P_{\mc X}W\vert P_{\mc X}))$ on the other hand, where $D_2^\varepsilon(W\Vert P_{\mc X}W\vert P_{\mc X})$ is an $\varepsilon$-smooth conditional R\'enyi $2$-divergence. Thus the separation of the tasks of error correction (which is incorporated in $W$) and generation of security is reflected in the form of the upper bound. This resembles the structure of the upper bounds obtained in classical leftover hash lemmas for privacy amplification as stated by, e.g., \cite{BBCM_genprivampl} and \cite[Lemma 3]{TV_UHF_preprint}, as well as similar bounds for the wiretap channel as in \cite[Lemma 4]{TV_UHF_preprint} and \cite[Lemma 21]{HayMat_Multiplex}.

\paragraph{Constructions of biregular irreducible functions}

The most important example of a biregular irreducible function is where every $G_{f,m}$ is $(d_1,d_2)$-biregular and \textit{Ramanujan}, which means that the second-largest eigenvalue of any $G_{f,m}$ is at most $\sqrt{d_1-1}+\sqrt{d_2-1}$. Our construction builds on the proof of the existence of Ramanujan graphs from \cite{MSS_interl_fams_i}. Since this proof is nonexplicit, we have no efficiently computable Ramanujan biregular irreducible functions so far. 

Ramanujan and nearly Ramanujan graphs are optimal or very good \textit{expander graphs}, respectively. The first examples were constructed independently by Lubotzky, Phillips and Sarnak \cite{LPS_Ramanujan} and Margulis \cite{Margulis_Ramanujan}. Expanders are a very active field of research and have many applications in mathematics, computer science and engineering. A good overview is given by Hoory, Linial and Wigderson \cite{HLW_expanders}. 

As a second example, it is shown that a universal hash function $\beta^o$ which was used in \cite{BT_Poly_time} and \cite{TV_UHF_preprint} as a component of modular UHF schemes for suitable parameters in fact is a biregular irreducible function $\beta$ with a large and completely known regularity set $\mc M$ and small $\lambda_2(\beta,m)$ for every $m\in\mc M$. Although this function at first sight seems to be more promising, it cannot be computed efficiently, either. 

Thus, no efficient constructions are known so far, in constrast to universal hash functions. We do not think that this is a problem inherent to the concept, but it means that no efficient high-rate modular BRI scheme is available at the moment. Compared to the security components of \cite{HayMat_Multiplex}, of which there exists at least one efficient example as discussed in Appendix \ref{app:HayMat}, our examples require a shorter seed.

\paragraph{Asymptotic performance}

The one-shot upper bound on the semantic security leakage of a modular BRI scheme and the constructions of biregular irreducible functions do not yet indicate how modular BRI schemes compare to other wiretap codes. We test them in the asymptotic setting for discrete and Gaussian memoryless wiretap channels. We show the existence of sequences of (Ramanujan) biregular irreducible functions with real parameters $r\geq 0$ and $0\leq t<1$ such that every sequence of modular BRI schemes constructed from these biregular irreducible functions and any sequence of error-correcting codes with rate larger than $r$ achieves rate $(1-t)r$ if $tr$ is larger than a mutual information term determined by the channel to the eavesdropper. These two simple criteria reflect the separation of error correction and security generation which is present in modular BRI schemes.

It follows from this coding theorem that the secrecy capacity of arbitrary discrete and Gaussian wiretap channels is achievable with semantic security by modular BRI schemes in a simple way. Thus one loses nothing by separating error correction and the generation of semantic security. The same then also holds for the codes without common randomness constructed from the modular BRI schemes through seed reuse.  Due to the similarity of their one-shot upper bounds, analogous results could be formulated based on the bound from \cite[Lemma 21]{HayMat_Multiplex} and, in the case of strong secrecy, on \cite[Lemma 4]{TV_UHF_preprint}. The results of our paper have been extended recently to finite-dimensional classical and fully quantum wiretap channels \cite{QuantenIsit,QuantenLongVersion}. The quantum Gaussian wiretap channel remains an open problem.

In a further analysis of sequences of biregular irreducible functions $(f_i)_{i=1}^\infty$ with given parameters $r$ and $t$ as above, it turns out that they are optimal in terms of the trade-off between the rate of increase of the cardinalities of the regularity sets $\mc M_i$ vs.\ the rate of decrease of the second-largest eigenvalue moduli $\max_{m\in\mc M_i}\lambda_2(f_i,m)$. Interestingly, this follows from the coding theorem for discrete memoryless wiretap channels and our one-shot upper bound for the semantic security information leakage of modular BRI schemes. In a rather loose sense, these optimal sequences of biregular irreducible functions also are nearly Ramanujan. From the asymptotic bound of Feng and Li \cite{FengLi} on the second-largest eigenvalue of biregular graphs, it follows that the maximum of the associated degree pair has to grow exponentially in the blocklength.

\subsection{Other codes for semantic security}

Semantic security is shown implicitly in resolvability-based proofs of strong secrecy like in Devetak \cite{Devetak}, Hayashi \cite{Hay_longtitle} and Bloch and Laneman \cite{BlochLaneman_resolvability}. It is an explicit goal of random coding in the resolvability-based works of Goldfeld, Cuff and Permuter \cite{Goldfeld_wiretapII, Goldfeld_AVWC}, Bunin et al. \cite{BGPSCP_semsec} and Frey, Bjelakovi\'c and Sta\'nczak \cite{FreyBjelaStancz_MAC_semsec}.
 
Liu, Yan and Ling \cite{LYL_polar_semsec} use efficient polar codes to prove that the secrecy capacity of Gaussian wiretap channels is achievable with semantic security. The disadvantage of these codes is that polar codes, like codes found by random coding, do not separate the tasks of error-correction and security generation. To the authors' knowledge, no other codes apart from modular UHF schemes, polar codes and random codes have been shown to achieve semantic security for specific scenarios. However, it has been observed by Renes and Renner \cite{RR_codingfromaplif} and Hayashi and Matsumoto \cite{HayMat_Multiplex} that every code sequence which ensures strong secrecy on a single-state channel also ensures semantic security. This can be shown nonconstructively by an expurgation argument. We discuss this phenomenon in Appendix \ref{app:strongsec}.

Bloch, Hayashi and Thangaraj \cite{BHT_survey} give a nice survey of code constructions for wiretap channels for semantic security and weaker security measures.

\subsection{Overview}

The paper has two parts. Sections \ref{sect:prelims}-\ref{sect:asymptotics} contain the ``main story'' including proofs which are not too complex, Sections \ref{sect:proof_EV-UB}-\ref{sect:EV-qualset-sum_proof} are stand-alone sections which contain the more complex proofs of results from the first part. Section \ref{sect:conclusion} concludes the paper and briefly discusses the complexity of biregular irreducible functions and the coding schemes where they are applied.

Section \ref{sect:prelims} introduces the principal information-theoretic concepts used in this paper. Section \ref{sect:oneshotwiretap} defines one-shot wiretap channels and the corresponding concepts of wiretap codes. Biregular irreducible functions, modular BRI schemes and the upper bound on the incurred semantic security leakage are presented in Section \ref{sect:BRI}. In Section \ref{sect:Ramanujan}, we characterize biregular irreducible functions in graph-theoretic terms and state the existence of good Ramanujan biregular irreducible functions. Section \ref{sect:BT-BRI} contains the results on the function $\beta^o$ as a biregular irreducible function. The asymptotic analysis of modular BRI schemes and related wiretap codes is done in Section \ref{sect:asymptotics} together with an analysis of asymptotically optimal biregular irreducible functions. Sections \ref{sect:proof_EV-UB}-\ref{sect:EV-qualset-sum_proof} contain most of the proofs. The appendices contain some small auxiliary results, a discussion of the relation between the notion of strong secrecy and semantic security, a comparison of our results with those of \cite{HayMat_Multiplex} and some facts about graphs.

\section{Preliminaries}\label{sect:prelims}

\subsection{Basic definitions and notation}

    For a set $\mc A$ and a subset $\mc B\subset\mc A$, by $\mc A\setminus\mc B$ we mean the set difference of $\mc A$ and $\mc B$. If $E$ is any event, then $1_E$ equals $1$ if the event occurs and $0$ otherwise. The logarithm $\log$ and the exponential function $\exp$ will always be taken to base 2, the natural logarithm is denoted by $\ln$.
    
    The distribution of a random variable $X$ is denoted by $P_X$. If $X,Y$ are random variables with joint distribution $P_{XY}$, the conditional distribution of $Y$ given $X$ is written $P_{Y\vert X}$. The distribution obtained by fixing a realization $x$ of $X$ is denoted by $P_{Y\vert X=x}$. The uniform distribution on a finite set $\mc X$ is denoted by $P_{\mc X}$.
    
    If $\mc X$ is any finite set, then $\mbb R^{\mc X}$ denotes the set of real-valued functions on $\mc X$. $\mbb R^{\mc X}$ is isomorphic to $\mbb R^{\lvert\mc X\rvert}$. Similarly, we will work with matrices from $\mbb R^{\mc S\times\mc X}$. A matrix is called \textit{stochastic} if it has nonnegative entries and the entries of every row sum to $1$. A symmetric matrix $A\in\mbb R^{\mc X\times\mc X}$ is diagonalizable with real eigenvalues $\mu_1\geq\mu_2\geq\cdots\geq\mu_{\lvert\mc X\rvert}$. In this situation, the algebraic multiplicity of an eigenvalue is the same as its geometric multiplicity, and one can just speak of its \textit{multiplicity}. If $A$ also has nonnegative entries and constant row sums (e.g., if it is stochastic), then the all-one vector $\bm 1$ is an eigenvector for the largest eigenvalue. We always associate it with $\mu_1$. Then the \textit{second-largest eigenvalue modulus} of $A$ is $\max(\lvert\mu_2\rvert,\lvert\mu_{\lvert\mc X\rvert}\rvert)$. An eigenvalue with multiplicity 1 is called \textit{simple}.

\subsection{Basic probability definitions}

A \textit{measurable space} is a set $\mc Z$ equipped with a sigma algebra of measurable sets, which is suppressed in the notation. We will have to deal with probability measures $P$ satisfying $P(\mc Z)=1$ and with \textit{subnormalized measures} $M$ for which $0<M(\mc Z)\leq 1$ holds. In particular, probability measures are subnormalized measures as well, and every subnormalized measure can be turned into a probability measure by appropriate normalization. All subnormalized measures $M$ on any measurable set $\mc Z$ considered in this paper have a density $f$ with respect to some reference measure $\mu$ on $\mc Z$, i.e.,
\[
    M(\mc Z')=\int_{\mc Z'} m(z)\mu(dz)
\]
for any measurable $\mc Z'\subset\mc Z$. 

\begin{ex}\label{ex:discrete}
    If $\mc Z$ is a discrete set, then we will always assume that the reference measure is the \textit{counting measure} defined by $\mu(\mc Z')=\lvert\mc Z'\rvert$. Every subnormalized measure $M$ on $\mc Z$ has a density $m$ with respect to $\mu$ and
    \[
        M(\mc Z')=\sum_{z\in\mc Z'}m(z)\mu(z)=\sum_{z\in\mc Z'}m(z).
    \]
\end{ex}

\begin{ex}\label{ex:finite-support}
    We will also encounter subnormalized measures $M$ on arbitrary non-measurable sets $\mc A$. Such an $M$ is always defined as a discrete measure on a finite subset $\mc A'\subset\mc A$. On $\mc A'$, $M$ has a density $m$ with respect to the counting measure on $\mc A'$. The set $\{a\in\mc A':m(a)>0\}$ is called the \textit{support} of $M$ and denoted by $\supp(M)$; $M$ itself is said to have \textit{finite support}.  
\end{ex}

\begin{ex}
    The Gaussian distribution on $\mc Z=\mbb R$ with mean $a$ and variance $\sigma^2$  has the usual density 
    \begin{equation}\label{eq:Gauss_density}
        \frac{1}{\sqrt{2\pi\sigma^2}}e^{-\frac{(z-a)^2}{2\sigma^2}}
    \end{equation}    
    with respect to Lebesgue measure.
\end{ex}

\begin{ex}
    If $M_1$ has $\mu_1$-density $m_1$ and $M_2$ has $\mu_2$-density $m_2$, then the product $M_1\otimes M_2$ of $M_1$ and $M_2$ has density $r(x,y)=m_1(x)m_2(y)$ with respect to the product measure $\mu_1\otimes\mu_2$ determined by the rule $(\mu_1\otimes\mu_2)(\mc X'\times\mc Y')=\mu_1(\mc X')\mu_2(\mc Y')$.
\end{ex}

\begin{ex}
If the random variables $X$ and $Y$ assume their values in the Cartesian product $\mc X\times\mc Y$ and have joint density $p_{XY}$ with respect to the product measure $\mu\otimes\nu$, then $P_X$ has the $\mu$-density
\[
    p_X(x)=\int p_{XY}(x,y)\,\nu(y)
\]
and the conditional distribution $P_{Y\vert X=x}$ has the density
\[
    p_{Y\vert X}(y\vert x)=\frac{p_{XY}(x,y)}{p_X(x)}
\]
$P_X$-almost surely.
\end{ex}

If $M_1$ and $M_2$ are subnormalized measures on $\mc Z$ which have densities $m_1$ and $m_2$ with respect to the common reference measure $\mu$, then their \textit{total variation distance} is defined by
\[
    \lVert M_1-M_2\rVert=\int\lvert m_1(z)-m_2(z)\rvert\,\mu(dz).
\]
The \textit{Kullback-Leibler divergence} of $M_1$ and $M_2$ is given by
\[
    D(M_1\Vert M_2)=
	\begin{cases}
		\int m_1(z)\log\frac{m_1(z)}{m_2(z)}\,\mu(dz) & \text{if } \mu(m_2=0,m_1>0)=0,\\
		+\infty & \text{else}.
	\end{cases}
\]
The \textit{R\'enyi 2-divergence} of $M_1$ and $M_2$ is given by
\[
    D_2(M_1\Vert M_2)=
	\begin{cases}
		\log\int\frac{m_1(z)^2}{m_2(z)}\,\mu(dz) & \text{if }\mu(m_2=0,m_1>0)=0,\\
		+\infty & \text{else.}
	\end{cases}
\]
(The two divergences are traditionally only defined for probability distributions.) If $X$ and $Y$ have joint distribution $P_{XY}$, then the \textit{mutual information} of $X$ and $Y$ is given by
\[
    I(X\wedge Y)=D(P_{XY}\Vert P_X\otimes P_Y).
\]
Now assume that $P_{XY}$ has a density with respect to the reference measure $\mu\otimes\nu$. Then we also introduce the \textit{entropy}
\[
    H(X)=-\int p_X(x)\log p_X(x)\,\mu(dx),
\]
where $p_X$ is the $\mu$-density of $X$, and \textit{conditional entropy}
\[
    H(X\vert Y)=\int p_Y(y)H(X\vert y)\,\nu(dy),
\]
where $p_Y$ is the $\nu$-density of $Y$ and the random variable $X\vert y$ has distribution $P_{X\vert Y=y}$. Then 
\[
    I(X\wedge Y)=H(X)-H(X\vert Y).
\]

\subsection{Channels}\label{subsect:channels}

A \textit{subnormalized channel} $\tilde W$ with input alphabet $\mc A$ and measurable output alphabet $\mc Z$ assigns to every $a\in \mc A$ a subnormalized measure $\tilde W(\cdot\vert a)$ on $\mc Z$. If $\tilde W(\cdot\vert a)$ is a probability measure for every $a$, then we call $\tilde W$ an \textit{ordinary channel} or just a \textit{channel}\footnote{This definition of a channel does not encompass all concepts called ``channel'' in information theory. For example, channels with states (random or arbitrary) are not channels in the sense of this paper.}. To indicate the input and output alphabets of a subnormalized channel $\tilde W$, we will often write $\tilde W:\mc A\to\mc Z$. This should not lead to confusion with the analogous notation for functions. We will always assume that $\tilde W(\cdot\vert a)$ has a density $\tilde w(\cdot\vert a)$ with respect to some common reference measure $\mu$ on $\mc Z$ for every $a\in\mc A$, i.e.,
\[
	\tilde W(\mc Z'\vert a)=\int_{\mc Z'}\tilde w(z\vert a)\,\mu(dz)
\]
for every measurable $\mc Z'\subset\mc Z$.  We then say that $\tilde w$ is \textit{the} density of $\tilde W$.

\begin{ex}
	An ordinary channel $W:\mc A\to\mc Z$ with both $\mc A$ and $\mc Z$ finite is called a \textit{discrete channel}. Like for subnormalized measures on finite sets, the density is always taken with respect to the counting measure $\mu$ (see Example \ref{ex:discrete}). $W$ then is determined by the stochastic matrix $(w(z\vert a))_{a\in\mc A,z\in\mc Z}$ satisfying
	\[
		W(\mc Z'\vert a)=\int_{\mc Z'} w(z\vert a)\,\mu(dz)=\sum_{z\in\mc Z'}w(z\vert a)
	\]
	for every subset $\mc Z'$ of $\mc Z$.
\end{ex}

\begin{ex}
	The \textit{additive Gaussian noise channel} $W$ with noise variance $\sigma^2$ has alphabets $\mc A=\mc Z=\mbb R$. If $\mu$ is the Lebesgue measure on $\mc Z$, then $W$ has a density $w$ with respect to $\mu$ such that $w(z\vert a)$ equals \eqref{eq:Gauss_density}.
\end{ex}

\begin{ex}\label{ex:memoryless}
	If the subnormalized channel $\tilde W:\mc A\to\mc Z$ has density $\tilde w$ with respect to the measure $\mu$ on $\mc Z$, then the \textit{blocklength-$n$ memoryless extension} $\tilde W^n$ of $\tilde W$ is determined by the density
	\[
		\tilde w^n(z^n\vert a^n)=\prod_{i=1}^n\tilde w(z_i\vert a_i)
	\]
	with respect to the $n$-fold product measure $\mu\otimes\cdots\otimes\mu$, where $a^n=(a_1,\ldots,a_n)$ and $z^n=(z_1,\ldots,z_n)$. 
\end{ex}

\begin{ex}
    The conditional probability $P_{Y\vert X}$ of a random variable $Y$ with respect to the random variable $X$ defines an ordinary channel for $P_X$-almost every $x$.
\end{ex}

\begin{ex}
    Any deterministic function $\phi:\mc X\to\mc A$ from a finite set $\mc X$ into an arbitrary set $\mc A$ can be regarded as an (ordinary) channel.
\end{ex}

Assume that $\tilde V:\mc X\to\mc Y$ is a subnormalized channel such that $\tilde V(\cdot\vert x)$ is a finite-support subnormalized measure (see Example \ref{ex:finite-support}). Let $\tilde W:\mc Y\to\mc Z$ be an arbitrary subnormalized channel with $\mu$-density $\tilde w$. The concatenation of $\tilde V$ with $\tilde W$ is the channel\footnote{The notation $\tilde V\tilde W$ for the concatenation of $\tilde V$ and $\tilde W$ can be justified by identifying a discrete subnormalized channel $\tilde W$ with its density matrix $\tilde w$. The convention is that the rows of $\tilde w$ are indexed by the input alphabet and contain the subnormalized output measures (in other words, every row has the form $\tilde w(\cdot\vert y)$ for some $y\in\mc Y$). Then $\tilde u$ equals the matrix product $\tilde v\tilde w$.} $\tilde V\tilde W:\mc X\to\mc Z$ defined by its $\mu$-density 
\[
    \tilde u(z\vert x)=\sum_{y\in\mc Y}\tilde w(z\vert y)\tilde v(y\vert x).
\]

As a special case of the concatenation of subnormalized channels, let $\tilde W:\mc X\to\mc Z$ be a subnormalized channel with $\mu$-density $\tilde w$ and $P$ a finite-support probability distribution on $\mc X$ with density $p$ with respect to the counting measure $\nu$ on $\supp(P)$. We define a subnormalized measure $\tilde W\otimes P$ on $\mc Z\times\mc X$ by its density 
\[
    q(z,x)=\tilde w(z\vert x)p(x)
\]
with respect to $\mu\otimes\nu$. The $\mc Z$-marginal of $\tilde W\otimes P$ is denoted by $P\tilde W$ and has the $\mu$-density
\[
    r(z)=\sum_{x\in\supp(P)}p(x)\tilde w(z\vert x).
\]

If $(X,Y)$ is a pair of random variables on $\mc X\times\mc Y$ and there exists a channel $W$ such that $P_{Y\vert X}=W$ $P_X$-almost surely, then we say that $Y$ is \textit{generated by $X$ via $W$}. If $P_X=P$, then we often write $I(X\wedge Y)=I(P,W)$. Also, if $\tilde W:\mc S\to\mc Z$ is a subnormalized channel with finite input alphabet $\mc S$, $M$ a subnormalized measure on $\mc Z$ and $P$ a probability measure on $\mc S$ with density $p$, then we set
\[
    D(\tilde W\Vert M\vert P)=\sum_{s\in\mc S}p(s)D\bigl(\tilde W(\cdot\vert s)\Vert M\bigr)
\]
and
\[
    D_2(\tilde W\Vert M\vert P)
    =\log\sum_{s\in\mc S}p(s)\exp\bigl(D_2(\tilde W(\cdot\vert s)\Vert M)\bigr).
\]
These are the \textit{conditional Kullback-Leibler divergence} and the \textit{conditional R\'enyi 2-divergence of $\tilde W$ and $M$ with respect to $P$}, respectively\footnote{This is one of several possible definitions of conditional R\'enyi 2-divergence}. Note that
\[
    D_i(\tilde W\Vert M\vert P)=D_i(\tilde W\otimes P\Vert M\otimes P),
\]
where $D_i$ can be either $D$ or $D_2$.

It is well-known that Kullback-Leibler divergence is upper-bounded by R\'enyi $2$-divergence for probability measures \cite{vEHRenyiDiv}. This can be generalized to the case of subnormalized measures. If $M_1$ and $M_2$ are subnormalized measures on the measurable set $\mc Z$, define $Z_i=M_i(\mc Z)$ for $i=1,2$. It is straightforward to check that
\[
    D(M_1\Vert M_2)=Z_1\left(D\left(\left.\frac{M_1}{Z_1}\right\Vert\frac{M_2}{Z_2}\right)+\log\frac{Z_1}{Z_2}\right),
\]
and
\[
    D_2(M_1\Vert M_2)=D_2\left(\left.\frac{M_1}{Z_1}\right\rVert\frac{M_2}{Z_2}\right)+2\log Z_1-\log Z_2.
\]

\begin{lem}\label{lem:nonnorm_div_ineq}
	Let $\mc Z$ be a measurable space with measure $\mu$ and let $M_1,M_2$ be subnormalized measures on $\mc Z$, both with densities with respect to $\mu$. Then 
	\[
		D(M_1\Vert M_2)\leq Z_1\left(D_2(M_1\Vert M_2)-\log Z_1\right).
	\]
\end{lem}

\begin{proof}
    See Appendix \ref{app:proofs}.
\end{proof}

We will actually make use of the following consequence of Lemma \ref{lem:nonnorm_div_ineq}.

\begin{lem}\label{lem:RenyitoKL}
    Let $\mc Z$ be a measurable space with measure $\mu$ and $P$ a probability distribution on the finite set $\mc S$ with density $p$. Let $\tilde W:\mc S\to\mc Z$ be a subnormalized channel and let $M$ be a generalized measure on $\mc Z$. Assume there exists a $0<\varepsilon<1-e^{-1}$ such that $1-\varepsilon\leq \tilde W(\mc Z\vert s)\leq 1$ for all $s\in\mc S$. Then 
	\[
	    D(\tilde W\Vert M\vert P)\leq D_2(\tilde W\Vert M\vert P)-(1-\varepsilon)\log(1-\varepsilon).
	\]
\end{lem}

\begin{proof}
For any $s\in\mc S$ set $Z_s=\tilde W(\mc Z\vert s)$. Then
\begin{align*}
	D(\tilde W\Vert M\vert P)
	&\stackrel{(a)}{\leq}\log\left(\sum_sp(s)\exp\bigl(D(\tilde W(\cdot\vert s)\Vert M)\bigr)\right)\\
	&\stackrel{(b)}{\leq}\log\left(\sum_sp(s)\exp\bigl(Z_sD_2(\tilde W(\cdot\vert s)\Vert M)-Z_s\log Z_s\bigr)\right)\\
	&\stackrel{(c)}{\leq}\log\left(\sum_sp(s)\exp\bigl(D_2(\tilde W(\cdot\vert s)\Vert M)\bigr)(1-\varepsilon)^{-(1-\varepsilon)}\right)\\
 	&=D_2(\tilde W\Vert M\vert P)-(1-\varepsilon)\log(1-\varepsilon),
\end{align*}
where $(a)$ is due to the convexity of the exponential function, $(b)$ is a consequence of Lemma \ref{lem:nonnorm_div_ineq} and $(c)$ follows from $1-\varepsilon\leq Z_s\leq 1$ and the fact that the function $t\mapsto -t\log t$ decreases between $e^{-1}$ and $1$. 
\end{proof}

If $w$ is the density of the ordinary channel $W:\mc X\to\mc Z$ with respect to $\mu$, then for any subset $\mc T$ of $\mc X\times\mc Z$ such that $\{z:(x,z)\in\mc T\}$ is measurable,
\begin{equation}\label{eq:subnch-def}
	w_{\mc T}(z\vert x)=w(z\vert x)1_{\{(x,z)\in\mc T\}}
\end{equation}
defines the $\mu$-density of a subnormalized channel $W_{\mc T}:\mc X\rightarrow\mc Z$. For any probability distribution $P$ on $\mc X$ with finite support, the \textit{$\varepsilon$-smooth conditional R\'enyi 2-divergence of $W$ with respect to $P$} is defined as 
\[
    D_2^\varepsilon(W\Vert PW\vert P)=\inf_{\mc T}D_2(W_{\mc T}\Vert PW_{\mc T}\vert P),
\]
where $W_{\mc T}$ is defined as in \eqref{eq:subnch-def} and $\mc T$ ranges over all measurable subsets of $\mc X\times\mc Z$ satisfying 
\begin{equation}\label{eq:subnsets}
	W(\{z:(x, z)\in\mc T\}\vert x)\geq 1-\varepsilon
\end{equation}
for all $x\in\mc X$.

\section{One-Shot Wiretap Channels}\label{sect:oneshotwiretap}

\subsection{One-shot wiretap channels and codes}

A \textit{one-shot wiretap channel} is a pair of channels $(T:\mc A\to\mc Y,U:\mc A\to\mc Z)$. $T$ is the channel from the sender to the intended message recipient, $U$ is the channel from the sender to the eavesdropper. We will often denote a one-shot wiretap channel simply by $(T,U)$. 

A \textit{seeded wiretap code} for the one-shot wiretap channel $(T,U)$ consists of a \textit{seed set} $\mc S$, a \textit{message set} $\mc M$, an \textit{encoder} $\xi:\mc S\times\mc M\to\mc A$ and a \textit{decoder} $\zeta:\mc S\times\mc Y\to\mc M$. The encoder is a channel, whereas the decoder will always be an ordinary mapping. A seeded wiretap code will be denoted by $(\xi,\zeta)$. We call a seeded wiretap code whose seed set contains a single element an \textit{ordinary wiretap code}. The \textit{(maximal) error probability} of $(\xi,\zeta)$ is defined as
\[
    e(\xi,\zeta)=\max_{s\in\mc S}\max_{m\in\mc M}\,(\xi T)\bigl(\{y:\zeta(s,y)\neq m\}\vert s,m\bigr)
\]
(recall that $\xi T$ denotes the concatenation of the channels $\xi$ and $T$.) If $e(\xi,\zeta)$ is small, then the transmission of messages through $T$ applying the seeded wiretap code $(\xi,\zeta)$ is close to noiseless provided that the sender and the receiver use the same seed $s$.

At the same time, an eavesdropper observing the output of the channel $\xi U$ should learn as little as possible about the message $m$. If $S$ is uniformly distributed on $\mc S$, then we define the \textit{semantic security information leakage}
\[
    L_\sem(\xi,\zeta)=\max_{P_M}I(M\wedge Z,S),
\]
where the maximum ranges over all possible probability distributions on $\mc M$, the random variable $M$ is distributed according to $P_M$ and independent of $S$, and $Z$ is generated by $S$ and $M$ via $\xi U$. The smaller $L_\sem(\xi,\zeta)$ is, the less information does the eavesdropper obtain about the messages sent through $\xi U$.

For later comparison, we also mention the concept of \textit{strong secrecy.} Specifically, define
\[
    L_\str(\xi,\zeta)=I(\overline M\wedge\overline Z,S)
\]
where $\overline M$ is uniformly distributed on $\mc M$ and independent of $S$, and $\overline Z$ is generated by $S$ and $\overline M$ via $\xi U$. $L_\str(\xi,\zeta)$ is called the \textit{strong secrecy information leakage} of $(\xi,\zeta)$. Thus instead of considering the worst case among all message distributions as in the semantic security information leakage, strong secrecy assumes that the message is uniformly distributed over the message set. It clearly holds that
\begin{align*}
    L_\str(\xi,\zeta)&\leq L_\sem(\xi,\zeta).
\end{align*}
Naturally, it is an interesting question how big the difference between these concepts is. This is discussed in Appendix \ref{app:strongsec}.

Seeded wiretap codes are the general framework for seeded modular coding schemes with a security component, like the modular UHF and BRI schemes described in the introduction. Security is measured under the assumption that the seed $s$ is chosen uniformly from the seed set $\mc S$. It is not required that $s$ be unknown to the eavesdropper. In fact, the definitions of both security leakages add the random seed $S$ to the eavesdropper's knowledge. 

The concept of seeded wiretap codes does not specify how the seed is generated. One possibility is that an additional resource called ``common randomness'' generates $S$ and noiselessly transmits the realization $s$ to the sender and the intended receiver. An alternative for the case where multiple transmissions are possible is that the sender uniformly at random generates a seed $s$ and transmits this to the intended receiver. Since the seed may be known to the eavesdropper, this can be done without taking security into account. In a second step, the confidential message can be transmitted using the seeded wiretap code. In Section \ref{sect:asymptotics}, we will see that this method can be modified in such a way that one obtains an ordinary wiretap code which loses no rate compared to the original seeded wiretap code.

\section{Biregular Irreducible Functions}\label{sect:BRI}

In this section, we introduce biregular irreducible functions as a new type of security components for seeded modular coding schemes as described in the introduction. We also formally define modular BRI schemes and formulate the central result of this paper, which is an upper bound on the semantic security information leakage incurred by modular BRI schemes. The bound will be used to derive coding results for memoryless wiretap channels in Section \ref{sect:asymptotics}. An additional result of this section which is just stated for comparison is that biregular irreducible functions are universal hash functions on average.

\subsection{Biregular irreducible functions}

\begin{defn}\label{defn:BRI_function}
	A \textit{biregular irreducible function} is a function $f:\mc S\times\mc X\rightarrow\mc N$, where $\mc S,\mc X,\mc N$ are finite sets, for which there exists a subset $\mc M$ of $\mc N$ and two positive integers $d_{\mc S},d_{\mc X}$ such that for every $m\in\mc M$
	\begin{enumerate}
		\item $\mc S$\textit{-regularity:} $\lvert\{x:f(s,x)=m\}\rvert=d_{\mc S}$ for every $s\in\mc S$,
		\item $\mc X$\textit{-regularity:} $\lvert\{s:f(s,x)=m\}\rvert=d_{\mc X}$ for every $x\in\mc X$,
		\item \textit{Irreducibility:} the stochastic matrix $P_{f,m}$ on $\mc X\times\mc X$ defined by
		\begin{equation}\label{eq:Pfm}
			P_{f,m}(x,x')=\frac{\lvert\{s:f(s,x)=f(s,x')=m\}\rvert}{d_{\mc S}d_{\mc X}}
		\end{equation}
		has second-largest eigenvalue modulus $\lambda_2(f,m)<1$
		(that $P_{f,m}$ really is a stochastic matrix follows from conditions 1) and 2), see Lemma \ref{lem:reallystoch}). 
	\end{enumerate}
 	$\mc M$ is called the \textit{regularity set} and $\log(\lvert\mc M\rvert)/\log(\lvert\mc X\rvert)$ the \textit{rate} of $f$.
\end{defn}

We will always assume that $f(\mc S\times\mc X)=\mc N$. To prove that biregular irreducible functions are well-defined, we note the following lemma. 

\begin{lem}\label{lem:reallystoch}
    For any biregular irreducible function $f:\mc S\times\mc X\rightarrow\mc N$ with regularity set $\mc M$, the matrix $P_{f,m}$ as defined in \eqref{eq:Pfm} is a stochastic matrix for every $m\in\mc M$.
\end{lem}

\begin{proof}
    \begin{align*}
        \sum_{x'\in\mc X}\lvert\{s:f(s,x)=f(s,x')=m\}\rvert
        &=\sum_{s\in\mc S}1_{\{f(s,x)=m\}}\sum_{x'\in\mc X}1_{\{f(s,x')=m\}}
        =d_{\mc S}d_{\mc X}.
    \end{align*}
    Thus every row sum of $P_{f,m}$ equals $1$.
\end{proof}

We also note that by a simple and well-known double-counting argument, for any $m\in\mc M$,
\begin{align}\label{eq:doublecounting}
    d_{\mc X}\lvert\mc X\rvert
    &=\sum_{x\in\mc X}\lvert\{s:f(s,x)=m\}\rvert\notag\\
    &=\sum_{x\in\mc X}\sum_{s\in\mc S}1_{\{f(s,x)=m\}}\notag\\
    &=\sum_{s\in\mc S}\lvert\{x:f(s,x)=m\}\rvert\notag\\
    &=d_{\mc S}\lvert\mc S\rvert.
\end{align}

When a biregular irreducible function $f:\mc S\times\mc X\to\mc N$ with regularity set $\mc M$ is applied in wiretap coding, the message set will be given by $\mc M$. Observe that
\begin{equation}\label{eq:Mcard-ub}
    \lvert\mc M\rvert\leq\frac{\lvert\mc X\rvert}{d_{\mc S}}=\frac{\lvert\mc S\rvert}{d_{\mc X}},
\end{equation}
with equality if and only if $\mc M=\mc N$. This implies the upper bound
\[
    \frac{\log\lvert\mc M\rvert}{\log\lvert\mc X\rvert}
    \leq1-\frac{\log d_{\mc S}}{\log\lvert\mc X\rvert}
\]
for the rate of $f$. Inequality \eqref{eq:Mcard-ub} also implies 
\begin{equation}\label{eq:seedsize_lb}
    \lvert\mc S\rvert\geq d_{\mc X}\lvert\mc M\rvert,
\end{equation}
in particular, the seed of a biregular irreducible function has to be at least as long as the message.

For any fixed seed s, a biregular irreducible function $f:\mc S\times\mc X\to\mc N$ with regularity set $\mc M$ is not invertible in general. Its \textit{randomized inverse} is the channel $f_s^{-1}:\mc M\to\mc X$ defined by
\[
    f_s^{-1}(x\vert m)=\frac{1}{d_{\mc S}}1_{\{f(s,x)=m\}}
\]
(we introduce no special notation for its density). A similar channel is given for every fixed $m$. It is denoted by $Q_{f,m}:\mc S\to\mc X$ and defined by its density
\[
    q_{f,m}(x\vert s)=\frac{1}{d_{\mc S}}1_{\{f(s,x)=m\}}.
\]
Thus $Q_{f,m}(\cdot\vert s)=f_s^{-1}(\cdot\vert m)$. It satisfies
\begin{equation}\label{eq:unif_seed}
    P_{\mc S}Q_{f,m}=P_{\mc X},
\end{equation}
because
\begin{align*}
    \frac{1}{\lvert\mc S\rvert}\sum_{s\in\mc S}q_{f,m}(x\vert s)
    =\frac{\lvert\{s:f(s,x)=m\}\rvert}{d_{\mc S}\lvert\mc S\rvert}
    =\frac{d_{\mc X}}{d_{\mc S}\lvert\mc S\rvert}
    =\frac{1}{\lvert\mc X\rvert},
\end{align*}
where the last equality is due to \eqref{eq:doublecounting}.

\subsection{Modular BRI schemes}

We now formally define \textit{modular BRI schemes}. They are a special case of seeded wiretap codes.

Let $(T:\mc A\to\mc Y,U:\mc A\to\mc Z)$ be a one-shot wiretap channel. An \textit{error-correcting code}\footnote{We use the term \textit{error-correcting code} to emphasize the difference to wiretap codes. This difference consists in the fact that error-correcting codes do not use a seed and that we do not measure their security leakage. Our use of the term implies that the code alphabet is equal to the channel input alphabet. Steps which often are not considered to be part of an error-correcting code, like channel modulation, here are assumed to be part of the code.}  for $T$ is a pair $(\phi,\psi)$ such that
\begin{enumerate}
    \item $\phi:\mc X\to\mc A$ is a channel, where $\mc X$ is a finite set called the \textit{message set} of $(\phi,\psi)$,
    \item $\psi:\mc Y\to\mc X$ is an ordinary mapping.
\end{enumerate}
Its \textit{(maximal) error probability} $e(\phi,\psi)$ is defined as
\[
    e(\phi,\psi)=\max_x(\phi T\psi)(\mc X\setminus\{x\}\vert x).
\]
We need to allow $\phi$ to be a general channel (instead of an ordinary mapping) in order to be able to construct  modular seeded wiretap codes which achieve the capacity of arbitrary non-degraded discrete wiretap channels in the asymptotic analysis of Section \ref{sect:asymptotics}.

Now let $(\phi,\psi)$ be an error-correcting code for $T$ with message set $\mc X$ and let $f:\mc S\times\mc X\to\mc N$ be a biregular irreducible function with regularity set $\mc M$. Together, they determine a seeded wiretap code $(\xi,\zeta)$, where
\begin{enumerate}
    \item the channel $\xi:\mc S\times\mc M\to\mc A$ is defined by $\xi(a\vert s,m)=(f_s^{-1}\phi)(a\vert m)$ (equivalently $\xi(a\vert s,m)=(Q_{f,m}\phi)(a\vert s)$) and 
    \item the ordinary mapping $\zeta:\mc S\times\mc Y\to\mc M$ is defined by $\zeta(s,y)=f(s,\psi(y))$.
\end{enumerate}
We call $(\xi,\zeta)$ a \textit{modular BRI scheme} and denote it by $\Pi(f,\phi,\psi)$. Modular BRI schemes are a formalization of the seeded modular coding scheme depicted in Fig. \ref{fig:BT_scheme}, where the security component is a biregular irreducible function. For the error probability and the semantic security leakage of $\Pi(f,\phi,\psi)$, we introduce the notation
\begin{align*}
    e(\Pi(f,\phi,\psi))&=e(\xi,\zeta),\qquad
    L_\sem(\Pi(f,\phi,\psi))=L_\sem(\xi,\zeta).
\end{align*}
Clearly,
\begin{equation}\label{eq:maxerr_prefix}
    e(\Pi(f,\phi,\psi))\leq e(\phi,\psi).
\end{equation}

\subsection{Security by biregular irreducible functions}

Consider a modular BRI scheme $\Pi(f,\phi,\psi)$ for a biregular irreducible function $f:\mc S\times\mc X\to\mc N$ with regularity set $\mc M$. Set 
\begin{equation}\label{eq:U-phi-channel}
	W=\phi U:\mc X\longrightarrow\mc Z.
\end{equation}
In order to upper-bound the semantic security information leakage of $\Pi(f,\phi,\psi)$, we will upper-bound $D(Q_{f,m}W\Vert P_{\mc X}W\vert P_{\mc S})$ for every individual message $m$ (recall that $P_{\mc X}$ and $P_{\mc S}$ are the uniform distributions on $\mc X$ and $\mc S$, respectively). That $W$ has a structure like in \eqref{eq:U-phi-channel} is inessential for this result. The bound and its proof are inspired by the channel leftover hash lemma of Tyagi and Vardy \cite{TV_UHF_preprint}. The bound on $L_\sem(\Pi(f,\phi,\psi))$ follows from this per-message statement.

\begin{thm}\label{thm:EV-UB}
	Let $W:\mc X\to\mc Z$ be any channel and let $f:\mc S\times\mc X\to\mc N$ be a biregular irreducible function with regularity set $\mc M\subset\mc N$. Then for every $m\in\mc M$ and $0<\varepsilon<1-e^{-1}$,
	\begin{align*}
		&D\bigl(Q_{f,m}W\Vert P_{\mc X}W\vert P_{\mc S}\bigr)\\
		&\leq \frac{1}{\ln 2}\lambda_2(f,m)2^{D_2^{\varepsilon}(W\Vert P_{\mc X}W\vert P_{\mc X})}
		+\varepsilon\log\frac{\lvert\mc X\rvert}{d_{\mc S}}-(1-\varepsilon)\log(1-\varepsilon).
	\end{align*}
\end{thm}

\begin{proof}
    See Section \ref{sect:proof_EV-UB}.
\end{proof}

Denote the upper bound given in Theorem \ref{thm:EV-UB} by $\eta(f,m,W)$. We then have the following corollary.

\begin{cor}\label{cor:sec_by_BRI}
    Let $\Pi(f,\phi,\psi)$ be a modular BRI scheme for the one-shot wiretap channel $(T,U)$. Define $W$ as in \eqref{eq:U-phi-channel}. Then
    \begin{align*}
        L_\sem(\Pi(f,\phi,\psi))
        &\leq\max_{m\in\mc M}\eta(f,m,W).
    \end{align*}
\end{cor}

\begin{proof}
    Let $\Pi(f,\phi,\psi)$ be a modular BRI scheme with message set $\mc M$ and let $M$ be an arbitrary random variable $M$ on $\mc M$. Assume that $Z$ is the eavesdropper's output generated by $M$ and $S$ via the channel $W$. We then have
\begin{align*}
    I(M,S\wedge Z)
    =D(P_{Z\vert S,M}\Vert P_Z\vert P_{\mc S}\otimes P_M)
    \leq\max_{m\in\mc M}D(P_{Z\vert S,M=m}\Vert P_Z\vert P_{\mc S})
    \leq\max_{m\in\mc M}\eta(f,m,W),
\end{align*}
where the last inequality is due to \eqref{eq:unif_seed} and to Theorem \ref{thm:EV-UB}. In a second step, we observe that $I(M\wedge S,Z)\leq I(M,S\wedge Z)$, which is due to the following elementary calculation:
\begin{align*}
    &I(M\wedge S,Z)-I(M,S\wedge Z)\\
    &=H(M)+H(S,Z)-H(M,S,Z)-H(M,S)-H(Z)+H(M,S,Z)\\
    &=H(S,Z)-H(S)-H(Z)\\
    &=-I(S\wedge Z)\\
    &\leq 0,
\end{align*}
where we used the independence of $S$ and $M$ in the middle equality. Therefore $I(M\wedge S,Z)\leq\max_{m\in\mc M}\eta(f,m,W)$. Since $M$ was chosen arbitrarily, this completes the proof of the corollary.
\end{proof}

The main term of the upper bound of Theorem \ref{thm:EV-UB} clearly separates the effect of the biregular irreducible function from that of the channel. If $W_n:\mc X_n\to\mc Z_n$ is a sequence of channels and $f_n:\mc S_n\times\mc X_n\to\mc N_n$ a sequence of biregular irreducible functions with regularity sets $\mc M_n$ such that
\[
    \lim_{n\to\infty}\max_{m\in\mc M_n}\lambda_2(f_n,m)2^{D_2^\varepsilon(Q_{f_n,m}W_n\Vert P_{\mc X_n}W_n\vert P_{\mc S_n})}=0,
\]
then (ignoring the other terms for now) Theorem \ref{thm:EV-UB} and Corollary \ref{cor:sec_by_BRI} together imply that perfect semantic security is achieved asymptotically. This will be used in Section \ref{sect:asymptotics} to construct secrecy capacity-achieving modular BRI schemes for discrete and Gaussian memoryless wiretap channels. Thus the separation of error correction, which here is hidden in the channels $W_n$, and the generation of semantic security is optimal. This is analogous to the source-channel separation theorem for memoryless channels.

Hayashi and Matsumoto \cite{HayMat_Multiplex} prove a result similar to Corollary \ref{cor:sec_by_BRI} without noting that it can be used to establish semantic security directly. We have more to say about this in Appendix \ref{app:HayMat}.

\begin{rem}
    To our knowledge, the $\varepsilon$-smooth R\'enyi divergence has not been defined before. In Section \ref{sect:assec_proof}, we will upper-bound it for memoryless channels using the $\varepsilon$-smooth max-information of a channel as defined by Tyagi and Vardy \cite{TV_UHF_preprint}. There also exist several definitions and studies of the $\varepsilon$-smooth R\'enyi entropy. It goes back to Renner and Wolf \cite{RW_smoothRenyi}. It was used in the context of information reconciliation and privacy amplification by Renner and Wolf \cite{RW_info_privacy}. Hayashi used it to study the privacy amplification properties of $\varepsilon$-almost dual universal hash functions, a generalization of universal hash functions \cite{Hay_smooth}.
\end{rem}

\subsection{Biregular irreducible functions and universal hash functions}

To conclude this section, we examine the relation between biregular irreducible functions and universal hash functions. A universal hash function is a function $f:\mc S\times\mc X\to\mc N$ which satisfies
\begin{equation}\label{eq:UHF_def}
    \mbb P[f(S,x)=f(S,x')]\leq\frac{1}{\lvert\mc N\rvert}
\end{equation}
if $S$ is uniformly distributed on the seed set $\mc S$ and $x\neq x'$. A natural question is whether a biregular irreducible function is a universal hash function under the condition that the common value of $f(s,x)$ and $f(s,x')$ is an element of $\mc M$ and that the right-hand side of \eqref{eq:UHF_def} is replaced by $1/\lvert\mc M\rvert$. One obtains the following average result, which is not needed in this paper.

\begin{lem}\label{lem:UHF_single}
    If $f:\mc S\times\mc X\to\mc N$ is a biregular irreducible function with regularity set $\mc M$, then
    \begin{align*}
        \frac{1}{\lvert\mc X\rvert-1}\sum_{x'\neq x}\mbb P[f(S,x)=f(S,x')\vert f(S,x)\in\mc M]
        =\frac{d_{\mc S}-1}{\lvert\mc X\rvert-1}\leq\frac{1}{\lvert\mc M\rvert}.
    \end{align*}
\end{lem}

\begin{proof}
    Observe that
    \begin{equation}\label{eq:in_qualif_set}
        \mbb P[f(S,x)\in\mc M]=\sum_{m\in\mc M}\frac{\lvert\{s:f(s,x)=m\}\rvert}{\lvert\mc S\rvert}=\frac{\lvert\mc M\rvert d_{\mc X}}{\lvert\mc S\rvert}.
    \end{equation}
    Therefore
    \begin{align*}
        &\frac{1}{\lvert\mc X\rvert-1}\sum_{x'\neq x}\mbb P[f(S,x)=f(S,x')\vert f(S,x)\in\mc M]\\
        &=\frac{1}{\lvert\mc X\rvert-1}\sum_{x'\neq x}\frac{\mbb P[f(S,x)=f(S,x')\in\mc M]}{\mbb P[f(S,x)\in\mc M]}\\
        &\stackrel{(a)}{=}\frac{1}{\lvert\mc X\rvert-1}\sum_{m\in\mc M}\frac{\sum_{x'\neq x}\lvert\{s:f(s,x)=f(s,x')=m\}\rvert}{\lvert\mc M\rvert d_{\mc X}}\\
        &\stackrel{(b)}{=}\frac{(d_{\mc S}-1)d_{\mc X}}{(\lvert\mc X\rvert-1)d_{\mc X}}\\
		&\stackrel{(c)}{\leq}\frac{d_{\mc S}-1}{d_{\mc S}\lvert\mc M\rvert-1}\\
		&\leq\frac{1}{\lvert\mc M\rvert},
    \end{align*}
	where $(a)$ is due to \eqref{eq:in_qualif_set}, $(b)$ follows from the proof of Lemma \ref{lem:reallystoch} and $(c)$ is due to \eqref{eq:Mcard-ub}.
\end{proof}

\section{Biregular Irreducible Functions and Graphs}\label{sect:Ramanujan}

\subsection{Characterization of biregular irreducible functions}

Note that some basic graph-theoretic terms, like the adjacency matrix of a graph, are defined in Appendix \ref{app:graphs}. Additionally, we call a graph $G$ \textit{bipartite} if its vertex set is the union of two disjoint sets $\mc S$ and $\mc X$ such that every edge in $G$ has one vertex in $\mc S$ and one in $\mc X$. The pair $(\mc S,\mc X)$ is called a \textit{bipartition} of $G$. A bipartite graph $G$ with bipartition $(\mc S,\mc X)$ is called $(d_{\mc S},d_{\mc X})$-\textit{biregular}\footnote{Note that sometimes biregular graphs are defined without having to be bipartite.} if every element of $\mc S$ has degree $d_{\mc S}$ and every element of $\mc X$ has degree $d_{\mc X}$. If $d_{\mc S}=d_{\mc X}=d$, then the graph is bipartite and $d$-regular.

The complete bipartite graph $\mc K_{\mc S,\mc X}$ with bipartition $(\mc S,\mc X)$ is the graph on $\mc S\cup\mc X$ where every element of $\mc S$ is adjacent to every element of $\mc X$. Clearly $\mc K_{\mc S,\mc X}$ is $(\lvert\mc X\rvert,\lvert\mc S\rvert)$-biregular. Every function $f:\mc S\times\mc X\to\mc N$ is equivalent to a decomposition $(G_m)_{m\in\mc N}$ of $\mc K_{\mc S,\mc X}$ into edge-disjoint subgraphs, where two vertices $s\in\mc S$ and $x\in\mc X$ are adjacent in $G_m$ if and only if $f(s,x)=m$. We say that $f$ is \textit{defined by the family} $(G_m)_{m\in\mc N}$.

\begin{thm}\label{thm:BRI-characterization}
    A function $f:\mc S\times\mc X\to\mc N$ is a biregular irreducible function with regularity set $\mc M\subset\mc N$ if and only if it is defined by a decomposition $(G_{f,m})_{m\in\mc N}$ of the complete bipartite graph $\mc K_{\mc S,\mc X}$ into edge-disjoint subgraphs such that $G_{f,m}$ is $(d_{\mc S},d_{\mc X})$-biregular and connected\footnote{See Appendix \ref{app:graphs} for the definition of connectedness.} for every $m\in\mc M$. In this case, if $\lambda_2(G_{f,m})$ is the second-largest eigenvalue of $G_{f,m}$, then
    \[
        \lambda_2(f,m)=\frac{\lambda_2(G_{f,m})^2}{d_{\mc S}d_{\mc X}}<1
    \]    
    for every $m\in\mc M$.
\end{thm}

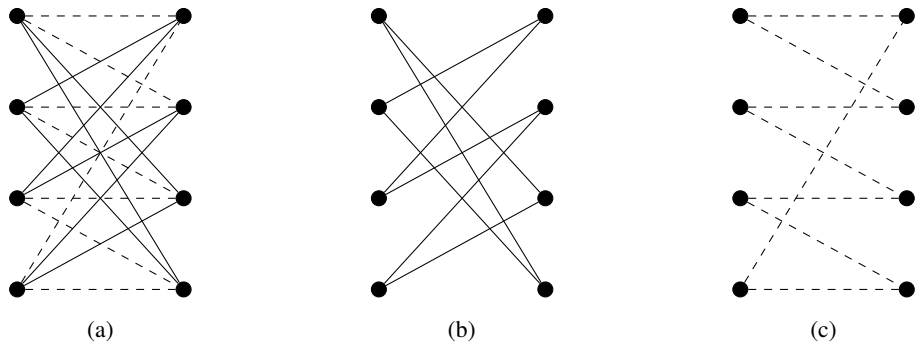
\begin{figure}
    \centering
    \subfloat[]{
    \begin{tikzpicture}[vertex/.style={draw, circle, fill, inner sep=2pt}, edge/.style={-}, edge2/.style={dashed}
    ]
        \node[vertex] (00) {};
        \node[vertex, below = 1cm of 00] (01) {};
        \node[vertex, below = 1cm of 01] (02) {};
        \node[vertex, below = 1cm of 02] (03) {};
        \node[vertex, right = 2cm of 00] (10) {};
        \node[vertex, right = 2cm of 01] (11) {};
        \node[vertex, right = 2cm of 02] (12) {};
        \node[vertex, right = 2cm of 03] (13) {};
        
        \draw[edge2] (00) -> (10);
        \draw[edge2] (00) -> (11);
        \draw[edge] (00) -> (12);
        \draw[edge] (00) -> (13);
        \draw[edge] (01) -> (10);
        \draw[edge2] (01) -> (11);
        \draw[edge2] (01) -> (12);
        \draw[edge] (01) -> (13);
        \draw[edge] (02) -> (10);
        \draw[edge] (02) -> (11);
        \draw[edge2] (02) -> (12);
        \draw[edge2] (02) -> (13);
        \draw[edge2] (03) -> (10);
        \draw[edge] (03) -> (11);
        \draw[edge] (03) -> (12);
        \draw[edge2] (03) -> (13);
    \end{tikzpicture}
    }
    \hfil
    \subfloat[]{
    \begin{tikzpicture}[vertex/.style={draw, circle, fill, inner sep=2pt}, edge/.style={-}, edge2/.style={--}
    ]
        \node[vertex] (00) {};
        \node[vertex, below = 1cm of 00] (01) {};
        \node[vertex, below = 1cm of 01] (02) {};
        \node[vertex, below = 1cm of 02] (03) {};
        \node[vertex, right = 2cm of 00] (10) {};
        \node[vertex, right = 2cm of 01] (11) {};
        \node[vertex, right = 2cm of 02] (12) {};
        \node[vertex, right = 2cm of 03] (13) {};
        
        \draw[edge] (00) -> (12);
        \draw[edge] (00) -> (13);
        \draw[edge] (01) -> (10);
        \draw[edge] (01) -> (13);
        \draw[edge] (02) -> (10);
        \draw[edge] (02) -> (11);
        \draw[edge] (03) -> (11);
        \draw[edge] (03) -> (12);
    \end{tikzpicture}
    }
    \hfil
    \subfloat[]{
    \begin{tikzpicture}[vertex/.style={draw, circle, fill, inner sep=2pt}, edge/.style={-}, edge2/.style={-, dashed}
    ]
        \node[vertex] (00) {};
        \node[vertex, below = 1cm of 00] (01) {};
        \node[vertex, below = 1cm of 01] (02) {};
        \node[vertex, below = 1cm of 02] (03) {};
        \node[vertex, right = 2cm of 00] (10) {};
        \node[vertex, right = 2cm of 01] (11) {};
        \node[vertex, right = 2cm of 02] (12) {};
        \node[vertex, right = 2cm of 03] (13) {};
        
        \draw[edge2] (00) -> (10);
        \draw[edge2] (00) -> (11);
        \draw[edge2] (01) -> (11);
        \draw[edge2] (01) -> (12);
        \draw[edge2] (02) -> (12);
        \draw[edge2] (02) -> (13);
        \draw[edge2] (03) -> (10);
        \draw[edge2] (03) -> (13);
    \end{tikzpicture}
    }
    \caption{(a) The complete bipartite graph with bipartition $(\mc S,\mc X)$ with $\lvert\mc S\rvert=\lvert\mc X\rvert=4$. Edges are partitioned into two classes, one dashed and one solid. (b) and (c) The connected bipartite biregular graphs whose edges are only from the dashed or only from the solid class.}\label{fig:compl_bipartite}
\end{figure}

\begin{ex}
    A biregular irreducible function $f:\mc S\times\mc X\to\{1,2\}$ with $\lvert\mc S\rvert=\lvert\mc X\rvert=4$ and regularity set $\{1,2\}$ is depicted in Fig.~\ref{fig:compl_bipartite}. The set $\{1,2\}$ represents the partition of the edge set of the complete bipartite graph with bipartition $(\mc S,\mc X)$ into the classes of dashed and solid edges. Both graphs $G_{f,1}$ and $G_{f,2}$ are isomorphic to the cycle on 8 vertices. They are regular of degree $d=2$ and $\lambda_2(G_{f,1})=\lambda_2(G_{f,2})=\sqrt{2}=\sqrt{d}$. (This is not hard to see. It can also be found in, e.g., \cite[1.4.3]{BrouwerHaemers_SpectraofGraphs}.)
\end{ex}

\begin{proof}[Proof of Theorem \ref{thm:BRI-characterization}]
    For simplicity of notation, we write $G_m$ instead of $G_{f,m}$. It is easy to see that the $\mc S$- and $\mc X$-regularity of $f$ for every $m\in\mc M$ is equivalent to the $(d_{\mc S},d_{\mc X})$-biregularity of $G_m$. We can therefore concentrate on the equivalence of irreducibility of $f$ on $\mc M$ and the connectedness of the $G_m$. 
    
    For any $m\in\mc M$, let $A_m$ be the adjacency matrix of $G_m$. Since $G_m$ is bipartite, it has the form
    \begin{equation}\label{eq:bip_adj_matr}
        A_m=
        \begin{bmatrix}
            0 & B_m \\
            B_m^T & 0
        \end{bmatrix}
    \end{equation}
    for an $\mc S\times\mc X$ matrix $B_m$. The rows of $B_m$ are indexed by the seed set $\mc S$, the columns by $\mc X$, and the $(s,x)$ entry $B_m(s,x)$ of $B_m$ equals $1$ if $s$ is adjacent to $x$ in $G_m$ and $0$ otherwise. The square of $A_m$ equals
    \[
        A_m^2=
        \begin{bmatrix}
            B_mB_m^T & 0 \\
            0 & B_m^TB_m
        \end{bmatrix}.
    \]
    Clearly, every eigenvalue of $A_m^2$ also is an eigenvalue of both $B_mB_m^T$ and $B_m^TB_m$. Since $\rank(A_m^2)=\rank(B_mB_m^T)+\rank(B_m^TB_m)$, $A_m^2$ has the same eigenvalues as both $B_mB_m^T$ and $B_m^TB_m$. It is well-known that the eigenvalue multiplicities of $B_m^TB_m$ and $B_mB_m^T$ coincide. Therefore the multiplicity of an eigenvalue for $A_m^2$ equals twice the multiplicity of this eigenvalue for $B_m^TB_m$.
        
    The $(x,x')$ entry of $B_m^TB_m$ equals
    \begin{align*}
        (B_m^TB_m)(x,x')
        &=\sum_{s\in\mc S}B_m(s,x)B_m(s,x')\\
        &=\sum_{s\in\mc S}1_{\{f(s,x)=m\}}1_{\{f(s,x')=m\}}\\
        &=\lvert\{s\in\mc S:f(s,x)=f(s,x')=m\}\rvert.
    \end{align*}
    Thus $P_{f,m}=d_{\mc S}^{-1}d_{\mc X}^{-1}B_m^TB_m$, in particular, $P_{f,m}$ is positive semidefinite. That $\lambda_2(f,m)<1$ therefore is equivalent to $d_{\mc S}d_{\mc X}$ being a simple eigenvalue of $B_m^TB_m$, and consequently a double eigenvalue of $A_m^2$. This is the minimal possible multiplicity for this eigenvalue, since the adjacency matrix of a $(d_{\mc S},d_{\mc X})$-biregular matrix always has eigenvalues $\pm\sqrt{d_{\mc S}d_{\mc X}}$ (this follows from the Perron-Frobenius theorem \cite[Theorem 2.2.1]{BrouwerHaemers_SpectraofGraphs} using the fact that the eigenvalue $\sqrt{d_{\mc S}d_{\mc X}}$ has the positive eigenvector $w$ with $w(x)=1$ for $x\in\mc X$ and $w(s)=\sqrt{d_{\mc S}/d_{\mc X}}$ for $s\in\mc S$). That $\pm\sqrt{d_{\mc S}d_{\mc X}}$ being simple eigenvalues of $A_m$ is equivalent to $G_m$ being connected is well-known \cite[Proposition 1.3.6]{BrouwerHaemers_SpectraofGraphs}.

    Now assume that $f$ is a biregular irreducible function. Since the second-largest eigenvalue modulus of $B_m^TB_m$ equals the second-largest eigenvalue modulus of $A_m^2$ by the above considerations, the formula for $\lambda_2(f,m)$ follows immediately and also that $\lambda_2(f,m)$ is strictly smaller than 1.
\end{proof}

\subsection{Construction of biregular irreducible functions}\label{subsect:Ramanujan}

In order to construct modular BRI schemes with large message set and small semantic security information leakage, the goal now is to find biregular irreducible functions with small $\lambda_2(f,m)$ and large regularity set $\mc M$. We will be interested in biregular irreducible functions where $\lvert\mc M\rvert$ is a fractional power of $\lvert\mc X\rvert$, but roughly $\lvert\mc X\rvert\leq\lambda_2(f,m)\lvert\mc M\rvert$ for every $m\in\mc M$. (For the precise statement see Theorem \ref{thm:existence_opt_BRI}.)

As a hint to what can be expected from the graph-theoretic side, recall that a $d$-regular graph always has maximal eigenvalue $d$, and $d$ is the largest eigenvalue modulus. If a $d$-regular graph is bipartite, it also has eigenvalue $-d$. By the Alon-Boppana bound \cite{Nilli_AlonBoppana,LPS_Ramanujan}, for every $\varepsilon>0$ the second-largest eigenvalue modulus of every sufficiently large connected $d$-regular graph is at least $2\sqrt{d-1}-\varepsilon$ (with $d$ fixed). The analogous statement for $(d_{\mc S},d_{\mc X})$-biregular graphs was shown by Feng and Li \cite{FengLi}, namely, for every $\varepsilon>0$, the second-largest eigenvalue of every sufficiently large connected $(d_{\mc S},d_{\mc X})$-biregular graph is at least $\sqrt{d_{\mc S}-1}+\sqrt{d_{\mc X}-1}-\varepsilon$.

Ramanujan graphs are optimal with respect to the bounds of Alon-Boppana and Feng-Li, respectively. A $d$-regular graph $G$ with adjacency matrix $A$ is called a \textit{Ramanujan graph} if every eigenvalue $\mu$ of $A$ satisfies $\mu=\pm d$ or $\lvert\mu\rvert\leq 2\sqrt{d-1}$. A $(d_{\mc S},d_{\mc X})$\textit{-biregular Ramanujan graph} $G$ with adjacency matrix $A$ has the property that every eigenvalue $\mu$ of $A$ satisfies $\mu=\pm\sqrt{d_{\mc S}d_{\mc X}}$ or $\lvert\mu\rvert\leq\sqrt{d_{\mc S}-1}+\sqrt{d_{\mc X}-1}$. Ramanujan graphs were first constructed by Lubotzky, Phillips and Sarnak \cite{LPS_Ramanujan} and Margulis \cite{Margulis_Ramanujan}. Since then, other constructions have followed, see \cite{MSS_interl_fams_i} for hints to the literature. 

There exist biregular irreducible functions $f:\mc S\times\mc X\to\mc M$ defined by a graph family $(G_{f,m})_{m\in\mc M}$ such that $G_{f,m}$ is a $(d_{\mc S},d_{\mc X})$-biregular Ramanujan graph for every $m\in\mc M$. 

\begin{thm}\label{thm:RamGraphs_existence}
    For every pair $(d_{\mc S},d_{\mc X})$ with $d_{\mc S},d_{\mc X}\geq 3$, every positive integer $k$ and disjoint sets $\mc S$ and $\mc X$ satisfying $\lvert\mc S\rvert=2^kd_{\mc X}$ and $\lvert\mc X\rvert=2^kd_{\mc S}$, there exists a decomposition of $\mc K_{\mc S,\mc X}$ into $2^k$ edge-disjoint connected $(d_{\mc S},d_{\mc X})$-biregular Ramanujan graphs.
\end{thm}

\begin{proof}
    See Subsection \ref{subsect:thmRamGraphs}.
\end{proof}

\begin{cor}\label{cor:RamanujanBRIs}
    For every pair $(d_{\mc S},d_{\mc X})$ with $d_{\mc S},d_{\mc X}\geq 3$ and every positive integer $k$ there exists a biregular irreducible function $f:\mc S\times\mc X\rightarrow\mc M$ with regularity set $\mc M$ satisfying
    \begin{enumerate}
        \item $\lvert\mc S\rvert=2^kd_{\mc X}$ and $\lvert\mc X\rvert=2^kd_{\mc S}$ and $\lvert\mc M\rvert=2^k$,
        \item $\lambda_2(f,m)\leq (\sqrt{d_{\mc S}-1}+\sqrt{d_{\mc X}-1})^2/(d_{\mc S}d_{\mc X})$ for every $m\in\mc M$.
    \end{enumerate}
    Such a biregular irreducible function is called a \emph{Ramanujan biregular irreducible function}.
\end{cor}

\begin{proof}
    Let $(G_{f,m})_{m\in\mc M}$ be the family of $2^k$ edge-disjoint connected $(d_{\mc S},d_{\mc X})$-biregular bipartite Ramanujan graphs with bipartition $(\mc S,\mc X)$ constructed in Theorem \ref{thm:RamGraphs_existence}. Let $f:\mc S\times\mc X\to\mc M$ be defined by $(G_{f,m})_{m\in\mc M}$. $f$ is well-defined because the family $(G_{f,m})_{m\in\mc M}$ is an edge-disjoint decomposition of $K_{\mc S,\mc X}$.
\end{proof}

Note that the Ramanujan biregular irreducible functions constructed in Corollary \ref{cor:RamanujanBRIs} satisfy equality in \eqref{eq:Mcard-ub} and \eqref{eq:seedsize_lb}. By a suitable choice of the degrees and message sizes, it will be shown in Theorem \ref{thm:existence_opt_BRI} that there exist sequences of Ramanujan biregular irreducible functions which exhibit the desired relations between $\lvert\mc X\rvert,\lvert\mc M\rvert$ and $\lambda_2(f,m)$ mentioned before.

The divergence bound of Theorem \ref{thm:EV-UB} obtains the following form for a Ramanujan biregular irreducible function.

\begin{cor}\label{cor:RamanujanBRIsSemSec}
    For a Ramanujan biregular irreducible function $f:\mc S\times\mc X\rightarrow\mc M$ as in Corollary \ref{cor:RamanujanBRIs}, 
    \begin{align*}
		&D(Q_{f,m}W\Vert P_{\mc X}W\vert P_{\mc S})\\
		&\leq\frac{(\sqrt{d_{\mc S}-1}+\sqrt{d_{\mc X}-1})^2}{d_{\mc S}d_{\mc X}\ln 2}\, 2^{D_2^\varepsilon(W\Vert P_{\mc X}W\vert P_{\mc X})}
		+\varepsilon k-(1-\varepsilon)\log(1-\varepsilon).
	\end{align*}
\end{cor}

\section{The Seeded Coset Biregular Irreducible Function}\label{sect:BT-BRI}

In this section, we analyze a universal hash function frequently used in the literature (references will be given after the definition) and transform it into a biregular irreducible function for suitable parameters. The challenge here is to identify a regularity set which is large and at the same time gives a second-largest eigenvalue which is close to the optimum according to the Feng-Li bound. It will turn out that the parameters cannot be chosen as flexibly as those of Ramanujan biregular irreducible functions.

Let $\mc X=\mc S=\mbb F_{2^\ell}^*$, the multiplicative group of the finite field with $2^\ell$ elements. $\mbb F_{2^\ell}$ is an $\ell$-dimensional vector space over $\mbb F_2$. Let $\mc V$ and $\mc N$ be linear subspaces of this vector space with $\dim\mc V=b$ and $\dim\mc N=k=\ell-b$ such that $\mc V+\mc N=\mbb F_{2^\ell}$. For $s,x\in\mbb F_{2^\ell}^*$ we define
\[
    \beta(s,x)=m\quad\text{if}\quad s\cdot x\in \mc V+m,
\]
where $s\cdot x$ denotes multiplcation in $\mbb F_{2^\ell}$ and $\mc V+m=\{v+m:v\in\mc V\}$. We call $\beta$ the \textit{seeded coset function determined by $\mc V$ and $\mc N$}. We show in this section that there exist parameters $\ell$ and $k$ for which $\mc V,\mc N$ can be chosen such that $\beta$ is a biregular irreducible function with large and precisely characterizable regularity set $\mc M$ and sufficiently small $\lambda_2(\beta,m)$.

If one chooses $\mc M=\mc N$, then one obtains the \textit{unconstrained seeded coset function} $\beta^o$. Choose basis elements $e_1,\ldots,e_\ell$ of $\mbb F_{2^\ell}$ over $\mbb F_2$ in such a way that $e_1,\ldots,e_k$ are a basis of $\mc N$ and $e_{k+1},\ldots,e_\ell$ are a basis of $\mc V$. Then $\beta^o$ obtains the form
\[
	\beta^o(s,x)=(s\cdot x)\vert_k,
\]
where every element $x$ of $\mbb F_{2^\ell}$ is represented by the binary sequence of length $\ell$ given by its coefficients in the basis $e_1,\ldots,e_{2^\ell}$ and $x\vert_k$ means the restriction of the coefficient sequence to the first $k$ bits, i.e., the coefficients of $e_1,\ldots,e_k$. In \cite{BBCM_genprivampl}, $\beta^o$ was defined in this bit-wise form and shown to be a universal hash function. In \cite{BT_Poly_time} and \cite{TalVardy_Upgrading} it was used as the security component of a modular UHF scheme which achieves the semantic security capacity of symmetric and degraded discrete wiretap channels. In \cite{TV_UHF_preprint} it was shown that modular UHF schemes with $\beta^o$ as the universal hash function achieve the strong secrecy capacity of Gaussian wiretap channels. By the discussion in Appendix \ref{app:strongsec}, there exists a large subset $\mc N'$ of $\mc N$ on which a small semantic security information leakage for the Gaussian wiretap channel is achieved by the modular UHF scheme restricted to $\mc N'$. Thus by showing that $\beta^o$ can be transformed into a good biregular irreducible function with a large message set for suitable parameters $k$ and $\ell$, we show that in this case, $\mc N'$ can even be chosen independent of the channel and characterized explicitly, even though it is not efficiently computable.
	
In order to obtain semantic security for more than symmetric and degraded discrete wiretap channels, $\mc N$ and $\mc V$ need to be chosen more specifically, and the regularity set $\mc M$ has to be a nontrivial subset of $\mc N$.

\subsection{Conditions for $\beta$ to be a biregular irreducible function}

The main result of this section is Theorem \ref{thm:BT-function}, where some combinations of $\ell,k$ and subspaces $\mc V,\mc N$ are found which make $\beta$ a biregular irreducible function with large regularity set $\mc M$ and small $\lambda_2(\beta,m)$ for every $m\in\mc M$. To define the seeded coset functions which are good biregular irreducible functions, recall the following lemma from finite field theory.

\begin{lem}[E.g., \cite{LidlNiederreiter}, Theorem 2.6]\label{lem:subfields}
    Let $p$ be a prime number. Every subfield of $\mbb F_{p^n}$ has $p^m$ elements for some positive divisor $m$ of $n$. Conversely, if $m$ is a positive divisor of $n$, then there is exactly one subfield of $\mbb F_{p^n}$ with $p^m$ elements. In particular, the unique subfield of $\mbb F_{p^n}$ with $p^m$ elements can be identified with $\mbb F_{p^m}$. 
\end{lem}

We can thus define a seeded coset function by choosing $\mc V$ to be the unique subspace of $\mbb F_{2^\ell}$ over $\mbb F_2$ which equals $\mbb F_{2^b}$ for any $b$ dividing $\ell$. The properties of the corresponding $\beta$ are summarized in the following theorem.

\begin{thm}\label{thm:BT-function}
    Assume that $b$ divides $\ell$. Let $\mc V=\mbb F_{2^b}$ and let $\mc N$ be any linear subspace of dimension $k=\ell-b$ satisfying $\dim(\mc N\cap\mc V)=0$. Define
	\[
        \mc M:=\{m\in\mc N:\mbb F_{2^b}(m)=\mbb F_{2^\ell}\},
	\]
	where $\mbb F_{2^b}(m)$ is the smallest subfield of $\mbb F_{2^\ell}$ which contains $\mbb F_{2^b}$ and $m$. Then the seeded coset function $\beta:\mbb F_{2^\ell}^*\times\mbb F_{2^\ell}^*\rightarrow\mc N$ defined by $\mc V$ and $\mc N$ is a biregular irreducible function with regularity set $\mc M$ satisfying
    \[
        \lambda_2(\beta,m)\leq\left(\frac{k}{b}\right)^22^{-b}.
    \]
    for every $m\in\mc M$. Moreover,
    \begin{equation}\label{eq:M-bound}
        \lvert\mc M\rvert
        =\frac{\ell}{b}N_{2^b}\left(\frac{\ell}{b}\right)2^{-b},
    \end{equation}
    where
    \[
        N_q(n)=\frac{1}{n}\sum_{d\mid n}\mu(d)q^{n/d}
    \]
    is the number of monic\footnote{A polynomial is \textit{monic} is its leading coefficient equals $1$.} irreducible polynomials of degree $n$ over $\mbb F_q$ and $\mu(d)$ is the \emph{M\"obius function} defined by
    \[
        \mu(d)=
        \begin{cases}
            1&d=1,\\
            (-1)^k&\text{if $d$ is the product of $k$ distinct primes},\\
            0&\text{else}.
        \end{cases}
    \]
\end{thm}

\begin{proof}
    See Section \ref{sect:BT-BRI}.
\end{proof}

\begin{cor}\label{cor:BT_rate}
    Let $\mc Z$ be any measurable space and $S$ uniformly distributed on $\mbb F_{2^\ell}^*$. Then for the $\beta$ from Theorem \ref{thm:BT-function} and any channel $W:\mbb F_{2^\ell}^*\rightarrow\mc Z$ and every $m\in\mc M$,
    \begin{align*}
        &D(Q_{f,m}W\Vert P_{\mc X}W\vert P_{\mc S})\\
        &\leq\frac{1}{\ln 2}\left(\frac{k}{b}\right)^22^{-(b-D_2^\varepsilon(W\Vert P_{\mc X}W\vert P_{\mc X}))}
        +\varepsilon\log k-(1-\varepsilon)\log(1-\varepsilon).
    \end{align*}
\end{cor}

Note that Corollary \ref{cor:BT_rate} gives almost the same bound on $D(Q_{f,m}W\Vert P_{\mc X}W\vert P_{\mc S})$, uniformly in the message, as the leftover hash lemma of \cite{TV_UHF_preprint} for modular UHF schemes does on $I(\overline M\wedge \overline Z,S)$, where the message $\overline M$ is uniformly distributed on the message set and the eavesdropper's observation $\overline Z$ is generated by $S$ and $\overline M$ (cf. also Subsection \ref{subsect:maxinf}).

A drawback of Theorem \ref{thm:BT-function} is that it only makes a statement for specific relations between $k$ and $b$ ($b$ has to divide $b+k$). In particular, it does not say anything about the case $k<b$. The main reason for the inflexible relation of the parameters $k$ and $b$ for seeded coset functions is that, to the authors' knowledge, no analog to Lemma \ref{lem:Katz} below exists for arbitrary linear subspaces $\mc V$ instead of subfields. However, recall that the seeded UHF wiretap code using $\beta^o$ has been shown to achieve strong secrecy for all Gaussian wiretap channels. Again, by the discussion in Appendix \ref{app:strongsec}, a large subset of the messages exists with small semantic security information leakage for the code restricted to this subset. The question then remains whether this set can also be chosen independently of the channel for general parameters.

The following lemma shows that the regularity set $\mc M$ remains large compared with the full subspace $\mc N$. It will become important in the asymptotic analysis in Section \ref{sect:asymptotics}.

\begin{lem}\label{lem:BT_qualifying_card}
    For all positive integers $a$ and $b$, with $k=(a-1)b$,
    \[
        aN_{2^b}(a)2^{-b}\geq 2^k\left(1-\frac{1}{2^{ab/2-1}}\right).
    \]
\end{lem}

\begin{proof}
    One easily shows that
    \[
        N_q(n)\geq\frac{1}{n}q^n-\frac{q}{n(q-1)}\bigl(q^{n/2}-1\bigr).
    \]
    for any prime power $q$ and positive integer $n$, see also \cite[Exercise 3.27]{LidlNiederreiter}. Therefore
    \begin{align*}
        aN_{2^b}(a)2^{-b}
        &\geq\left(2^{ab}-\frac{2^b}{2^b-1}\bigl(2^{ab/2}-1\bigr)\right)2^{-b}\\
        &=2^k-\frac{2^{ab/2}-1}{2^b-1}\\
        &=2^k\left(1-\frac{2^{ab/2}-1}{2^{ab}-2^k}\right)\\
        &\geq2^k\left(1-\frac{1}{2^{ab/2-1}}\right).
    \end{align*}
\end{proof}

\section{Asymptotic Consequences}\label{sect:asymptotics}

Next we test the performance of biregular irreducible functions in the asymptotic setting. We concentrate on memoryless wiretap channels, where formulas for the secrecy capacities are known. In the first subsection, we recall the definition of memoryless wiretap channels. In the second subsection, we state the central asymptotic coding results for sequences of modular BRI schemes and the ordinary wiretap codes constructed from these by seed reuse. The sequences of biregular irreducible functions applied in these codes are analyzed further in the last part of this section.

\subsection{Memoryless wiretap channels and capacities}

A \textit{memoryless wiretap channel} is determined by 
\begin{enumerate}
    \item a one-shot wiretap channel $(T:\mc A\to\mc Y,U:\mc A\to\mc Z)$, and
    \item a sequence $(\mc A_n')_{n=1}^\infty$ of sets such that $\mc A_n'\subset\mc A^n$, called the \textit{sets of permissible inputs}. We say a channel has \textit{no input constraints} if $\mc A_n'=\mc A^n$ for all $n$.
\end{enumerate}
We will denote a memoryless wiretap channel by its determining one-shot wiretap channel; the sequence $(\mc A_n')$ will usually be omitted in the notation.

A \textit{blocklength-$n$ seeded wiretap code} for the memoryless wiretap channel $(T,U)$ is a seeded wiretap code for the one-shot wiretap channel $(T^n:\mc A_n'\to\mc Y^n,U^n:\mc A_n'\to\mc Z^n)$, where $T^n$ and $U^n$ are the blocklength-$n$ memoryless extensions of $T$ and $U$, respectively (see Example \ref{ex:memoryless}), with inputs restricted to $\mc A_n'$. We call a blocklength-$n$ seeded wiretap code \textit{ordinary} if its seed set contains a single element.

If $(\xi_n,\zeta_n)_{n=1}^\infty$ is a sequence of seeded wiretap codes, where $(\xi_n,\zeta_n)$ is a blocklength-$n$ code with message set $\mc M_n$, we say that this sequence \textit{achieves the rate} $r\geq 0$ \textit{with semantic security} if 
\begin{align}
    \liminf_{n\rightarrow\infty}\frac{\log\lvert\mc M_n\rvert}{n}&\geq r,\label{eq:rate_rate}\\
    \lim_{n\rightarrow\infty}e(\xi_n,\zeta_n)&=0\label{eq:rate_rel}\\
    \lim_{n\rightarrow\infty}L_\sem(\xi_n,\zeta_n)&=0.\label{eq:rate_semsec}
\end{align}

A number $r\geq 0$ is called an \textit{achievable semantic security rate} if there exists a sequence of \textit{ordinary} wiretap codes which achieves the rate $r$ with semantic security. A number $r\geq 0$ is called an \textit{achievable semantic security rate with common randomness} if there exists a sequence of seeded wiretap codes which achieves the rate $r$ with semantic security. The supremum of all achievable semantic security rates (with common randomness) is called the \textit{semantic security capacity (with common randomness)} of the wiretap channel $(T,U)$.

The definitions of \textit{achievable strong secrecy rate (with common randomness)} and \textit{strong secrecy capacity (with common randomness)} are analogous to the above with the exception that in \eqref{eq:rate_semsec}, the semantic security information leakage $L_\sem(\xi_n,\zeta_n)$ is replaced by the strong secrecy information leakage $L_\str(\xi_n,\zeta_n)$ (see Section \ref{sect:oneshotwiretap}).

\begin{ex}\label{ex:discreteWiretap}
    If both $T:\mc A\to\mc Y$ and $U:\mc A\to\mc Z$ are discrete channels, then the memoryless wiretap channel $(T,U)$ is called a \textit{discrete wiretap channel}. We assume it has no input constraints. Its strong secrecy capacity is given by
    \begin{equation}\label{eq:discrseccap}
        \max\bigl(I(R\wedge Y)-I(R\wedge Z)\bigr),
    \end{equation}
    where the maximum is over finite sets $\mc R$ of size $\lvert\mc R\rvert\leq\lvert\mc A\rvert$, channels $\rho:\mc R\to\mc A$ and random variables $R$ on $\mc R$ such that $Y$ is generated by $R$ via $\rho T$ and $Z$ is generated by $R$ via $\rho U$ (Csisz\'ar \cite{Cs_strongsec}). It follows from the results of Wiese, N\"otzel and Boche \cite{AVWC} that even if common randomness is available, no higher rate is achievable. In particular, \eqref{eq:discrseccap} is both the strong secrecy capacity and the strong secrecy capacity with common randomness of $(T,U)$.

\end{ex}

\begin{ex}\label{ex:GaussianWiretap}
    Let $T:\mbb R\to\mbb R$ and $U:\mbb R\to\mbb R$ be Gaussian channels with noise variances $\sigma_T^2$ and $\sigma_U^2$, respectively. For any $\Gamma\geq 0$, the memoryless wiretap channel $(T,U)$ is called a \textit{Gaussian wiretap channel with input power constraint $\Gamma$} if for every blocklength $n$, the set of permissible inputs is given by the closed ball
    \[
        \overline{\mc B_n(\sqrt{n\Gamma})}
        =\{a\in\mbb R^n:\lVert a\rVert^2\leq n\Gamma\},
    \]
    where $\lVert\cdot\rVert$ denotes the Euclidean norm. The strong secrecy capacity of the Gaussian wiretap channel with input power constraint $\Gamma$ is given by
    \begin{equation}\label{eq:Gaussianseccap}
        \begin{cases}
            \frac{1}{2}\log\left(1+\frac{\Gamma}{\sigma_T^2}\right)-\frac{1}{2}\log\left(1+\frac{\Gamma}{\sigma_U^2}\right),&\text{if }\sigma_T^2\geq\sigma_U^2,\\
            0&\text{else,}
        \end{cases}
    \end{equation}
    as was shown, e.g., in \cite{TV_UHF_preprint}. We are not aware of any results upper-bounding the achievable strong secrecy rates with common randomness for the Gaussian wiretap channel, but we conjecture them to be no larger than \eqref{eq:Gaussianseccap} similar to the discrete case.

\end{ex}

The discussion in Appendix \ref{app:strongsec} shows that every achievable strong secrecy rate (with common randomness) for any asymptotic wiretap channel, not limited to memoryless wiretap channels, also is an achievable semantic security rate (with common randomness). Therefore \eqref{eq:discrseccap} also is the semantic security capacity and the semantic security capacity with common randomness of the discrete wiretap channel $(T,U)$, so henceforth we will just call it the \textit{secrecy capacity} of $(T,U)$. Similarly, \eqref{eq:Gaussianseccap} is the semantic security capacity of the Gaussian wiretap channel $(T,U)$. Since we cannot rule out the possibility that one achieves more with common randomness, we will call it the \textit{ordinary secrecy capacity} of the Gaussian wiretap channel.

\subsection{Asymptotics for biregular irreducible functions}\label{subsect:asy_codes}

For all the coding results on modular BRI schemes and the corresponding ordinary wiretap codes obtained by seed reuse, we apply the following sequences of biregular irreducible functions.

\begin{thm}\label{thm:existence_opt_BRI}
    Let $r\geq 0$ and $0\leq t< 1$. Then there exists a sequence $(f_n:\mc S_n\times\mc X_n\to\mc N_n)_{n=1}^\infty$ of biregular irreducible functions, with $\mc M_n$ the regularity set of $f_n$, satisfying
    \begin{align}
        \lim_{n\to\infty}\frac{\log\lvert\mc X_n\rvert}{n}
        &= r,\label{eq:BRI_ECC_rate}\\
        \lim_{n\rightarrow\infty}\dfrac{\log\lvert\mc M_n\rvert}{\log\lvert\mc X_n\rvert}
        &=1-t,\label{eq:BRI_rate}\\
        \liminf_{n\rightarrow\infty}\frac{\min_{m\in\mc M_n}(-\log\lambda_2(f_n,m))}{\log\lvert\mc X_n\rvert}
        &\geq t.\label{eq:BRI_minlb}
    \end{align}
    Moreover, every $f_n$ satisfies $\lvert\mc S_n\rvert\leq\lvert\mc X_n\rvert$.
\end{thm}

Theorem \ref{thm:existence_opt_BRI} is proved in Section \ref{sect:existence_opt_BRI_proof}. To show the general statement, we will apply the Ramanujan biregular irreducible functions constructed in Section \ref{sect:BRI}. If $t$ is the inverse of a positive integer, then seeded coset biregular irreducible functions can be applied as well. We will analyze sequences satisfying \eqref{eq:BRI_rate}-\eqref{eq:BRI_minlb} in the last subsection of this section. We will see that \eqref{eq:BRI_rate} and \eqref{eq:BRI_minlb} together imply equality and the existence of the limit in \eqref{eq:BRI_minlb}. In this sense, every sequence satisfying \eqref{eq:BRI_rate} and \eqref{eq:BRI_minlb} is optimal in terms of the trade-off between the growth of the regularity set and the decrease of the second-largest eigenvalue modulus.

We can now state the coding results for modular BRI schemes applied to discrete and Gaussian memoryless wiretap channels. When we say that a sequence of blocklength-$n$ error-correcting codes $(\phi_n,\psi_n)$ for the channel $T$ achieves a rate $r$, then we mean that
\begin{align*}
    \liminf_{n\to\infty}\frac{\log\lvert\mc X_n\rvert}{n}\geq r, \qquad\lim_{n\to\infty} e(\phi_n,\psi_n)=0,
\end{align*}
where $\mc X_n$ is the message set of $(\phi_n,\psi_n)$.

For the discrete case, we also need to define $\delta$-typical sets. If $\delta>0$, $\mc A$ is a finite set, $P$ a probability distribution on $\mc A$ and $n$ a positive integer, then the \textit{$\delta$-typical set} $T_{P,\delta}^n$ of $P$ is defined as the set of $(a_1,\ldots,a_n)\in\mc A^n$ satisfying
\[
    \left\lvert\frac{\lvert\{i:a_i=a\}\rvert}{n}-P(a)\right\rvert\leq\delta
\]
for every $a\in\mc A$ and where $P(a)=0$ implies $\{i:a_i=a\}=\varnothing$. Sometimes we will call a code whose encoder $\phi_n:\mc X_n\to\mc A^n$ satisfies $\phi_n(T_{P,\delta}^n\vert x)=1$ for every $x\in\mc X_n$ a \textit{constant composition code}.

\begin{lem}\label{lem:discrete_secfromBRI_types}
    Let $r\geq 0$ and $0\leq t<1$ and let $(f_n)_{n=1}^\infty$ be a sequence of biregular irreducible functions satisfying \eqref{eq:BRI_ECC_rate}-\eqref{eq:BRI_minlb}. Then, for any discrete wiretap channel $(T,U)$ without input constraints and every sequence of blocklength-$n$ codes $(\phi_n,\psi_n)$ for $T$ which achieves a rate strictly larger than $r$, the sequence of modular BRI schemes $\Pi(f_n,\phi_n,\psi_n)$ achieves the semantic security rate 
    \begin{equation}\label{eq:semsecrate}
        (1-t)r
    \end{equation}
    with exponentially decreasing semantic security leakage if
    \[
        tr>\max_PI(P,U).
    \]
    
    Moreover, if there exists a probability distribution $P$ and a $\delta_1>0$ such that all blocklength-$n$ codewords are contained in $T_{P,\delta_1}^n$ for all $n$, then the $\Pi(f_n,\phi_n,\psi_n)$ achieve the semantic security rate \eqref{eq:semsecrate} with exponentially decreasing semantic security leakage if
    \begin{equation}\label{eq:cond_sec_gen}
        tr>I(P,U)+\gamma_d(\delta_1,\lvert\mc Z\rvert),
    \end{equation}
    where $\gamma_d(\delta_1,\lvert\mc Z\rvert)$ is a function which tends to zero as $\delta_1$ tends to zero.
\end{lem}

The analogous result for Gaussian wiretap channels is the following.

\begin{lem}\label{lem:Gaussian_secfromBRI}
    Let $r\geq 0$ and $0\leq t<1$ and let $(f_n)_{n=1}^\infty$ be a sequence of biregular irreducible functions satisfying \eqref{eq:BRI_ECC_rate}-\eqref{eq:BRI_minlb}. Then for any Gaussian wiretap channel $(T,U)$, where $T$ has noise variance $\sigma_T^2$ and $U$ has noise variance $\sigma_U^2<\sigma_T^2$ and with input power constraint $\Gamma$, and for every sequence of blocklength-$n$ codes $(\phi_n,\psi_n)$ for $(T,U)$ which achieves a rate strictly larger than $r$, the sequence of modular BRI schemes $\Pi(f_n,\phi_n,\psi_n)$ achieves the semantic security rate 
    \[
        (1-t)r
    \]
    with exponentially decreasing semantic security leakage if
    \[
        tr>\frac{1}{2}\log\left(1+\frac{\Gamma}{\sigma_U^2}\right).
    \]
\end{lem}

\begin{proof}[Proofs of Lemmas \ref{lem:discrete_secfromBRI_types} and \ref{lem:Gaussian_secfromBRI}]
    See Section \ref{sect:assec_proof}.
\end{proof}

Lemmas \ref{lem:discrete_secfromBRI_types} and \ref{lem:Gaussian_secfromBRI} cleanly separate the tasks of error correction and security generation. Whether semantic security is achievable only depends on the comparison of two parameters, one due to the biregular irreducible functions and one due to the channel to the eavesdropper. This gives some robustness with respect to the required knowledge about the channel to the eavesdropper. For example, as long as the wiretap channel is discrete, security can be ensured for all eavesdropper channels whose capacity $\max_PI(P,U)$ is at most $tr$, and possibly even more if the error-correcting code is a constant-composition code. Thus the generation of security seems to be much simpler than the correction of errors.

Both lemmas imply in particular that the secrecy capacities of the discrete and Gaussian wiretap channels (\eqref{eq:discrseccap} and \eqref{eq:Gaussianseccap}) are achievable by modular BRI schemes ensuring semantic security. The lemmas can be applied directly to the discrete memoryless wiretap channel $(T:\mc A\to\mc Y,U:\mc A\to\mc Z)$ when $T$ is \textit{more capable} than $U$, meaning that $I(P,T)\geq I(P,U)$ for all probability distributions $P$ on $\mc A$. In this case, its secrecy capacity is given by $\max_P(I(P,T)-I(P,U))$, i.e., the capacity expression does not involve any maximization over auxiliary channels. Assume that the maximum is attained by the distribution $P$ and choose a sequence of error-correcting codes $(\phi_n,\psi_n)$ which achieves rate $I(P,T)$ and whose codewords are all contained in $T_{P,\delta}^n$ for some $\delta>0$. By choosing 
\[
    r=I(P,T)-\delta,\qquad t=\frac{I(P,U)+\gamma_d(\delta,\lvert\mc Z\rvert)+\delta}{r},
\]
we have
\[
    tr>I(P,U)+\gamma_d(\delta,\lvert\mc Z\rvert).
\]
Thus with any sequence of biregular irreducible functions satisfying \eqref{eq:BRI_ECC_rate}-\eqref{eq:BRI_minlb}, Lemma \ref{lem:discrete_secfromBRI_types} implies that the sequence of modular BRI schemes $\Pi(f_n,\phi_n,\psi_n)$ achieves the semantic security rate 
\[
    I(P,T)-I(P,U)-\gamma_d(\delta,\lvert\mc Z\rvert)-2\delta
\]
with common randomness, which is arbitrarily close to the secrecy capacity of $(T,U)$. The analogous method works for Gaussian wiretap channels. For all these codes, it is sufficient to restrict the encoders of the error-correcting codes to be deterministic mappings instead of allowing them to be channels.

For general discrete memoryless wiretap channels, the maximum in the capacity expression \eqref{eq:discrseccap} is in general attained by a nontrivial auxiliary channel $\rho:\mc R\to\mc A$ and a probability distribution $R$ on $\mc R$. In this case, Lemma \ref{lem:discrete_secfromBRI_types} cannot be applied to $(T,U)$ directly. Instead, it is applied to the effective wiretap channel $(\rho T,\rho U)$. In the same way as when $T$ is more capable than $U$, one can conclude from Lemma \ref{lem:discrete_secfromBRI_types} that the rate
\[
    \max_{P_R}\bigl(I(P_R,\rho T)-I(P_R,\rho U)\bigr)
\]
is achievable over $(\rho T,\rho U)$ by modular BRI schemes, where $P_R$ ranges over the probability distributions on $\mc R$. By considering the auxiliary channel $\rho$ to be part of the encoders of the error-correcting codes, every rate achievable for $(\rho T,\rho U)$ translates to an achievable rate for $(T,U)$. (More precisely, if $\phi_n:\mc X_n\to\mc R^n$ is the encoder of a blocklength-$n$ error-correcting code for $\rho T$, then $\phi_n\rho^n$ is the encoder of a blocklength-$n$ error-correcting code for $T$.) This explains the necessity of allowing the encoders of error-correcting codes to be channels in general.

The case where $T$ is not more capable than $U$ also shows that condition \eqref{eq:cond_sec_gen} could still be improved by looking more closely at the internal structure of the error-correcting encoders. The condition comes from the bound derived in Lemma \ref{lem:maxinf_discrete} below, where only completely general encoders as well as encoders of constant composition codes are considered. By extending Lemma \ref{lem:maxinf_discrete} to the case of encoders $\phi_n$ of the form $\tilde\phi_n\rho^n$, Lemma \ref{lem:discrete_secfromBRI_types} could be strengthened in such a way that it would imply the existence of capacity-achieving modular BRI schemes as directly as for wiretap channels where $T$ is more capable than $U$.

Modular BRI schemes suffer from the disadvantage of having to assume the random seed to be known to the sender and the receiver. In order to transform a modular BRI scheme into an ordinary wiretap code without rate loss, one can first transmit the seed without security constraints and then apply a modular BRI scheme with this seed multiple times. This was already proposed by \cite{BT_Poly_time} for modular UHF schemes. Our analysis of this method is different from that in \cite{BT_Poly_time} and more information-theoretic in nature. Although it is formulated in the context of modular BRI schemes, the construction and its analysis carries over to modular UHF schemes and even to general seeded wiretap codes provided their seed lengths do not increase too fast with blocklength in comparison with the decrease of the error probability and the semantic security information leakage.

We now formalize the seed reuse technique. Let $(T:\mc A\to\mc Y,U:\mc A\to\mc Z)$ be a memoryless wiretap channel. Let $(\phi,\psi)$ be any blocklength-$n$ error-correcting code for $T$ with message set $\mc X$ and let $f:\mc S\times\mc X\to\mc N$ be a biregular irreducible function with regularity set $\mc M$ such that $\lvert\mc S\rvert\leq\lvert\mc X\rvert$, so that $\mc S$ can be considered to be embedded in $\mc X$. The code $(\phi,\psi)$ and the biregular irreducible function $f$ determine a modular BRI scheme $\Pi(f,\phi,\psi)$. For any positive integer $N$, we denote by $R_N(f,\phi,\psi)$ the \textit{ordinary} wiretap code $(\xi:\mc M^N\to\mc A^{(N+1)n},\zeta:\mc Y^{(N+1)n}\to\mc M^N)$ with blocklength $(N+1)n$ whose components are defined as follows:
\begin{enumerate}
    \item Let $\tilde\xi:\mc S\times\mc M\to\mc A_n'$ be the encoder of the modular BRI scheme $\Pi(f,\phi,\psi)$. Then in order to encode the message $(m_1,\ldots,m_N)$, the encoder $\xi$ of $R_N(f,\phi,\psi)$ uniformly at random chooses an $s\in\mc S$ and then applies $\phi(\cdot\vert s),\tilde\xi(\cdot\vert s,m_1),\ldots,\tilde\xi(\cdot\vert s,m_N)$ in this order. In other words, the probability that the channel input $a^{(N+1)n}\in\mc A^{(N+1)n}$ consisting of $N+1$ blocks $a_1,\ldots,a_{N+1}$ of length $n$ each is chosen by $\xi$ is given by
    \[
        \frac{1}{\lvert\mc S\rvert}\sum_{s\in\mc S}\phi(a_1\vert s)\prod_{j=1}^N\tilde\xi(a_{j+1}\vert s,m_j).
    \]
    Note that for $\xi$ to map into $A_{(N+1)n}'$, the input constraints have to be consistent in that $(\mc A_n')^{N+1}\subset\mc A_{(N+1)n}'$. This is clearly satisfied in the case of no input constraints and also for the Gaussian input constraints of Example \ref{ex:GaussianWiretap}.
    
    \item Let $\tilde\zeta:\mc S\times\mc Y\to\mc M$ be the decoder of the modular BRI scheme. Then a channel output $y^{N+1}=(y_1,\ldots,y_{N+1})$ of $T^{(N+1)n}$, where each $y_i$ has length $n$, is mapped to the message sequence
    \[
        (\tilde\zeta(\hat s,y_2),\ldots,\tilde\zeta(\hat s,y_{N+1})),
    \]
    where $\hat s=\psi(y_1)$.
\end{enumerate}
The structure of $R_N(f,\phi,\psi)$ is illustrated in Fig. \ref{fig:ordinary_from_mod}. Clearly, $R_N(f,\phi,\psi)$ is an ordinary wiretap code. We call it an \textit{ordinary BRI wiretap code} and denote its error probability and semantic security leakage by $e(R_N(f,\phi,\psi))$ and $L_\sem(R_N(f,\phi,\psi))$, respectively. The random seed necessary for the application of $\Pi(f,\phi,\psi)$ now becomes part of the stochastic encoding done locally by the sender. 

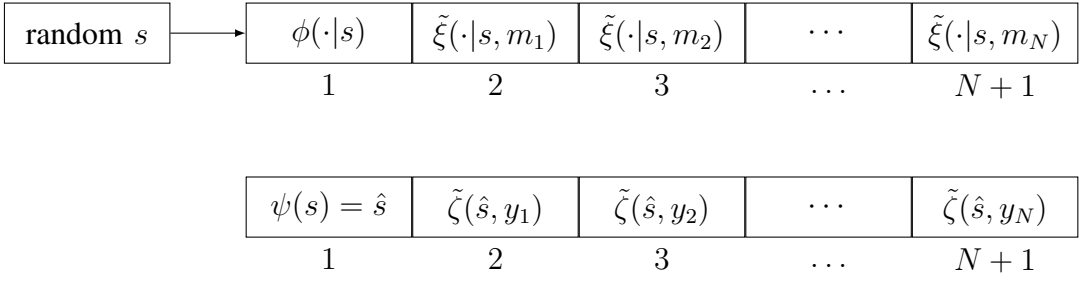
\begin{figure*}
\centering
\begin{tikzpicture}[kasten/.style={draw, minimum height = .8cm, minimum width = 2.2cm, inner sep=4pt}, pfeil/.style={->, >=latex}
	]
    \node[kasten] (seed_enc) {$\phi(\cdot\vert s)$};
    \node[kasten, left = 1cm of seed_enc] (random_seed) {random $s$};
    \node[kasten, right = 0cm of seed_enc] (enc1) {$\tilde\xi(\cdot\vert s,m_1)$};
    \node[kasten, right = 0cm of enc1] (enc2) {$\tilde\xi(\cdot\vert s,m_2)$};
    \node[kasten, right = 0cm of enc2] (dots_enc) {$\cdots$};
    \node[kasten, right = 0cm of dots_enc] (encN) {$\tilde\xi(\cdot\vert s,m_N)$};
    
    \node[below = 0cm of seed_enc] {$1$};
    \node[below = 0cm of enc1] {$2$};
    \node[below = 0cm of enc2] {$3$};
    \node[below = 0.2cm of dots_enc] {$\ldots$};
    \node[below = 0cm of encN] {$N+1$};
    
    \node[kasten, below = 1.5cm of seed_enc] (seed_dec) {$\psi(s)=\hat s$};
    \node[kasten, right = 0cm of seed_dec] (dec1) {$\tilde\zeta(\hat s,y_1)$};
    \node[kasten, right = 0cm of dec1] (dec2) {$\tilde\zeta(\hat s,y_2)$};
    \node[kasten, right = 0cm of dec2] (dots_dec) {$\cdots$};
    \node[kasten, right = 0cm of dots_dec] (decN) {$\tilde\zeta(\hat s,y_N)$};
    
    \node[below = 0cm of seed_dec] {$1$};
    \node[below = 0cm of dec1] {$2$};
    \node[below = 0cm of dec2] {$3$};
    \node[below = 0.2cm of dots_dec] {$\ldots$};
    \node[below = 0cm of decN] {$N+1$};
    
    \draw[pfeil] (random_seed) -> (seed_enc);
 \end{tikzpicture}
    \caption{The encoder (top) and decoder (bottom) structure of the seeded wiretap code $R_N(f,\phi,\psi)$.}
    \label{fig:ordinary_from_mod}
\end{figure*}

\begin{cor}\label{cor:assec}
    Let $r\geq 0$ and $0\leq t<1$ and let $(f_n)_{n=1}^\infty$ be a sequence of biregular irreducible functions satisfying \eqref{eq:BRI_ECC_rate}-\eqref{eq:BRI_minlb} and $\lvert\mc S_n\rvert\leq\lvert\mc X_n\rvert$ for all $n$. Then the conclusions of Lemmas \ref{lem:discrete_secfromBRI_types} and \ref{lem:Gaussian_secfromBRI} continue to hold if one replaces the sequence of modular BRI schemes by the sequence $(R_{N_n}(f_n,\phi_n,\psi_n))_{n=1}^\infty$ of ordinary BRI wiretap codes, where the sequence $(N_n)_{n=1}^\infty$ satisfies 
    \begin{align*}
        N_n\longrightarrow\infty,\qquad
        &N_n\max\bigl\{e(\Pi(f_n,\phi_n,\psi_n)),L_\sem(\Pi(f_n,\phi_n,\psi_n))\bigr\}\longrightarrow 0,
    \end{align*}
    e.g., 
    \[
        N_n=\left\lfloor-\log\max\bigl\{e(\Pi(f_n,\phi_n,\psi_n)),L_\sem(\Pi(f_n,\phi_n,\psi_n))\bigr\}\right\rfloor.
    \]
\end{cor}

\begin{proof}
    See Section \ref{sect:assec_proof}.
\end{proof}

Similar remarks as those made after Lemma \ref{lem:discrete_secfromBRI_types} on the separation of error correction and security generation as well as on how to achieve the secrecy capacity apply here.

\subsection{Further analysis}\label{subsect:further_analysis}

Here we analyze sequences of biregular irreducible functions $f_i:\mc S_i\times\mc X_i\to\mc N_i$ satisfying \eqref{eq:BRI_rate} and \eqref{eq:BRI_minlb} of Theorem \ref{thm:existence_opt_BRI}. We ignore \eqref{eq:BRI_ECC_rate} in the analysis because for any $r\geq 0$, there exists a subsequence $n_i$ of the positive integers such that $n_i^{-1}\log\lvert\mc X_i\rvert$ tends to $r$. Which $r$ is necessary, whether there may be gaps between the $n_i$, and if so, how large these are allowed to be, depends on the application. When we say that ``the sequence of biregular irreducible functions $(f_i)_{i=1}^\infty$ satisfies \eqref{eq:BRI_rate} and \eqref{eq:BRI_minlb} with parameter $t$'' for some $0\leq t<1$, then we mean that $\lvert\mc X_i\rvert\to\infty$ and
\begin{align*}
    \lim_{i\to\infty}\frac{\log\lvert\mc M_i\rvert}{\log\lvert\mc X_i\rvert}= 1-t,
    \qquad\liminf_{i\to\infty}\frac{\min_{m\in\mc M_i}(-\log\lambda_2(f_i,m))}{\log\lvert\mc X_i\rvert}\geq t.
\end{align*}

\paragraph{Eigenvalues vs. rate of biregular irreducible functions}

The first important observation we make is that the trade-off between the size of the regularity set and the magnitude of the second-largest eigenvalue cannot be improved. Hence it is justified to call sequences satisfying \eqref{eq:BRI_rate} and \eqref{eq:BRI_minlb} optimal.

\begin{lem}\label{lem:EV-qualset-sum}
    If the sequence of biregular irreducible functions $(f_i)_{i=1}^\infty$ satisfies \eqref{eq:BRI_rate} and \eqref{eq:BRI_minlb} with parameter $t$ for some $0\leq t<1$, then the limit
    \begin{equation}\label{eq:EV_lim}
        \lim_{i\to\infty}\frac{\min_{m\in\mc M_i}(-\log\lambda_2(f_i,m))}{\log\lvert\mc X_i\rvert}
    \end{equation}
    exists and equals $t$.
\end{lem}

\begin{proof}
    See Section \ref{sect:EV-qualset-sum_proof}.
\end{proof}

To prove this result, we apply hypothetical sequences of biregular irreducible functions which satisfy \eqref{eq:BRI_minlb} with strict inequality and \eqref{eq:BRI_rate} to a wiretap channel. Proceeding similarly as in the proof of Lemma \ref{lem:degree_expgrowth}, this yields higher semantic security rates than possible for this channel.

\paragraph{Good biregular irreducible functions are nearly Ramanujan}

Conditions \eqref{eq:BRI_rate} and \eqref{eq:BRI_minlb} are formulated on the logarithmic scale, and so they ignore lots of details of the parameters of the corresponding graphs. However, on this scale, it turns out that sequences of biregular irreducible functions satisfying \eqref{eq:BRI_rate} and \eqref{eq:BRI_minlb} are ``nearly Ramanujan''. If the limits
\begin{equation}\label{eq:ex_lim_d_S_X}
    \lim_{i\to\infty}\frac{\log d_{\mc S_i}}{\log\lvert\mc X_i\rvert},
    \qquad\lim_{i\to\infty}\frac{\log d_{\mc X_i}}{\log\lvert\mc X_i\rvert}
\end{equation}
exist, we can say even more. 

\begin{lem}\label{lem:nearly_Ram}
    For every $i$, let $f_i:\mc S_i\times\mc X_i\to\mc N_i$ be a biregular irreducible function with regularity set $\mc M_i$. Assume that there exists a $0\leq t<1$ such that the sequence $(f_i)_{i=1}^\infty$ satisfies \eqref{eq:BRI_rate} and \eqref{eq:BRI_minlb} with parameter $t$. Then 
    \[
        \limsup_{i\to\infty}\frac{\max_{m\in\mc M_i}\log\lambda_2(G_{f_i,m})}{\log\lvert\mc X_i\rvert}
        \leq\limsup_{i\to\infty}\frac{\log\sqrt{d_{\mc X_i}}}{\log\lvert\mc X_i\rvert}.
    \]
    If the limits \eqref{eq:ex_lim_d_S_X} exist, then the limit on the left-hand side exists and
    \begin{equation}\label{eq:genau_nearly_Ram}
        \lim_{i\to\infty}\frac{\max_{m\in\mc M_i}\log\lambda_2(G_{f_i,m})}{\log\lvert\mc X_i\rvert}
        =\lim_{i\to\infty}\frac{\log\sqrt{d_{\mc X_i}}}{\log\lvert\mc X_i\rvert}.
    \end{equation}
\end{lem}

\begin{proof}
    See Section \ref{sect:EV-qualset-sum_proof}.
\end{proof}

Thus the growth of $\lambda_2(G_{f_i,m})$ as a fractional power of $\lvert\mc X_i\rvert$ can be at most as fast as that of $\sqrt{d_{\mc X_i}}$. This justifies calling the biregular irreducible functions ``nearly Ramanujan''. The sequences constructed in the proof of Theorem \ref{thm:existence_opt_BRI} for arbitrary parameters $r$ and $t$ consist of Ramanujan biregular irreducible functions, and these  in particular are ``nearly Ramanujan''. If $t$ is the inverse of a positive integer, then also a sequence of seeded coset biregular irreducible functions satisfying \eqref{eq:BRI_rate} and \eqref{eq:BRI_minlb} with parameter $t$ in this sense is ``nearly Ramanujan''.

It is not surprising that the existence of the limit \eqref{eq:EV_lim} plays an important role in the proof of \eqref{eq:genau_nearly_Ram}. If the limits \eqref{eq:ex_lim_d_S_X} exist and $d_{\mc S_i}=d_{\mc X_i}$ for large $i$, then the lemma says that not all of the graphs $G_{f_i,m}$ can have an exceptionally small second-largest eigenvalue the sense of the Feng-Li bound.  

Now assume that both limits in \eqref{eq:ex_lim_d_S_X} exist, but that the right-hand limit is strictly smaller than the left-hand one (in particular implying $\lvert\mc S_i\rvert<\lvert\mc X_i\rvert$ for large $i$). Then the  second-largest eigenvalue limit \eqref{eq:genau_nearly_Ram} is strictly smaller than the limit
\[
    \lim_{i\to\infty}\frac{\log(\sqrt{d_{\mc S_i}-1}+\sqrt{d_{\mc X_i}-1})}{\log\lvert\mc X_i\rvert}=\lim_{i\to\infty}\frac{\log\sqrt{d_{\mc S_i}}}{\log\lvert\mc X_i\rvert}
\]
corresponding to the Feng-Li bound. (For some more details see Remark \ref{rem:imbalanced_Feng-Li}.) This does not rule out the existence of such sequences a priori because the maximal degree has to increase with $\mc X_i$, see Lemma \ref{lem:degree_expgrowth}. However, note that we did not construct any such sequences.

\paragraph{Growth of degrees}

Conditions \eqref{eq:BRI_rate} and \eqref{eq:BRI_minlb} do not directly say anything about the degree pairs of the biregular irreducible functions. However, we know from the Feng-Li bound that the degree pair of a biregular graph is coupled to the second-largest eigenvalue. From this it follows that \eqref{eq:BRI_rate} and \eqref{eq:BRI_minlb} imply a lower bound on the growth rate of the maximum degree.

\begin{lem}\label{lem:degree_expgrowth}
    For every $i$, let $f_i:\mc S_i\times\mc X_i\to\mc N_i$ be a biregular irreducible function with regularity set $\mc M_i$ satisfying \eqref{eq:BRI_rate} and \eqref{eq:BRI_minlb} with parameter $0\leq t<1$. Then 
    \begin{equation}\label{eq:lb_deggrowth}
        \liminf_{i\to\infty}\frac{\log\max(d_{\mc S_i},d_{\mc X_i})}{\log\lvert\mc X_i\rvert}
        \geq\min\left\{\frac{1-t}{13},t\right\}.
    \end{equation}
    The right-hand side of \eqref{eq:lb_deggrowth} can be replaced by $t$ if one of the following conditions holds:
    \begin{enumerate}
        \item for every large $i$ there exists an $m\in\mc M_i$ such that the diameter\footnote{The diameter of a graph is defined in Appendix \ref{app:graphs}} $\Delta_{f_i,m}$ of $G_{f_i,m}$ is at least $8$, or
        \item for every large $i$, the regularity set $\mc M_i$ of $f_i$ is equal to $\mc N_i$.
    \end{enumerate}
\end{lem}

\begin{proof}
    See Section \ref{sect:EV-qualset-sum_proof}.
\end{proof}

\section{Proof of Theorem \ref{thm:EV-UB}}\label{sect:proof_EV-UB}

The proof of Theorem \ref{thm:EV-UB} is divided into three subsections. The first one reduces the statement of the theorem to one about subnormalized channels, the main calculation is done in the second subsection, and these two parts are put together in the concluding subsection. This strategy is the same as in \cite{TV_UHF_preprint}, but the steps themselves differ since we deal with the divergence for an individual message and with biregular functions, whereas \cite{TV_UHF_preprint} treated mutual information for uniformly distributed messages and universal hash functions.

\subsection{Reduction to subnormalized channels}

We first derive expressions for the density of $Q_{f,m}W$. Assume that $f$ is defined by the graph family $(G_{f,m})_{m\in\mc N}$. For $m\in\mc M$, $G_{f,m}$ is $(d_{\mc S},d_{\mc X})$-biregular with adjacency matrix $A_{f,m}$. Like in \eqref{eq:bip_adj_matr}, let $B_{f,m}$ be the top right $\mc S\times\mc X$ component of $A_{f,m}$ where $B_{f,m}(s,x)=1$ if and only if $f(s,x)=m$. Further, assume that the density of $Q_{f,m}$ is given by the stochastic matrix $q_{f,m}$. If $W$ has the $\mu$-density $w$, for every $z\in\mc Z$ define $w_z\in\mbb R^{\mc X}$ by $w_z(x)=w(z\vert x)$. Then the $\mu$-density of $Q_{f,m}W$ can be expressed as 
\begin{align}\label{eq:density_m_sent}
    \sum_{x\in\mc X}q_{f,m}(x\vert s)w(z\vert x)
    =\frac{1}{d_{\mc S}}\sum_{x:f(s,x)=m}w(z\vert x)
    &=\frac{1}{d_{\mc S}}(B_{f,m}w_z)(s).
\end{align}

The following reduction to subnormalized channels is proved in exactly the same way in \cite{TV_UHF_preprint} for universal hash functions instead of biregular irreducible functions.

\begin{lem}\label{lem:subnormalization}
    For some $\varepsilon>0$, let the measurable set $\mc T\subset\mc X\times\mc Z$ satisfy 
    \[
        W(\{z\in\mc Z:(x,z)\in\mc T\}\vert x)>1-\varepsilon
    \]
    for all $x\in\mc X$. If $S$ is uniformly distributed on $\mc S$, then
	\begin{align}
		D(Q_{f,m}W\Vert P_{\mc X}W\vert P_{\mc S})
		&\leq D(Q_{f,m}W_{\mc T}\Vert P_{\mc X}W_{\mc T}\vert P_{\mc S})+\varepsilon \log\frac{\lvert\mc X\rvert}{d_{\mc S}}.\label{eq:lem_subnormalization}
	\end{align}
\end{lem}

\begin{proof}
    The density of $Q_{f,m}W$ is given in \eqref{eq:density_m_sent}. Thus
\begin{align}
	&D(Q_{f,m}W\Vert P_{\mc X}W\vert P_{\mc S})\notag\\
	&=\frac{1}{d_{\mc S}\lvert\mc S\rvert}\sum_{s\in\mc S}\int\sum_{x:f(s,x)=m}w(z\vert x)
			\log\frac{\sum_{x':f(s,x')=m}w(z\vert x')}
				{d_{\mc S}\lvert\mc X\rvert^{-1}\sum_{x''\in\mc X}w(z\vert x'')}\,\mu(dz).\label{eq:tbc}
\end{align}
Due to the log-sum inequality (e.g., \cite[Lemma 2.7]{CK}), for every $s\in\mc S$, the integrand in \eqref{eq:tbc} can be upper-bounded by
\begin{align}
	&\sum_{x:(x,z)\in\mc T}w(z\vert x)1_{\{f(s,x)=m\}}\log\frac{\sum_{x':(x',z)\in\mc T}w(z\vert x')1_{\{f(s,x')=m\}}}{d_{\mc S}\lvert\mc X\rvert^{-1}\sum_{x'':(x'',z)\in\mc T}w(z\vert x'')}\notag\\
	&\quad +\sum_{x:(x,z)\notin\mc T}w(z\vert x)1_{\{f(s,x)=m\}}\log\frac{\sum_{x':(x',z)\notin\mc T}w(z\vert x')1_{\{f(s,x')=m\}}}{d_{\mc S}\lvert\mc X\rvert^{-1}\sum_{x'':(x'',z)\notin\mc T}w(z\vert x'')}\notag\\
	&\leq\sum_{x\in\mc X}w_{\mc T}(z\vert x)1_{\{f(s,x)=m\}}\log\frac{\sum_{x'\in\mc X}w_{\mc T}(z\vert x')1_{\{f(s,x')=m\}}}{d_{\mc S}\lvert\mc X\rvert^{-1}\sum_{x''\in\mc X}w_{\mc T}(z\vert x'')}\label{eq:subnSummand1}\\
	&\quad +\sum_{x:(x,z)\notin\mc T}w(z\vert x)1_{\{f(s,x)=m\}}\log\frac{\lvert\mc X\rvert}{d_{\mc S}}.\label{eq:subnSummand2}
\end{align}
If we denote the expression in \eqref{eq:subnSummand1} by $g(s,z)$, then clearly
\[
    \frac{1}{d_{\mc S}\lvert\mc S\rvert}\sum_{s\in\mc S}\int g(s,z)\,\mu(dz)
\]
exists and equals the conditional divergence in \eqref{eq:lem_subnormalization}. The term \eqref{eq:subnSummand2} is responsible for the $\varepsilon\log(\lvert\mc X\rvert/d_{\mc S})$ term in \eqref{eq:lem_subnormalization} due to the assumption on $\mc T$.
\end{proof}

\subsection{Eigenvalue upper bound on R\'enyi $2$-divergence}

Next we state the main ingredient to the proof of Theorem \ref{thm:EV-UB}, which is an upper bound on the mean of the exponential R\'enyi 2-divergence between $(Q_{f,m}W_{\mc T})(\cdot\vert s)$ and $P_{\mc X}W_{\mc T}$. The connection to the Kullback-Leibler divergence will be made later through an application of Lemma \ref{lem:RenyitoKL}.

\begin{lem}\label{lem:Renyi2byEV}
Choose any measurable set $\mc T\subset\mc X\times\mc Z$ such that $W_{\mc T}(\mc Z\vert x)>0$ for all $x\in\mc X$. For every biregular irreducible function $f:\mc S\times\mc X\rightarrow\mc N$ with regularity set $\mc M$, every $m\in\mc M$ satisfies
\begin{align*}
	\exp\bigl(D_2(Q_{f,m}W_{\mc T}\Vert P_{\mc X}W_{\mc T}\vert P_{\mc S})\bigr)
	&\leq1+\lambda_2(f,m)\exp\bigl(D_2(W_{\mc T}\Vert P_{\mc X}W_{\mc T}\vert P_{\mc X})\bigr).
\end{align*}
\end{lem}

\begin{proof}
As seen in  \eqref{eq:density_m_sent}, the density of $(Q_{f,m}W_{\mc T})(\cdot\vert s)$ equals $d_{\mc S}^{-1}(B_{f,m}w_z)(s)$. Also, denoting the $\mbb R^{\mc X}$-vector with every component equal to $1$ by $\bm 1$, the density of $P_{\mc X}W$ equals
\[
    \frac{1}{\lvert\mc X\rvert}\sum_{x\in\mc X}w(z\vert x)=\frac{1}{\lvert\mc X\rvert}\bm 1^Tw_z.
\]
Thus
\begin{align*}
    &\exp\bigl(D_2(Q_{f,m}W_{\mc T}\Vert P_{\mc X}W_{\mc T}\vert P_{\mc S})\bigr)\notag\\
	&=\frac{1}{\lvert\mc S\rvert}\sum_{s\in\mc S}\int\frac{\left(d_{\mc S}^{-1}(B_{f,m}w_z)(s)\right)^2}{\lvert\mc X\rvert^{-1}\bm 1^Tw_z}\,\mu(dz)\notag\\
	&=\frac{\lvert\mc X\rvert}{d_{\mc S}^2\lvert\mc S\rvert}\int\frac{\sum_{s\in\mc S}(B_{f,m}w_z(s))^2}{\bm 1^Tw_z}\,\mu(dz)\\
	&=\frac{\lvert\mc X\rvert}{d_{\mc S}^2\lvert\mc S\rvert}\int\frac{w_z^TB_{f,m}^TB_{f,m}w_z}{\bm 1^Tw_z}\,\mu(dz)\\
	&\stackrel{(a)}{=}\int\frac{w_z^TP_{f,m}w_z}{\bm 1^Tw_z}\,\mu(dz)\\
	&\stackrel{(b)}{\leq}\int\frac{\lvert\mc X\rvert^{-1}(\bm 1^Tw_z)^2+\lambda_2(f,m)w_z^Tw_z}{\bm 1^Tw_z}\,\mu(dz)\\
	&=1+\lambda_2(f,m)\frac{1}{\lvert\mc X\rvert}\sum_{x\in\mc X}\int\frac{w(z\vert x)^2}{\lvert\mc X\rvert^{-1}\sum_{x'\in\mc X}w(z\vert x')}\,\mu(dz)\\
	&=1+\lambda_2(f,m)\exp\bigl(D_2(W_{\mc T}\Vert P_{\mc X}W_{\mc T}\vert P_{\mc X})\bigr).
\end{align*}
Here, $(a)$ follows from $P_{f,m}=d_{\mc S}^{-1}d_{\mc X}^{-1}B_{f,m}^TB_{f,m}$, which was observed in the proof of Theorem \ref{thm:BRI-characterization}, and from \eqref{eq:doublecounting}. The inequality in $(b)$ is due to Lemma \ref{lem:stochmEV} below.
\end{proof}

\begin{lem}\label{lem:stochmEV}
    Let $P\in\mbb R^{\mc X\times\mc X}$ be a symmetric stochastic matrix. If $\lambda_2$ denotes the second-largest eigenvalue modulus of $P$, then 
    \[
        w^\top Pw\leq\lambda_2w^Tw+\frac{1}{\lvert\mc X\rvert}\left(\bm 1^Tw\right)^2
    \]
    for every $w\in\mbb R^{\mc X}$.    
\end{lem}

\begin{proof}
    This result is well known and can be found in, e.g., \cite{Bremaud}. A proof is included in Appendix \ref{app:proofs} for the sake of self-containedness.
\end{proof}

\subsection{Completion of the Proof}

To complete the proof of Theorem \ref{thm:EV-UB}, take $0<\varepsilon\leq 1-e^{-1}$ and choose any subset $\mc T$ of $\mc X\times\mc Z$ which satisfies \eqref{eq:subnsets}. It follows that
\begin{align*}
	&D\bigl(Q_{f,m}W\Vert P_{\mc X}W\vert P_{\mc S}\bigr)\\
	&\stackrel{(a)}{\leq}D\bigl(Q_{f,m}W_{\mc T}\Vert P_{\mc X}W_{\mc T}\vert P_{\mc S}\bigr)+\varepsilon \log\frac{\lvert\mc X\rvert}{d_{\mc S}}\\
	&\stackrel{(b)}{\leq}D_2(Q_{f,m}W_{\mc T}\Vert P_{\mc X}W_{\mc T}\vert P_{\mc S})-(1-\varepsilon)\log(1-\varepsilon)+\varepsilon \log\frac{\lvert\mc X\rvert}{d_{\mc S}}\\
	&\stackrel{(c)}{\leq}\log\left(1+\lambda_2(f,m)2^{D_2(W_{\mc T}\Vert P_{\mc X}W_{\mc T}\vert P_{\mc X})}\right)-(1-\varepsilon)\log(1-\varepsilon)+\varepsilon \log\frac{\lvert\mc X\rvert}{d_{\mc S}}\\
	&\stackrel{(d)}{\leq}\frac{1}{\ln 2}\lambda_2(f,m)2^{D_2(W_{\mc T}\Vert P_{\mc X}W_{\mc T}\vert P_{\mc X})}-(1-\varepsilon)\log(1-\varepsilon)+\varepsilon \log\frac{\lvert\mc X\rvert}{d_{\mc S}},
\end{align*}
where $(a)$ is due to Lemma \ref{lem:subnormalization}, $(b)$ is due to Lemma \ref{lem:RenyitoKL}, $(c)$ is due to Lemma \ref{lem:Renyi2byEV} and $(d)$ is due to the fact that $\log(1+t)\leq t/\ln 2$ for all nonnegative $t$. By minimizing over $\mc T$, one obtains the desired upper bound. This completes the proof of Theorem \ref{thm:EV-UB}.

\section{Proof of Theorem \ref{thm:RamGraphs_existence}}\label{subsect:thmRamGraphs}

The proof of Theorem \ref{thm:RamGraphs_existence} is based on a celebrated result by Marcus, Spielman and Srivastava \cite{MSS_interl_fams_i} about the existence of infinite families of biregular Ramanujan graphs for any given degree pair. They use $2$-lifts of graphs to iteratively construct large Ramanujan graphs from smaller ones. Our addition is the observation that one obtains two edge-disjoint Ramanujan graphs on a common vertex set in every step.

\paragraph{$2$-lifts of graphs}

For any graph $G$ with vertex set $\mc V(G)$ and edge set $\mc E(G)$, define a \textit{signing} to be a function $s:\mc E(G)\rightarrow\{-1,1\}$. We denote edges by their start and end vertices, and when we write $e=(x,y)\in\mc E$, then also $e=(y,x)$ since we only consider undirected graphs. In other words, $s(x,y)=s(y,x)$ for all vertex pairs $(x,y)\in\mc E(G)$. 

Given the signing $s$ one defines a graph $\hat G$ called the \textit{$2$-lift of $G$ associated to $s$} as follows: The vertex set of $\hat G$ consists of two disjoint copies $\mc V_0(G)$ and $\mc V_1(G)$ of $\mc V(G)$, so that every $x$ in $\mc V(G)$ corresponds to vertices $x_0,x_1$ in $\mc V(\hat G)$. For any edge $(x,y)\in\mc E(\mc G)$, the edge set $\mc E(\hat G)$ contains edges $(x_0,y_0)$ and $(x_1,y_1)$ if $s(x,y)=1$ and $(x_0,y_1)$ and $(x_1,y_0)$ if $s(x,y)=-1$. Observe that if $G$ is bipartite, then so is $\hat G$, and if $G$ is $(d_1,d_2)$-biregular, then $\hat G$ is $(d_1,d_2)$-biregular as well.

The \textit{signed adjacency matrix} of $G$ corresponding to the signing $s$ is the symmetric matrix $A_s$ with rows and columns indexed by the vertices of $G$, where the $(x,y)$ entry equals $s(x,y)$ if $(x,y)\in\mc E(G)$ and 0 else. Bilu and Linial derived the following result for signed adjacency matrices.

\begin{lem}[\cite{BL_lifts}, Lemma 3.1]\label{lem:lift-spectrum}
    Let $A$ be the adjacency matrix of a graph $G$ and $A_s$ the signed adjacency matrix associated with a $2$-lift $\hat G$. Then every eigenvalue of $A$ and every eigenvalue of $A_s$ are eigenvalues of $\hat G$. Furthermore, the multiplicity of each eigenvalue of $\hat G$ is the sum of its multiplicities in $A$ and $A_s$.
\end{lem}

An immediate consequence is the following lemma.

\begin{lem}\label{lem:disj_lifts}
    If $\hat G$ is the $2$-lift of a bipartite graph $G$ associated to the signing $s$, then the $2$-lift $\hat G_-$ of $G$ associated to the signing $-s$ has the same spectrum as $\hat G$ and is edge-disjoint from $\hat G$. 
\end{lem}

\begin{proof}
    Denote by $\hat A$ the adjacency matrix of $\hat G$. Since $\hat G$ is bipartite, the spectrum of $\hat A$ is symmetric about $0$ including multiplicities (see, e.g., \cite[Proposition 3.4.1]{BrouwerHaemers_SpectraofGraphs}). Since $G$ is bipartite, the spectrum of $A$ is also symmetric about 0. By Lemma \ref{lem:lift-spectrum}, the spectrum of $A_s$ must therefore be symmetric about 0 as well. This implies that $A_{-s}=-A_s$ has the same spectrum as $A_s$, and again by Lemma \ref{lem:lift-spectrum}, the adjacency matrix $\hat A_-$ of $\hat G_-$ has the same spectrum as $\hat A$.
    
    That $\hat G_-$ is edge-disjoint from $\hat G$ is obvious from the definition of $2$-lifts.
\end{proof}

The other ingredient to our construction is the following result due to Marcus, Spielman and Srivastava.

\begin{lem}[\cite{MSS_interl_fams_i}]\label{lem:Ramanujan2-lifts}
    For all $d_1,d_2\geq 3$ and every connected $(d_1,d_2)$-biregular Ramanujan graph $G$ there exists a signing $s$ such that the $2$-lift $\hat G$ of $G$ associated to $s$ is connected, $(d_1,d_2)$-biregular and Ramanujan as well.
\end{lem}

\begin{proof}
    By Theorems 5.3 and 5.6 of \cite{MSS_interl_fams_i}. The connectedness follows from the fact that $G$ is connected and that the eigenvalues of $\hat G$ which are not eigenvalues of $G$ are bounded by $\sqrt{d_{\mc S}-1}+\sqrt{d_{\mc X}-1}$, so that the eigenvalue $\sqrt{d_{\mc S}d_{\mc X}}$ still has multiplicity $1$. This is well-known to be equivalent to $\hat G$ being connected \cite[Proposition 1.3.6]{BrouwerHaemers_SpectraofGraphs}.
\end{proof}

\paragraph{Construction of graph family}

Since the vertex set will change in the construction, we notationally decouple the degrees from the vertex set and just call them $d_1,d_2$, where $d_1$ corresponds to $d_{\mc S}$ and $d_2$ to $d_{\mc X}$. We start the construction with the complete bipartite graph $G_0=\mc K_{\mc S_0,\mc X_0}$ on the disjoint union of sets $\mc S_0$ and $\mc X_0$ with $\lvert\mc S_0\rvert=d_2$ and $\lvert\mc X_0\rvert=d_1$. The adjacency matrix of $G_0$ has rank $2$ and nonzero eigenvalues $\pm\sqrt{d_1d_2}$. Therefore $G_0$ is Ramanujan.

Recursively for every $1\leq t\leq k$ and every sequence $\kappa_1,\ldots,\kappa_t\in\{-1,1\}^t$ we define a graph $G_{\kappa_1,\ldots,\kappa_t}$ as follows: For any $t\geq1$, given $\kappa_1,\ldots,\kappa_{t-1}$, we set $G_{\kappa_1,\ldots,\kappa_{t-1},1}$ to be any $2$-lift of $G_{\kappa_1,\ldots,\kappa_{t-1}}$ which is connected and Ramanujan. Its existence follows from Lemma \ref{lem:Ramanujan2-lifts}. If $G_{\kappa_1,\ldots,\kappa_{t-1},1}$ is the $2$-lift associated to the signing $s_t$ of $G_{\kappa_1,\ldots,\kappa_{t-1}}$, then $G_{\kappa_1,\ldots,\kappa_{t-1},-1}$ is defined to be the $2$-lift of $G_{\kappa_1,\ldots,\kappa_{t-1}}$ associated to the signing $-s_t$ of $G_{\kappa_1,\ldots,\kappa_{t-1}}$. By Lemma \ref{lem:disj_lifts}, $G_{\kappa_1,\ldots,\kappa_{t-1},-1}$ is connected and Ramanujan as well and edge-disjoint from $G_{\kappa_1,\ldots,\kappa_{t-1},1}$. Clearly, the common vertex set $\mc V_k$ of all $G_{\kappa_1,\ldots,\kappa_k}$ has a bipartition into a set of size $2^kd_1$ and one of size $2^kd_2$.

\begin{lem}
    Let $k\geq 1$ and let $(\kappa_1,\ldots,\kappa_k)\neq(\kappa_1',\ldots,\kappa_k')\in\{-1,1\}^k$. Then $G_{\kappa_1,\ldots,\kappa_k}$ and $G_{\kappa_1',\ldots,\kappa_k'}$ have disjoint edge sets.
\end{lem}

\begin{proof}
    We prove this by induction. The claim follows from Lemma \ref{lem:disj_lifts} for $k=1$. 
    
    Assume that $k>1$ and that the claim has been proven for every $1\leq t<k$. If $(\kappa_1,\ldots,\kappa_{k-1})=(\kappa_1',\ldots,\kappa_{k-1}')$, then the claim follows from Lemma \ref{lem:disj_lifts}. We may therefore assume that $(\kappa_1,\ldots,\kappa_{k-1})\neq(\kappa_1',\ldots,\kappa_{k-1}')$.
    
    For any element $x$ of the common vertex set $\mc V_k$ of $G_{\kappa_1,\ldots,\kappa_k}$ and $G_{\kappa_1',\ldots,\kappa_k'}$, denote by $\pi_k(x)$ the element of $\mc V_{k-1}$ of which $x$ is a copy. By the definition of $2$-lifts, two vertices $x$ and $y$ which are adjacent in $G_{\kappa_1,\ldots,\kappa_k}$ satisfy that $\pi_k(x)$ and $\pi_k(y)$ are adjacent in $G_{\kappa_1,\ldots,\kappa_{k-1}}$. However, the induction hypothesis and $(\kappa_1,\ldots,\kappa_{k-1})\neq(\kappa_1',\ldots,\kappa_{k-1}')$ imply that $\pi_k(x)$ and $\pi_k(y)$ are not adjacent in $G_{\kappa_1',\ldots\kappa_{k-1}'}$. Therefore $x$ and $y$ cannot be adjacent in $G_{\kappa_1',\ldots\kappa_k'}$. 
\end{proof}

It follows from the construction that the graphs $G_{\kappa_1,\ldots,\kappa_k}$ form an edge-disjoint decomposition of the complete bipartite graph. Thus the proof of Theorem \ref{thm:RamGraphs_existence} is complete.

\section{Proof of Theorem \ref{thm:BT-function}}\label{subsect:BT-main}

The proof of Theorem \ref{thm:BT-function} has two parts. In the first one, it is shown that $\beta$ is a biregular irreducible function with regularity set $\mc M$. The cardinality of $\mc M$ is determined in the second part. 

\subsection{$\beta$ is a biregular irreducible function}

Consider the bipartite graphs $G_{\beta,m}$ with bipartition $(\mc S,\mc X)$ for nonzero $m\in\mc N$. The symmetry of $\beta$ in $s$ and $x$ implies regularity of $G_{\beta,m}$ with $d_{\mc S}=d_{\mc X}=2^b$. If we denote the adjacency matrix of $G_{\beta,m}$ by $A_{\beta,m}$, then
\[
    A_{\beta,m}=
    \begin{bmatrix}
        0 & B_{\beta,m} \\
        B_{\beta,m} & 0
    \end{bmatrix}
\]
for a symmetric $\mc X\times\mc X$ matrix $B_{\beta,m}$. Define $P_{\beta,m}$ as in \eqref{eq:Pfm}. As in the proof of Theorem \ref{thm:BRI-characterization}, it follows that $P_{\beta,m}=d_{\mc X}^{-2}B_{\beta,m}^2=2^{-2b}B_{\beta,m}^2$. (Note that this follows from the symmetry of $\beta$ alone. More structure is not necessary to obtain this form.)

The analysis of the spectrum of $P_{\beta,m}$ can thus be reduced to that of $B_{\beta,m}$. Since $B_{\beta,m}$ is a symmetric matrix with entries equal to $0$ or $1$, it is the adjcency matrix of a graph $H_{\beta,m}$. Two vertices $x,x'$ of $H_{\beta,m}$ are adjacent if $\beta(x,x')=m$. $H_{\beta,m}$ may have loops, i.e., edges with the same start and end point (namely, if $\beta(x,x)=m$). It is regular due to the regularity of $G_{\beta,m}$. Thus the largest eigenvalue of $B_{\beta,m}$ is $2^b$. The multiplicity of this eigenvalue and the size of the other eigenvalues are determined in the next lemma, and now we concentrate on $m\in\mc M$.

\begin{lem}\label{lem:BT-graph}
    For every $m\in\mc M$, the largest eigenvalue $2^b$ of $B_{\beta,m}$ has multiplicity $1$, and the absolute value of every eigenvalue not equal to $2^b$ is upper-bounded by $k2^{b/2}/b$. 
\end{lem}

\begin{cor}\label{cor:BT-eigenvalues}
    For every $m\in\mc M$,
    \[
        \lambda_2(\beta,m)\leq\left(\frac{k}{b}\right)^22^{-b}.
    \]
\end{cor}

\begin{proof}[Proof of Corollary \ref{cor:BT-eigenvalues}]
    This follows immediately from Lemma \ref{lem:BT-graph} using $P_{\beta,m}=2^{-2b}B_{\beta,m}^2$.
\end{proof}

The corollary implies that $\beta$ is a biregular irreducible function with regularity set $\mc M$ and nearly optimal $\lambda_2(f,m)$ (compare with the Ramanujan biregular irreducible functions). It remains to establish Lemma \ref{lem:BT-graph}. The proof is based on the fact that $H_{\beta,m}$ is isomorphic to a special Cayley sum graph. A graph $H$ on the set $\{0,\ldots,n-1\}$ is called a \textit{Cayley sum graph} if there exists a subset $\mc D$ of $\{0,\ldots,n-1\}$ such that two numbers $x,y\in\{0,\ldots,n-1\}$ are adjacent in $H$ if and only if their sum modulo $n$ is contained in $\mc D$. 
    
Two vertices $s,x$ are adjacent in $H_{\beta,m}$ if $s\cdot x\in\mbb F_{2^b}+m$. Let $\alpha$ be a primitive element of $\mbb F_{2^\ell}$, i.e., $\alpha$ generates the multiplicative group $\mbb F_{2^\ell}^*$ of $\mbb F_{2^\ell}$. Such an $\alpha$ exists \cite[Theorem 2.8]{LidlNiederreiter}. Thus every nonzero element $x$ of $\mbb F_{2^\ell}$ can be written $x=\alpha^{a}$ for some unique $0\leq a\leq 2^\ell-2$. In particular, there exists a set $\mc D=\{d_1,\ldots,d_{2^b}\}$ such that $\mbb F_{2^b}+m=\{\alpha^{d_1},\ldots,\alpha^{d_{2^b}}\}$ (clearly, $v+m\neq 0$ for all $v\in\mbb F_{2^b}$ since $m\notin\mbb F_{2^b}$). Two elements $s=\alpha^{a_1}$ and $x=\alpha^{a_2}$ are adjacent in $H_{\beta,m}$ if and only if $a_1+a_2\in\mc D\pmod{2^\ell-1}$. Therefore $H_{\beta,m}$ is isomorphic to the Cayley sum graph on $\{0,\ldots,2^\ell-2\}$ determined by $\mc D$. The eigenvalues of $H_{\beta,m}$ are determined in the following general result on Cayley sum graphs which is due to Chung.

\begin{lem}[\cite{Chung_Diam_EV}, Lemma 2]\label{lem:Chung_EVs}
    Let $H$ be the Cayley sum graph on $\{0,\ldots,n-1\}$ determined by the set $\mc D=\{d_1,\ldots,d_k\}$. Then its largest eigenvalue equals $k$. The other eigenvalues have the form
    \[
        \pm\left\lvert\sum_{d\in\mc D}\theta^{d}\right\rvert,
    \]
    where $\theta$ ranges over the $n$-th complex unit roots $\theta\neq\pm 1$ with positive real part, and if $n$ is even, an additional eigenvalue is given by
    \begin{equation}\label{eq:Cayleyneven}
        \sum_{d\in\mc D}(-1)^{d}.
    \end{equation}
\end{lem}

Note that \eqref{eq:Cayleyneven} is not an eigenvalue in our case because $H_{\beta,m}$ is a graph with $2^\ell-1$ vertices. Graphs like $H_{\beta,m}$ were already considered by Chung in \cite{Chung_Diam_EV}, but explicitly so only with $\mbb F_{2^b}$ replaced by $\mbb F_p$ with $p$ prime and $\mbb F_{2^\ell}$ by $\mbb F_{p^\ell}$. We give the general argument for completeness. Chung used the following Lemma by Katz \cite{Katz}.

\begin{lem}[\cite{Katz}]\label{lem:Katz}
    Let $q$ be a prime power and $t$ a nonnegative integer. Let $\psi:\mbb F_{q^t}^*\rightarrow\mbb C$ be a nontrivial multiplicative character, i.e., a homomorphism from the multiplicative group $\mbb F_{q^t}^*$ to the unit circle $\{z\in\mbb C:\lvert z\rvert =1\}$ such that $\psi(x)\neq 1$ for some $x\in\mbb F_{q^t}^*$. Then for any $m$ with $\mbb F_q(m)=\mbb F_{q^t}$,
    \[
        \left\lvert\sum_{x\in\mbb F_q}\psi(x+m)\right\rvert\leq(t-1)\sqrt q.
    \]
\end{lem}

To apply Katz's lemma, let $m\in\mc M$, i.e., $\mbb F_{2^b}(m)=\mbb F_{2^\ell}$. The mapping $\psi:\mbb F_{2^\ell}^*\to\mbb C$ defined by $\psi(\alpha^a)=\theta^a$ is a nontrivial multiplicative character of $\mbb F_{2^\ell}^*$ for every $(2^\ell-1)$-th unit root $\theta\neq 1$. It follows that
\begin{align*}
    \sum_{d\in\mc D}\theta^d
    =\sum_{x'\in\mbb F_{2^b}+m}\psi(x')
    =\sum_{x\in\mbb F_{2^b}}\psi(x+m).
\end{align*}
We conclude that $B_{\beta,m}$ apart from the eigenvalue $2^b$ has $2^{\ell}-2$ eigenvalues which are upper-bounded by
\[
    \left(\frac{\ell}{b}-1\right)2^{b/2}=\frac{k2^{b/2}}{b}.
\]
The multiplicity of the eigenvalue $2^b$ is $1$. This proves Lemma \ref{lem:BT-graph} and completes the first part of the proof of Theorem \ref{thm:BT-function}.

\subsection{Cardinality of $\mc M$}

To complete the proof of Theorem \ref{thm:BT-function}, it remains to compute the cardinality of $\mc M$. An $m\in\mc N$ does not generate $\mbb F_{2^\ell}$ over $\mbb F_{2^b}$ (i.e., $\mbb F_{2^b}(m)\neq\mbb F_{2^\ell}$) if and only if it is contained in a strict subfield of $\mbb F_{2^\ell}$ containing $\mbb F_{2^b}$. By Lemma \ref{lem:subfields}, every such subfield equals $\mbb F_{2^t}$ for some multiple $t$ of $b$ which divides $\ell$. Therefore we need to compute
\begin{equation}\label{eq:card_non-generators}
    \left\lvert\bigcup_{t<\ell:b\mid t\mid\ell}(\mc N\cap\mbb F_{2^t})\right\rvert,
\end{equation}
where $a\mid b$ for positive integers $a,b$ means that $a$ divides $b$.

We denote the set of all multiples $t$ of $b$ which are strictly smaller than $\ell$ and divide $\ell$ by $\{t_1,\ldots,t_K\}$. We will need the following simple lemma.

\begin{lem}\label{lem:Ncapfieldinters}
    For any subset $\mc I$ of $\{1,\ldots,K\}$,
    \[
        \dim\left(\mc N\cap\bigcap_{i\in\mc I}\mbb F_{2^{t_i}}\right)
        =\dim\left(\bigcap_{i\in\mc I}\mbb F_{2^{t_i}}\right)-b.
    \]
\end{lem}

\begin{proof}
    Recall the formula for linear subspaces $\mc V_1,\mc V_2$
    \begin{equation}\label{eq:vector_inclusion-exclusion}
        \dim(\mc V_1+\mc V_2)=\dim(\mc V_1)+\dim(\mc V_2)-\dim(\mc V_1\cap\mc V_2).
    \end{equation}
    Then
    \begin{align*}
        \dim\left(\mc N\cap\bigcap_{i\in\mc I}\mbb F_{2^{t_i}}\right)
        &\stackrel{(a)}{=}\dim(\mc N)+\dim\left(\bigcap_{i\in\mc I}\mbb F_{2^{t_i}}\right)-\dim\left(\mc N+\bigcap_{i\in\mc I}\mbb F_{2^{t_i}}\right)\\
        &\stackrel{(b)}{=}k+\dim\left(\bigcap_{i\in\mc I}\mbb F_{2^{t_i}}\right)-\ell\\
        &=\dim\left(\bigcap_{i\in\mc I}\mbb F_{2^{t_i}}\right)-b,
    \end{align*}
    where $(a)$ follows from \eqref{eq:vector_inclusion-exclusion} and $(b)$ from the definition of $\mc N$ and the fact that $\bigcap_{i\in\mc I}\mbb F_{2^{t_i}}$ contains $\mbb F_{2^b}$, which implies that the sum of $\mc N$ and $\bigcap_{i\in\mc I}\mbb F_{2^{t_i}}$ equals $\mbb F_{2^\ell}$.
\end{proof}

The second lemma applied in the proof of \eqref{eq:M-bound} is a statement about those elements of $\mbb F_{2^\ell}$ which generate $\mbb F_{2^\ell}$ over $\mbb F_{2^b}$.

\begin{lem}\label{lem:field_generators}
    The number of elements $m$ of $\mbb F_{2^\ell}$ satisfying $\mbb F_{2^b}(m)=\mbb F_{2^\ell}$ equals $(\ell/b)N_{2^b}(\ell/b)$.
\end{lem}

\begin{proof}
    The elements of $\mbb F_{2^\ell}$ which generate $\mbb F_{2^\ell}$ over $\mbb F_{2^b}$ are exactly the zeros of the monic irreducible polynomials of degree $\ell/b$ with coefficients in $\mbb F_{2^b}$ \cite[Section 2.2]{LidlNiederreiter}. The zero sets of these polynomials are disjoint, and every monic irreducible polynomial of degree $d$ has exactly $d$ distinct zeros.  Therefore the number of generators of $\mbb F_{2^\ell}$ over $\mbb F_{2^b}$ equals $\ell/b$ times the number of monic irreducible polynomials of degree $\ell/b$ over $\mbb F_{2^b}$, which is given by $N_{2^b}(\ell/b)$ \cite[Theorem 3.25]{LidlNiederreiter}. 
\end{proof}

Now the calculation of \eqref{eq:card_non-generators} goes as follows:
\begin{align}
    \left\lvert\bigcup_{t<\ell:b\vert t\vert\ell}(\mc N\cap\mbb F_{2^t})\right\rvert\notag
    &\stackrel{(a)}{=}\sum_{k=1}^K(-1)^{k+1}\sum_{\substack{\mc I\subset\{1,\ldots,K\}:\\\lvert\mc I\rvert=k}}\left\lvert\bigcap_{i\in\mc I}(\mc N\cap\mbb F_{2^{t_i}})\right\rvert\notag\\
    &\stackrel{(b)}{=}2^{-b}\sum_{k=1}^K(-1)^{k+1}\sum_{\substack{\mc I\subset\{1,\ldots,K\}:\\\lvert\mc I\rvert=k}}\left\lvert\bigcap_{i\in\mc I}\mbb F_{2^{t_i}}\right\rvert\notag\\
    &\stackrel{(c)}{=}2^{-b}\left\lvert\bigcup_{t<\ell:b\mid t\mid\ell}\mbb F_{2^t}\right\rvert\notag\\
    &\stackrel{(d)}{=}2^{-b}\left(2^\ell-\frac{\ell}{b}N_{2^b}\left(\frac{\ell}{b}\right)\right)\notag\\
    &=2^k-\frac{\ell}{b}N_{2^b}\left(\frac{\ell}{b}\right)2^{-b},\label{eq:nonongenerators}
\end{align}
where $(a)$ follows from the inclusion-exclusion formula, $(b)$ is a consequence of Lemma \ref{lem:Ncapfieldinters}, $(c)$ again is the inclusion-exclusion formula, and $(d)$ follows from Lemma \ref{lem:field_generators}. Thus the cardinality of $\mc M$ equals $\lvert\mc N\rvert=2^k$ minus \eqref{eq:nonongenerators}, as claimed. The proof of Theorem \ref{thm:BT-function} is complete.

\section{Proof of Theorem \ref{thm:existence_opt_BRI}}\label{sect:existence_opt_BRI_proof}

Ramanujan biregular irreducible functions provide a class of biregular irreducible functions which can satisfy the conditions of Theorem \ref{thm:existence_opt_BRI} for arbitrary parameters. As noted before, seeded coset biregular irreducible functions are less flexible.

\subsection{Ramanujan biregular irreducible functions}

Choose 
\[
    k_n=\left\lfloor r(1-t)n\right\rfloor,\qquad
    d_n=\left\lfloor\exp\left(\frac{tk_n}{1-t}\right)\right\rfloor.
\]
Then 
\begin{equation}\label{eq:kn_vs_n}
    \lim_{n\to\infty}\frac{k_n}{n}=r(1-t),\qquad
    \lim_{n\to\infty}\frac{\log d_n}{k_n}=\frac{t}{1-t}.
\end{equation}

For every $n$ with $d_n\geq 3$, construct a Ramanujan biregular irreducible function with parameters $k_n$ and $d_{\mc S_n}=d_{\mc X_n}=d_n$ as in Corollary \ref{cor:RamanujanBRIs} (so the graphs $G_{f_n,m}$ are actually regular with exponentially increasing degree and $\lvert\mc S_n\rvert=\lvert\mc X_n\rvert$). Moreover, it follows from \eqref{eq:kn_vs_n} that
\begin{align*}
    &\lim_{n\to\infty}\frac{\log\lvert\mc X_n\rvert}{n}
    =\lim_{n\to\infty}\left(\frac{k_n}{n}+\frac{\log d_n}{k_n}\frac{k_n}{n}\right)=r(1-t)+\frac{t}{1-t}r(1-t)=r
\end{align*}
and 
\[
    \lim_{n\to\infty}\frac{\log\lvert\mc M_n\rvert}{\log\lvert\mc X_n\rvert}
    =\lim_{n\to\infty}\left(1+\frac{\log d_n}{k_n}\right)^{-1}
    =1-t.
\]
The statement on the asymptotic behavior of the eigenvalues follows from
\begin{align*}
    &\liminf_{n\to\infty}\frac{\min_m(-\log\lambda(f_n,m))}{\log\lvert\mc X_n\rvert}\\
    &\geq\liminf_{n\to\infty}\frac{2\log d_n-\log(d_n-1)-2}{k_n+\log d_n}\\
    &\geq\liminf_{n\to\infty}\left(1+\frac{k_n}{\log d_n}\right)^{-1}+\liminf_{n\to\infty}\frac{\log d_n-\log(d_n-1)-2}{k_n+\log d_n}\\
    &=t.
\end{align*}

\subsection{Seeded coset biregular irreducible functions}

Biregular irreducible functions constructed from seeded coset functions as in Section \ref{sect:BT-BRI} have limited rate flexibility. For $n\geq 1$, let $\beta_n:\mbb F_{2^{\ell_n}}^*\times\mbb F_{2^{\ell_n}}^*\to\mbb F_{2^{k_n}}$ be a seeded coset biregular irreducible function with regularity set $\mc M_n$ as in Theorem \ref{thm:BT-function}. It holds that $d_{\mc S_n}=d_{\mc X_n}=2^{b_n}$ with $b_n=\ell_n-k_n$, and $b_n$ divides $\ell_n$. The regularity of $\beta_n$ implies $\lvert\mc S_n\rvert=\lvert\mc X_n\rvert$. By Lemma \ref{lem:BT_qualifying_card}, the cardinality of the regularity set can be bounded as
\[
     k_n+\log\left(1-\frac{1}{2^{\ell_n/2-1}}\right)\leq\log\lvert\mc M_n\rvert\leq k_n.
\]
Since $\ell_n$ has to tend to infinity, it holds that
\begin{align}
    \liminf_{n\to\infty}\frac{\log\lvert\mc M_n\rvert}{\log\lvert\mc X_n\rvert}&=\liminf_{n\to\infty}\frac{k_n}{\ell_n},\label{eq:BT-req_rate}\\
    \liminf_{n\to\infty}\frac{\min_{m\in\mc M_n}(-\log\lambda_2(\beta_n,m))}{\log\lvert\mc X_n\rvert}&\geq\liminf_{n\to\infty}\frac{b_n}{\ell_n}.\label{eq:BT-req_sec}
\end{align}
Since $b_n$ divides $\ell_n$, the right-hand side of \eqref{eq:BT-req_sec} equals the inverse of a positive integer, say $N$, if the limit exists. Therefore the asymptotic rate \eqref{eq:BT-req_rate} has the form $1-1/N$. In particular, the asymptotic rates achieved by modular BRI schemes using seeded coset biregular irreducible functions as security components cannot back off from the rate of the error correcting code of the modular BRI scheme very much. This means that seeded coset biregular irreducible functions can only be used if the eavesdropper's channel is very noisy compared with the channel to the intended receiver.

\section{Proofs for Section \ref{subsect:asy_codes}}\label{sect:assec_proof}

\subsection{$\varepsilon$-smooth max-information}\label{subsect:maxinf}

To apply Theorem \ref{thm:EV-UB} to discrete and Gaussian wiretap channels, upper bounds on the respective $\varepsilon$-smooth conditional  R\'enyi 2-divergences are needed. Any subnormalized channel $\tilde W$ satisfies
\begin{align}
    \exp\bigl(D_2(\tilde W\Vert P_{\mc X}\tilde W\vert P_{\mc X})\bigr)
    =\int\frac{\sum_x\tilde w(z\vert x)^2}{\sum_{x'}\tilde w(z\vert x')}\,\mu(dz)
    &\leq\int\max_x\tilde w(z\vert x)\,\mu(dz).\label{eq:maxinf_def}
\end{align}
Tyagi and Vardy \cite{TV_UHF_preprint} denote the logarithm of the right-hand side of \eqref{eq:maxinf_def} by $I_\maxsub(\tilde W)$ and call it the \textit{max-information} of $\tilde W$. They also introduce the \textit{$\varepsilon$-smooth max-information} $I_\maxsub^\varepsilon(W)$ of a channel $W$ similar to the $\varepsilon$-smooth R\'enyi 2-divergence by
\[
    I_\maxsub^\varepsilon(W)=\inf_{\mc T}I_\maxsub(W_{\mc T}),
\]
where the minimum is taken over all sets satisfying \eqref{eq:subnsets}. Finally, Tyagi and Vardy bound the $\varepsilon$-smooth max-information both of discrete and Gaussian memoryless channels in \cite{TV_UHF_preprint} as described in the following.

\begin{lem}[\cite{TV_UHF_preprint}, Lemma 5]\label{lem:maxinf_discrete}
    Let $U:\mc A\rightarrow\mc Z$ be a discrete channel. For every $n$, let $\mc X_n$ be a finite set and $\phi_n:\mc X_n\to\mc A^n$ any function.
    \begin{enumerate}
        \item Assume there exists a $\delta>0$ and a probability distribution $P$ on $\mc A$ such that $\phi_n(x)\in T_{P,\delta}^n$ for all $n$ and all $x\in\mc X_n$. Then there exists a positive constant $c=c(\lvert\mc A\rvert,\lvert\mc Z\rvert)$ and a positive $\gamma_d=\gamma_d(\delta,\lvert\mc Z\rvert)$ which tends to $0$ as $\delta$ tends to $0$ such that
        \[
            \limsup_{n\to\infty}\frac{I_\maxsub^{\varepsilon_n}(\phi_n U^n)}{n}\leq I(P,U)+\gamma_d
        \]
        for $\varepsilon_n=2^{-nc\delta^2}$.
        \item If $\delta>0$ and the $\phi_n$ are arbitrary, then
        \[
            \limsup_{n\to\infty}\frac{I_\maxsub^{\varepsilon_n}(\phi_n U^n)}{n}\leq \max_P I(P,U)+\gamma_d
        \]
        where $P$ varies over the probability distributions on $\mc A$ and for the same $\varepsilon_n$ and $\gamma_d$ as in 1).
    \end{enumerate}
\end{lem}

The next lemma gives the upper bound on $I_\maxsub^\varepsilon(\phi_n U^n)$ if $U$ is Gaussian.

\begin{lem}[\cite{TV_UHF_preprint}, Lemma 6]\label{lem:maxinf_Gauss}
    Let $n$ be a positive integer, $\delta>0$ and set $\varepsilon_n=e^{-n\delta^2/8}$. If $U:\mbb R\rightarrow\mbb R$ is a Gaussian channel with noise variance $\sigma^2$ and $\phi_n:\mc X\rightarrow\mbb R^n$ is any function satisfying
    \[
        \lVert\phi_n(x)\rVert^2\leq n\Gamma
    \]
    for all $x\in\mc X$, then there exists a $\gamma_G=\gamma_G(\delta)$ such that
    \[
        \limsup_{n\to\infty}\frac{I_\maxsub^{\varepsilon_n}(\phi_n U^n)}{n}\leq\frac{1}{2}\log\left(1+\frac{\Gamma}{\sigma^2}\right)+\gamma_G.
    \]
\end{lem}

\subsection{Proofs of Lemmas \ref{lem:discrete_secfromBRI_types} and \ref{lem:Gaussian_secfromBRI}}

We only show the proof of Lemma \ref{lem:discrete_secfromBRI_types}. The proof of Lemma \ref{lem:Gaussian_secfromBRI} is analogous, one only has to use the upper bound from Lemma \ref{lem:maxinf_Gauss} instead of that from Lemma \ref{lem:maxinf_discrete}.

Let $(f_n)_{n=1}^\infty$ be a sequence of biregular irreducible functions as in the statement of the lemma. The result for arbitrary error-correcting codes easily follows from the statement for constant composition error-correcting codes, so it is sufficient to prove the latter. We thus take any discrete wiretap channel $(T:\mc A\to\mc Y,U:\mc A\to\mc Z)$ and a sequence of blocklength-$n$ codes $(\phi_n,\psi_n)$ for $T$ with message set $\tilde{\mc X}_n$ satisfying $\phi_n(\tilde{\mc X}_n)\subset T_{P,\delta_1}^n$ and, for some $\delta>0$,
\[
    \liminf_{n\to\infty}\frac{\log\lvert\tilde{\mc X}_n\rvert}{n}\geq r+\delta,
    \qquad\lim_{n\to\infty}e(\phi_n,\psi_n)=0.
\]
We also assume that $tr>I(P,U)+\gamma_d(\delta_1,\lvert\mc Z\rvert)$ for the $\gamma_d(\delta_1,\lvert\mc Z\rvert)$ defined in Lemma \ref{lem:maxinf_discrete}. For sufficiently large $n$, we may assume $\mc X_n\subset\tilde{\mc X}_n$. We show that the sequence $(\Pi(f_n,\phi_n,\psi_n))_{n=1}^\infty$ of modular BRI schemes has the claimed properties. That $e(\Pi(f_n,\phi_n,\psi_n))$ tends to $0$ is clear from the corresponding property of the error-correcting codes and \eqref{eq:maxerr_prefix}. The rate achieved by the seeded modular coding schemes satisfies
\begin{align*}
    \liminf_{i\to\infty}\frac{\log\lvert\mc M_n\rvert}{n}
    &=\liminf_{n\to\infty}\frac{\log\lvert\mc M_n\rvert}{\log\lvert\mc X_n\rvert}\frac{\log\lvert\mc X_n\rvert}{n}\\
    &\stackrel{(a)}{\geq}\left(\liminf_{n\to\infty}\frac{\log\lvert\mc M_n\rvert}{\log\lvert\mc X_n\rvert}\right)\left(\liminf_{n\to\infty}\frac{\log\lvert\mc X_n\rvert}{n}\right)\\
    &\geq(1-t)r,
\end{align*}
where $(a)$ is due to the positivity of the sequences. 
    
It remains to check whether semantic security is achieved. Write $W_n=\phi_n U^n$. For $\varepsilon_n=2^{-nc\delta_1^2}$, where $c=c(\lvert\mc A\rvert,\lvert\mc Z\rvert)$ is the constant from Lemma \ref{lem:maxinf_discrete}, and $m\in\mc M_n$
\begin{align}
    &\limsup_{n\to\infty}\frac{D_2^{\varepsilon_n}(W_n\Vert P_{\mc X_n}W_n\vert P_{\mc X_n})+\log\lambda_2(f_n,m)}{n}\notag\\
    &\stackrel{(b)}{\leq}\limsup_{n\to\infty}\frac{I_\maxsub^{\varepsilon_n}(W_n)+\log\lambda_2(f_n,m)}{n}\notag\\
    &\leq\limsup_{n\to\infty}\frac{I_\maxsub^{\varepsilon_n}(W_n)}{n}+\limsup_{n\to\infty}\frac{\log\lambda_2(f_n,m)}{n}\notag\\
    &\stackrel{(c)}{\leq}I(P,U)+\gamma_d-\liminf_{n\to\infty}\frac{-\log\lambda_2(f_n,m)}{\log\lvert\mc X_n\rvert}\frac{\log\lvert\mc X_i\rvert}{n}\notag\\
    &\stackrel{(d)}{\leq}I(P,U)-\left(\liminf_{n\to\infty}\frac{-\log\lambda_2(f_n,m)}{\log\lvert\mc X_n\rvert}\right)\left(\liminf_{n\to\infty}\frac{\log\lvert\mc X_n\rvert}{n}\right)+\gamma_d\notag\\
    &\leq I(P,U)-tr+\gamma_d\notag\\
    &<0,\notag
\end{align}
where $(b)$ is due to \eqref{eq:maxinf_def}, $(c)$ follows from Lemma \ref{lem:maxinf_discrete} and $(d)$ is possible because the involved sequences are positive. Therefore $\max_{m\in\mc M_n}\lambda_2(f_n,m)\exp(D_2^{\varepsilon_n}(W_n\Vert P_{\mc X_n}W_n\vert P_{\mc X_n}))$ tends to zero at exponential speed. Together with the exponential decrease of $\varepsilon_n$, it follows from Theorem \ref{thm:EV-UB} that $\max_{m\in\mc M_n}D(Q_{f_n,m}W_n\Vert P_{\mc X_n}W_n\vert P_{\mc S_n})$ tends to zero exponentially. Now, Corollary \ref{cor:sec_by_BRI} implies the exponential decrease of $L_\sem(\Pi(f_n,\phi_n,\psi_n))$.

\subsection{Proof of Corollary \ref{cor:assec}}

We first state some bounds on the performance of a single ordinary BRI wiretap code for the discrete or Gaussian memoryless wiretap channel $(T,U)$. Let $(\phi,\psi)$ be an error-correcting code for $T$ with message set $\mc X$, and assume that $f:\mc S\times\mc X\to\mc N$ is a biregular irreducible function with regularity set $\mc M$.

\begin{lem}\label{lem:BRI_WT_codes_single}
    The rate, error probability and semantic security information leakage of $R_N(f,\phi,\psi)$ satisfy
    \begin{align}
        \frac{\log\lvert\mc M\rvert^N}{(N+1)n}&=\frac{N}{N+1}\frac{\log\lvert\mc M\rvert}{n},\label{eq:seed_reuse_rate}\\
        e(R_N(f,\phi,\psi))&\leq e(\phi,\psi)+Ne(\Pi(f,\phi,\psi))\label{eq:seed_reuse_error},\\
        L_\sem(R_N(f,\phi,\psi))&\leq NL_\sem(\Pi(f,\phi,\psi)).\label{eq:seed_reuse_leakage}
    \end{align}
\end{lem}

\begin{proof}
    The validity of \eqref{eq:seed_reuse_rate} is obvious. \eqref{eq:seed_reuse_error} follows immediately from the union bound. To show \eqref{eq:seed_reuse_leakage}, choose any random variable $M^N=(M_1,\ldots,M_N)$ on $\mc M^N$ and let $Z^{N+1}=(Z_1,\ldots,Z_{N+1})$ be the eavesdropper's output generated by $M^N$ and $S$. Due to the independence of $S$ and $M^N$,
    \[
        I(M^N\wedge Z^{N+1})
        \leq I(M^N\wedge S,Z^{N+1})
        =I(M^N\wedge Z^{N+1}\vert S).
    \]
    It is thus sufficient to upper-bound the latter term. Denote by $\tilde\xi$ the encoder of $\Pi(f,\phi,\psi)$. For any seed $s$ and any message sequence $m^N=(m_1,\ldots,m_N)$, 
    \begin{align*}
        &p_{Z^{N+1}\vert S,M^N}(z^{N+1}\vert s,m^N)=\sum_{a_1\in\mc A_n'}u^n(z_1\vert a_1)\phi(a_1\vert s)\prod_{j=2}^{N+1}\sum_{a_j\in\mc A_n'}u^n(z_j\vert a_j)\tilde\xi(a_j\vert s,m_{j-1}).
    \end{align*}
    Therefore 
    \[
        H(Z^{N+1}\vert S,M^N)
        =H(Z_1\vert S)+\sum_{j=2}^{N+1}H(Z_j\vert S,M_{j-1}).
    \]
    By choosing $M_0$ to be any constant random variable, it follows that
    \begin{align}
        I(M^N\wedge Z^{N+1}\vert S)
        &=H(Z^{N+1}\vert S)-H(Z^{N+1}\vert M^N,S)\notag\\
        &\leq\sum_{j=1}^{N+1}\bigl(H(Z_j\vert S)-H(Z_j\vert M_{j-1},S)\bigr)\notag\\
        &=\sum_{j=1}^NI(M_j\wedge Z_{j+1}\vert S),\notag
    \end{align}
    where the inequality is due to the chain rule of entropy \cite[Corollary 3.4]{CK}. Now, maximization over the distribution of $M^N$ yields $L_\sem(R_N(f,\phi,\psi))\leq NL_\sem(\Pi(f,\phi,\psi))$. 
\end{proof}

\begin{rem}
    The proof of Lemma \ref{lem:BRI_WT_codes_single} shows that the structure of ordinary BRI wiretap codes allows for a tighter security analysis than ordinarily required for wiretap codes, since it implies that security is still given if the eavesdropper knows a certain part of the encoder's local randomness.
\end{rem}

We can now prove Corollary \ref{cor:assec}. We restrict our attention to discrete wiretap channels. For Gaussian wiretap channels, the proof is analogous. Let $(T:\mc A\to\mc Y,U:\mc A\to\mc Z)$ be an arbitrary discrete wiretap channel. Choose modular BRI schemes $\Pi(f_n,\phi_n,\psi_n)$ as in Lemma \ref{lem:discrete_secfromBRI_types} and let 
\[
    \varepsilon_n=\max\{e(\Pi(f_n,\phi_n,\psi_n)), L_\sem(\Pi(f_n,\phi_n,\psi_n))\}.
\]
By assumption, the $\varepsilon_n$ tend to zero. Take any sequence $(N_n)_{n=1}^\infty$ satisfying $N_n\to\infty$ and $N_n\varepsilon_n\to 0$ and consider the ordinary BRI wiretap codes $R_{N_n}(f_n,\phi_n,\psi_n)$. By \eqref{eq:seed_reuse_rate} and since $N_n\to\infty$, the asymptotic rate achieved by the $R_{N_n}(f_n,\phi_n,\psi_n)$ equals the rate achieved by the $\Pi(f_n,\phi_n,\psi_n)$, i.e., 
\[
    \liminf_{n\to\infty}\frac{\log\lvert\mc M_n\rvert^{N_n}}{(N_n+1)n}=\liminf_{n\to\infty}\frac{\log\lvert\mc M_n\rvert}{n}.
\]
The asymptotic error probability and semantic security information leakage of $R_{N_n}(f_n,\phi_n,\psi_n)$ by \eqref{eq:seed_reuse_error} and \eqref{eq:seed_reuse_leakage} satisfy
\begin{align*}
    \limsup_{n\to\infty}e(R_{N_n}(f_n,\phi_n,\psi_n))
    \leq \limsup_{n\to\infty}e(\phi_n,\psi_n)+N_n\limsup_{n\to\infty}e(\Pi(f_n,\phi_n,\psi_n))
    &\leq (N_n+1)\varepsilon_n
\end{align*}
and 
\begin{align*}
    \limsup_{n\to\infty} L_\sem(R_{N_n}(f_n,\phi_n,\psi_n))
    &\leq N_nL_\sem(\Pi(f_n,\phi_n,\psi_n))\leq N_n\varepsilon_n.
\end{align*}
Since $N_n\varepsilon_n\to 0$, the upper bounds both of the error and the semantic security leakage incurred by the $R_{N_n}(f_n,\phi_n,\psi_n)$ vanish asymptotically. Hence the sequence of wiretap codes $(R_{N_n}(f_n,\phi_n,\psi_n))_{n=1}^\infty$ achieves the same semantic security rate as the sequence of modular BRI schemes $(\Pi(f_n,\phi_n,\psi_n))_{n=1}^\infty$.

\section{Proofs for Section \ref{subsect:further_analysis}}\label{sect:EV-qualset-sum_proof}

\subsection{Proof of Lemma \ref{lem:EV-qualset-sum}}

Choose a sequence of biregular irreducible functions $f_i$ satisfying \eqref{eq:BRI_rate} and \eqref{eq:BRI_minlb} with parameter $t$. We have to show that
\[
    \limsup_{i\to\infty}\frac{\min_{m\in\mc M_i}(-\log\lambda_2(f_i,m))}{\log\lvert\mc X_i\rvert}\leq t.
\]
Assume this were not true. By passing to a subsequence if necessary, we can without loss of generality assume that there exists a $0<\delta<1-t$ such that
\begin{equation}\label{eq:contrad_assumption}
    \lim_{i\to\infty}\frac{\min_{m\in\mc M_i}(-\log\lambda_2(f_i,m))}{\log\lvert\mc X_i\rvert}> t+\delta.
\end{equation}
We will show that this leads to a contradiction when applying the functions from this sequence as security components of modular BRI schemes.

Let $\mc A=\{0,1\}$ and define $T:\mc A\to\mc A$ to be the noiseless binary channel, where $T(a\vert a)=1$ for $a\in\mc A$. Further, choose any $p$ such that
\[
    1-t-\delta<h(p)<1-t,
\]
where $h(p)=-p\log p-(1-p)\log(1-p)$ is the binary entropy of $p$. Define the channel $U:\mc A\to\mc A$ to be the binary symmetric channel with flipping probability $p$, i.e., 
\[
    U(a\vert a)=1-p,\qquad U(1-a\vert a)=p
\]
for all $a\in\mc A$. By \cite{Wyner}, the secrecy capacity of $(T,U)$ is given by
\[
    \max_P\bigl(I(P,T)-I(P,U)\bigr)=h(p),
\]
where the maximum on the left-hand side is over probability distributions on $\mc A$ and the uniform distribution $P_{\mc A}$ is a maximizer. This rate cannot be exceeded by sequences of wiretap codes whose blocklengths are only a subsequence of the positive integers, as one easily sees from the converse part of the proof of the coding theorem for the discrete memoryless wiretap channel.
    
To come to a contradiction, we will proceed as in the proof of Lemma \ref{lem:discrete_secfromBRI_types}. Choose $n_i$ such that
\[
    2^{n_i-1}<\lvert\mc X_i\rvert\leq 2^{n_i}.
\]
Then $\mc X_i$ can be considered to be a subset of $\mc A^{n_i}$. Since $T$ is a noiseless channel, no error correction is necessary, so we obtain a blocklength-$n_i$ error-correcting code $(\phi_{n_i},\psi_{n_i})$ with $e(\phi_{n_i},\psi_{n_i})=0$ by taking both $\phi_{n_i}$ and $\psi_{n_i}$ to be the identity mappings on $\mc X_i$. In particular, the modular BRI scheme $\Pi(f_i,\phi_{n_i},\psi_{n_i})$ asymptotically has rate $1-t>h(p)$ and its error probability equals zero for every $i$. Moreover, the capacity of $U$ equals $\max_PI(P,U)=1-h(p)$ (see \cite{CK}). From this and \eqref{eq:contrad_assumption} it follows like in the proof of Lemma \ref{lem:discrete_secfromBRI_types} that $L_\sem(\Pi(f_i,\phi_{n_i},\psi_{n_i}))$ tends to zero.

Thus one obtains a sequence of wiretap codes which asymptotically achieves a semantic security rate strictly larger than $h(p)$ on $(T,U)$. This contradicts the fact that $h(p)$ is the secrecy capacity of $(T,U)$. Therefore the assumption that a subsequence satisfying \eqref{eq:contrad_assumption} for any $\delta>0$ exists must be wrong, and this completes the proof.

\subsection{Proof of Lemma \ref{lem:nearly_Ram}}

    Recall that $\lambda_2(G_{f,m})=\sqrt{d_{\mc S}d_{\mc X}\lambda_2(f,m)}$ for every biregular irreducible function. Thus
    \begin{align}
        &\limsup_{i\to\infty}\frac{\max_m\log\lambda_2(G_{f_i,m})}{\log\lvert\mc X_i\rvert}\notag\\
        &=\limsup_{i\to\infty}\frac{\log\sqrt{d_{\mc S_i}d_{\mc X_i}}+\max_m\log\sqrt{\lambda_2(f_i,m)}}{\log\lvert\mc X_i\rvert}\notag\\
        &\leq\limsup_{i\to\infty}\frac{\log\sqrt{d_{\mc S_i}d_{\mc X_i}}}{\log\lvert\mc X_i\rvert}+\frac{1}{2}\lim_{i\to\infty}\frac{\max_m\log\lambda_2(f_i,m)}{\log\lvert\mc X_i\rvert},\label{eq:limsups}
    \end{align}
    where we used the existence of the limit proved in Lemma \ref{lem:EV-qualset-sum}. The second summand equals
    \begin{align*}
        -\frac{t}{2}&=\frac{1}{2}\bigl((1-t)-1\bigr)\\
        &\stackrel{(a)}{=}\frac{1}{2}\left(\lim_{i\to\infty}\frac{\log\lvert\mc M_i\rvert}{\log\lvert\mc X_i\rvert}-1\right)\\
        &\stackrel{(b)}{\leq}\liminf_{i\to\infty}\frac{-\log\sqrt{d_{\mc S_i}}}{\log\lvert\mc X_i\rvert}\\
        &=-\limsup_{i\to\infty}\frac{\log\sqrt{d_{\mc S_i}}}{\log\lvert\mc X_i\rvert},
    \end{align*}
    where $(a)$ is due to \eqref{eq:BRI_rate} and $(b)$ follows from \eqref{eq:Mcard-ub}. Thus \eqref{eq:limsups} is at most
    \begin{align*}
        \limsup_{i\to\infty}\frac{\log\sqrt{d_{\mc X_i}}}{\log\lvert\mc X_i\rvert},
    \end{align*}
    as claimed.
    
    If the limits \eqref{eq:ex_lim_d_S_X} exist, then the inequalities above become equalities and
    \begin{equation}\label{eq:limit_lambda}
        \lim_{i\to\infty}\frac{\max_m\log\lambda_2(G_{f_i,m})}{\log\lvert\mc X_i\rvert}
        =\lim_{i\to\infty}\frac{\log\sqrt{d_{\mc X_i}}}{\log\lvert\mc X_i\rvert}
    \end{equation}
    exists as well. 
    This completes the proof of Lemma \ref{lem:nearly_Ram}.

\begin{rem}\label{rem:imbalanced_Feng-Li}
    Note that if the limits \eqref{eq:ex_lim_d_S_X} exist, then the limit
    \[
        s:=\lim_{i\to\infty}\frac{\log\lvert\mc S_i\rvert}{\log\lvert\mc X_i\rvert}
    \]
    exists as well due to \eqref{eq:doublecounting}. Assume that $s\leq 1$, as it is the case for the biregular irreducible functions constructed in Theorem \ref{thm:existence_opt_BRI}. For sufficiently large $i$, this implies 
    \begin{align*}
        &\log\lvert\mc X_i\rvert
        \leq\log(\lvert\mc S_i\rvert+\lvert\mc X_i\rvert)
        =\log\lvert\mc X_i\rvert+\log\left(1+\frac{\lvert\mc S_i\rvert}{\lvert\mc X_i\rvert}\right)
        \leq 2+\log\lvert\mc X_i\rvert
    \end{align*}
    and 
    \begin{align*}
        \lim_{i\to\infty}\frac{\log d_{\mc X_i}}{\log\lvert\mc X_i\rvert}
        &\leq\lim_{i\to\infty}\frac{\log d_{\mc S_i}+\log\lvert\mc S_i\rvert-\log\lvert\mc X_i\rvert}{\log\lvert\mc X_i\rvert}
        \leq\lim_{i\to\infty}\frac{\log d_{\mc S_i}}{\log\lvert\mc X_i\rvert}.
    \end{align*}
    By this and \eqref{eq:limit_lambda}, we obtain the symmetric form
    \begin{align*}
        \lim_{i\to\infty}\frac{\max_m\log\lambda_2(G_{f_i,m})}{\log(\lvert\mc S_i\rvert+\lvert\mc X_i\rvert)}
        =\lim_{i\to\infty}\frac{\log\sqrt{d_{\mc X_i}}}{\log(\lvert\mc S_i\rvert+\lvert\mc X_i\rvert)}
        &=\lim_{i\to\infty}\frac{\log\min(\sqrt{d_{\mc S_i}},\sqrt{d_{\mc X_i}})}{\log(\lvert\mc S_i\rvert+\lvert\mc X_i\rvert)}.
    \end{align*}
    By \eqref{eq:doublecounting} it holds that 
    \begin{align*}
        \lim_{i\to\infty}\frac{\log d_{\mc S_i}}{\log(\lvert\mc S_i\rvert+\lvert\mc X_i\rvert)}
        =1+\lim_{i\to\infty}\left(\frac{\log d_{\mc X_i}}{\log\lvert\mc X_i\rvert}-\frac{\log\lvert\mc S_i\rvert}{\log\lvert\mc X_i\rvert}\right)
        &=1-s+\lim_{i\to\infty}\frac{\log d_{\mc X_i}}{\log\lvert\mc X_i\rvert}.
    \end{align*}
    Thus the difference between 
    \[
        \lim_{i\to\infty}\frac{\log\min(\sqrt{d_{\mc S_i}},\sqrt{d_{\mc X_i}})}{\log(\lvert\mc S_i\rvert+\lvert\mc X_i\rvert)}
    \quad\text{and}\quad
        \lim_{i\to\infty}\frac{\log\max(\sqrt{d_{\mc S_i}},\sqrt{d_{\mc X_i}})}{\log(\lvert\mc S_i\rvert+\lvert\mc X_i\rvert)}
    \]
    equals $(1-s)/2$.
\end{rem}

\subsection{Proof of Lemma \ref{lem:degree_expgrowth}}

We need the precise form of the Feng-Li bound.

\begin{lem}[\cite{FengLi}]\label{lem:FengLi}
    If $G$ is a $(d_{\mc S},d_{\mc X})$-biregular graph with diameter $\Delta\geq8$, then the second-largest eigenvalue $\lambda_2(G)$ of $G$ satisfies
    \[
        \lambda_2(G)^2\geq d_{\mc S}+d_{\mc X}-2+2\sqrt{(d_{\mc S}-1)(d_{\mc X}-1)}\left(1-\frac{1}{\Delta-1}\right).
    \]
\end{lem}

If $G$ is a connected $(d_{\mc S},d_{\mc X})$-biregular graph with $d_{\mc S},d_{\mc X}\geq 2$ and bipartition $(\mc S,\mc X)$, then it is well-known that its diameter $\Delta$ satisfies
\begin{equation}\label{eq:diameter}
    \Delta\geq\frac{\log(\lvert\mc X\rvert+\lvert\mc S\rvert)}{\log(d_{\mc S})+\log(d_{\mc X})}.
\end{equation}
This can be seen as follows: Starting from any vertex $x$ in $\mc X$, say, every other vertex of $G$ can be reached by a path starting in $x$. Due to the bipartiteness of $G$, an upper bound on the number of vertices which can be reached from $x$ in $n$ steps is 
\[
    \begin{cases}
        d_{\mc S}^{n/2}d_{\mc X}^{n/2}&\text{if $n$ even},\\
        d_{\mc S}^{(n-1)/2}d_{\mc X}^{(n+1)/2}&\text{if $n$ odd.}
    \end{cases}
\]
The expression if one starts in $s\in\mc S$ is analogous. Therefore the total number of vertices of the graph is at most $(d_{\mc S}d_{\mc X})^\Delta$. This gives the rough lower bound \eqref{eq:diameter} for the diameter of $G$. 

We can now start with the main part of the proof of Lemma \ref{lem:degree_expgrowth}. Let $f_i$ be defined by the family $(G_{f_i,m})_{m\in\mc M_i}$. We first note that $d_{\mc S_i}$ and $d_{\mc X_i}$ must be at least 2 for $i$ sufficiently large. Otherwise, say if $d_{\mc X_i}=1$, then $\lvert\mc S_i\rvert=1$ due to the connectedness of $G_{f_i,m}$. Thus all $G_{f_i,m}$ coincide for $m\in\mc M_i$, implying $\lvert\mc M_i\rvert=1$, in contradiction to the assumption that $\lvert\mc M_i\rvert$ tends to infinity.
    
Let $\Delta_{f_i,m}$ be the diameter of $G_{f_i,m}$. If $\Delta_{f_i,m}\leq7$, then 
\begin{align*}
    \log\max(d_{\mc S_i},d_{\mc X_i})
    &\geq\frac{\log d_{\mc S_i}+\log d_{\mc X_i}}{2}\\
    &\stackrel{(a)}{\geq}\frac{\log(\lvert\mc S_i\rvert+\lvert\mc X_i\rvert)}{14}\\
    &\stackrel{(b)}{\geq}\frac{\log(d_{\mc S_i}+d_{\mc X_i})+\log\lvert\mc M_i\rvert}{14}\\
    &\geq\frac{\log\max(d_{\mc S_i},d_{\mc X_i})+\log\lvert\mc M_i\rvert}{14},
\end{align*}
where $(a)$ is due to the assumption $\Delta_{f_i,m}\leq 8$ and \eqref{eq:diameter} and $(b)$ comes from \eqref{eq:Mcard-ub}. We conclude that
\begin{equation}\label{eq:diamleq7}
    \log\max(d_{\mc S_i},d_{\mc X_i})\geq\frac{\log\lvert\mc M_i\rvert}{13}.
\end{equation}

Otherwise, if $\Delta_{f_i,m}\geq 8$, then
\begin{align*}
    \lambda_2(f_i,m)
    &\stackrel{(c)}{\geq}\frac{d_{\mc S_i}+d_{\mc X_i}-2+2\sqrt{(d_{\mc S_i}-1)(d_{\mc X_i}-1)}(1-\frac{1}{\Delta_{f_i,m}-1})}{d_{\mc S_i}d_{\mc X_i}}\notag\\
    &\geq\frac{1}{d_{\mc X_i}}\left(1-\frac{1}{d_{\mc S_i}}\right)+\frac{1}{d_{\mc S_i}}\left(1-\frac{1}{d_{\mc X_i}}\right)\notag\\
    &\stackrel{(d)}{\geq}\frac{1}{\max(d_{\mc S_i},d_{\mc X_i})},
\end{align*}
where $(c)$ is due to Lemma \ref{lem:FengLi} and $(d)$ is due to $d_{\mc S_i},d_{\mc X_i}\geq 2$. One thus obtains 
\begin{equation}\label{eq:diamgeq8}
    \log\max(d_{\mc S_i},d_{\mc X_i})
    \geq\min_{m\in\mc M_i}(-\log\lambda_2(f_i,m)).
\end{equation}

In summary, it follows from \eqref{eq:diamleq7} and \eqref{eq:diamgeq8} together with \eqref{eq:BRI_rate} and \eqref{eq:BRI_minlb} that
\begin{equation}\label{eq:lb_deggrowth_proof}
    \liminf_{i\to\infty}\frac{\log\max(d_{\mc S_i},d_{\mc X_i})}{\log\lvert\mc X_i\rvert}
    \geq\min\left\{\frac{1-t}{13},t\right\},
\end{equation}
as claimed. 

We finally show that the right-hand side of \eqref{eq:lb_deggrowth_proof} can be replaced by $t$ if one of the two conditions in the statement of the lemma holds. For the first one, the claim is immediate from above. If the second condition holds, then 
\[
    1-\frac{\log\lvert\mc M_i\rvert}{\log\lvert\mc X_i\rvert}
    =1-\frac{\log(\lvert\mc X_i\rvert/d_{\mc S_i})}{\log\lvert\mc X_i\rvert}
    =\frac{\log d_{\mc S_i}}{\log\lvert\mc X_i\rvert},
\]
where the first equality is due to \eqref{eq:Mcard-ub}. By \eqref{eq:BRI_rate}, the limit of the left-hand side is $t$. This completes the proof.

\section{Conclusion}\label{sect:conclusion}

\subsection{Summary}

Biregular irreducible functions, modular BRI schemes and ordinary BRI wiretap codes are introduced. A bound on the semantic security information leakage of modular BRI schemes is derived which clearly separates the effects of the biregular irreducible function and of the concatenation of encoder and channel. A characterization of biregular irreducible functions in terms of edge-disjoint connected biregular graphs is derived. This characterization is applied to construct Ramanujan biregular irreducible functions via an edge-disjoint decomposition of the complete bipartite graph into biregular Ramanujan subgraphs. It is shown that the unconstrained seeded coset function, which is a universal hash function frequently used in modular UHF schemes, can be interpreted as a biregular irreducible function for suitable parameters. 

To test the performance of biregular irreducible functions, optimal sequences of biregular irreducible functions are constructed. Using these sequences, it is shown that the secrecy capacities of discrete and Gaussian wiretap channels are achievable by modular BRI schemes and ordinary BRI wiretap codes. Hence the separation of error correction and security generation is optimal for these channels. By Theorem \ref{thm:EV-UB}, it only depends on the $\varepsilon$-smooth conditional R\'enyi $2$-divergence of a channel with respect to the uniform input distribution by how much a given biregular irreducible function can decrease the information leakage through this channel. Thus modular BRI schemes provide some robustness with respect to the channel. In contrast, the polar wiretap code constructed by Liu, Yan and Ling \cite{LYL_polar_semsec} is tailored to the Gaussian wiretap channel. 

Asymptotically, every graph in an optimal sequence of biregular irreducible functions is nearly Ramanujan. The maximum degree of the corresponding graph families must grow exponentially in the blocklength. The analysis of biregular irreducible functions also shows that they are universal hash functions on average. 

\subsection{Practical aspects of biregular irreducible functions}

\paragraph{General}

For the application of modular BRI schemes and ordinary BRI wiretap codes, their efficiency is relevant. By efficiency, we mean that encoding and decoding can be done in time polynomial in the blocklength. Equivalently, since we are interested in codes achieving a positive rate, one can consider coding to be efficient if it can be done in time which is polynomial in the ``length'' of the message, i.e., in the logarithm of the cardinality of the message set. A seeded modular coding scheme is efficent if its components are, as already observed by Bellare and Tessaro \cite{BTV_cryptWiretap} in the case of modular UHF schemes. Since we do not have any efficiently computable biregular irreducible functions with positive rate, we will not go deeply into the discussion of the efficiency of modular BRI schemes or ordinary BRI wiretap codes. Some aspects of the complexity of biregular irreducible functions are mentioned below.

One additional point we would like to mention is that the seed reuse performed in a ordinary BRI wiretap code does not alter the complexity class compared with the underlying modular BRI schemes. On the one hand, the complexity of encoding and decoding $R_N(f,\phi,\psi)$ is less than $N+1$ times the respective complexities of $\Pi(f,\phi,\psi)$, but on the other hand, the blocklength of the ordinary BRI wiretap code also equals $N+1$ times the blocklength of the modular BRI scheme, and the message length increases by $N$. 

\paragraph{Biregular irreducible functions}

For a biregular irreducible function $f:\mc S\times\mc X\to\mc N$ with regularity set $\mc M$, efficiency does not only mean the efficient computability of $f(s,x)$ given $s$ and $x$, but also efficient invertibility, i.e., the efficient realization of the uniform distribution on the set $\{x:f(s,x)=m\}$ for any message $m$ and seed $s$. If $f$ is defined by a graph family $(G_{f,m})_{m\in\mc N}$, then the computation of $f(s,x)$ requires determining in which of the exponentially many graphs $G_{f,m}$ the arguments $s$ and $x$ are adjacent. Efficiency of this process would mean that this is possible in time polynomial in $\log\lvert\mc X\rvert$. This is harder than computing $f_s^{-1}(\cdot\vert m)$, since in this situation $G_{f,m}$ is determined by the message $m$.

Ramanujan biregular irreducible functions so far cannot be constructed efficiently. Known explicit constructions of Ramanujan graphs like Cohen's \cite{Coh_Constr} do not seem to generalize to a decomposition of the complete bipartite graph into Ramanujan graphs. One should expect the construction of such families to get easier if one backs off a little bit from the best possible security rates and looks for good edge-disjoint families of non-Ramanujan expanders. Some methods of constructing expanders are presented in \cite{HLW_expanders}. 

The complexity of the unconstrained seeded coset function $\beta^o$ was discussed in \cite{TV_UHF_preprint}. It can be computed efficiently for a large variety of parameters. In contrast, the seeded coset biregular irreducible function cannot be computed efficiently. This is due to the nontrivial interplay of vector space and field operations on $\mbb F_{2^\ell}$ needed in its computation and its randomized inversion. 

A family of efficient security components which can be employed in seeded modular coding schemes and which achieves the semantic security of discrete and Gaussian memoryless wiretap channels (given suitable error-correcting codes) is given by Hayashi and Matsumoto \cite{HayMat_Multiplex}, see Appendix \ref{app:HayMat}. The disadvantage of their example is that the required seed length is roughly twice as long as that of the biregular irreducible functions constructed above. Compared with ordinary BRI wiretap codes, this leads to a worse finite-blocklength error and security performance of the corresponding codes without common randomness obtained by seed reuse.

\appendices

\section{Proofs of Lemmas \ref{lem:nonnorm_div_ineq} and \ref{lem:stochmEV}}\label{app:proofs}

\begin{proof}[Proof of Lemma \ref{lem:nonnorm_div_ineq}]
Let $m_1$ and $m_2$ be the $\mu$-densities of $M_1$ and $M_2$, respectively. If $\mu(m_1>0,m_2=0)>0$, then both sides of the inequality equal $\infty$. Otherwise,
\begin{align*}
	D(M_1\Vert M_2)
	&=Z_1\left(D\left(\left.\frac{M_1}{Z_1}\right\Vert\frac{M_2}{Z_2}\right)+\log\frac{Z_1}{Z_2}\right)\\
	&\leq Z_1\left(D_2\left(\left.\frac{M_1}{Z_1}\right\Vert\frac{M_2}{Z_2}\right)+\log\frac{Z_1}{Z_2}\right)\\
	&=Z_1\left(D_2(M_1\Vert M_2)-2\log Z_1+\log Z_2+\log Z_1-\log Z_2\right)\\
	&=Z_1\left(D_2(M_1\Vert M_2)-\log Z_1\right),
\end{align*}
    where the inequality is due to the fact that $D(\cdot\Vert\cdot)\leq D_2(\cdot\Vert\cdot)$ for probability densities \cite{vEHRenyiDiv}.
\end{proof}

\begin{proof}[Proof of Lemma \ref{lem:stochmEV}]
The all-one vector $\bm 1$ is an eigenvector to the eigenvalue $1$ of $P$, in other words,
\begin{equation}\label{eq:P1=1}
    P\bm 1=\bm 1.
\end{equation}
We define the scalar product $\langle \cdot,\cdot\rangle_{\mc X}$ on $\mbb R^{\mc X}$ by
\[
	\langle u,v\rangle_{\mc X}=\frac{1}{\lvert\mc X\vert}u^Tv.
\]
The norm induced by $\langle\cdot,\cdot\rangle_{\mc X}$ is denoted by $\lVert\cdot\rVert_{\mc X}$, in particular, $\langle w, w\rangle_{\mc X}=\lVert w\rVert^2_{\mc X}$. Note that
\begin{equation}\label{eq:X-scalp_normed}
    \lVert\bm 1\rVert_{\mc X}=1.
\end{equation}
For any $w\in\mbb R^{\mc X}$ write
\[
	\overline w=\frac{1}{\lvert\mc X\rvert}\bm 1^Tw=\langle w,\bm 1\rangle_{\mc X}.
\]
Then
\begin{align*}
	w^\top Pw
	&=\frac{\lvert\mc X\rvert}{\lvert\mc X\rvert}\sum_x(Pw)(x)w(x)\\
	&=\lvert\mc X\rvert\langle Pw,w\rangle_{\mc X}\\
	&=\lvert\mc X\rvert\Bigl[\langle P(w-\overline w\bm 1),w-\overline w\bm 1\rangle_{\mc X}+\overline w\langle Pw,\bm 1\rangle_{\mc X}+\overline w\langle P\bm 1,w\rangle_{\mc X}-\overline w^2\langle\bm 1,\bm 1\rangle_{\mc X}\Bigr]\\
	&\stackrel{(a)}{\leq}\lvert\mc X\rvert\Bigl[\lambda_2\lVert w-\overline w\bm 1\rVert^2_{\mc X}+\overline w\langle w,P\bm 1\rangle_{\mc X}+\overline w\langle\bm 1,w\rangle_{\mc X}-\overline w^2\Bigr]\\
	&\stackrel{(b)}{=}\lvert\mc X\rvert\Bigl[\lambda_2\lVert w\rVert^2_{\mc X}-2\lambda_2\overline w\langle w,\bm 1\rangle_{\mc X}+\lambda_2\overline w^2\langle\bm 1,\bm 1\rangle_{\mc X}+\overline w\langle w,\bm 1\rangle_{\mc X}+\overline w^2-\overline w^2\Bigr]\\
	&\stackrel{(c)}{=}\lvert\mc X\rvert\Bigl[\lambda_2\lVert w\rVert^2_{\mc X}-2\lambda_2\overline w^2+\lambda_2\overline w^2+\overline w^2\Bigr]\\
	&=\lvert\mc X\rvert\Bigl[\lambda_2\lVert w\rVert^2_{\mc X}+(1-\lambda_2)\overline w^2\Bigr]\\
	&=\lambda_2w^Tw+(1-\lambda_2)\lvert\mc X\rvert\left(\frac{\bm 1^Tw}{\lvert\mc X\rvert}\right)^2\\
	&\leq\lambda_2w^Tw+\frac{1}{\lvert\mc X\rvert}\left(\bm 1^Tw\right)^2,
\end{align*}
where $(a)$ is due to the fact that $w-\overline w\bm 1$ is orthogonal to the eigenspace of the eigenvector $1$ and the variational characterization of eigenvalues, to the symmetry of $P$ and \eqref{eq:P1=1} and \eqref{eq:X-scalp_normed}. In $(b)$, the binomial formula for $\lVert\cdot\rVert_{\mc X}^2$ was used, together with \eqref{eq:P1=1}. $(c)$ is a final application of \eqref{eq:X-scalp_normed}.
\end{proof}

\section{Discussion of secrecy criteria}\label{app:strongsec}

Here we show that for the types of channels considered in this paper, semantic security and strong secrecy effectively are the same concepts in terms of achievable rates, in the sense that every achievable strong secrecy rate also is an achievable semantic security rate. To make this statement formal, some details must be taken into account.

Let a one-shot wiretap channel $(T,U)$ and a seeded wiretap code $(\xi,\zeta)$ with seed set $\mc S$ and message set $\mc M$ be given. For a subset $\mc M'$ of $\mc M$, the seeded wiretap code $(\xi\rvert_{\mc M'},\zeta\rvert_{\mc M'})$ is the \textit{restriction} of $(\xi,\zeta)$ to $\mc M'$ whose message set is $\mc M'$, the seed set remains $\mc S$, and $\xi\rvert_{\mc M'}$ is the restriction of $\xi$ to inputs from $\mc S\times\mc M'$ whereas $\zeta\rvert_{\mc M'}$ operates like $\zeta$, but declares an error if it decodes an $m\notin\mc M'$. Clearly, $e(\xi\rvert_{\mc M'},\zeta\rvert_{\mc M'})\leq e(\xi,\zeta)$. Renes and Renner \cite{RR_codingfromaplif} obtained a result for the relation of strong secrecy and semantic security when formulated in terms of the total variation distance. More precisely, for a seeded wiretap code $(\xi,\zeta)$ with message set $\mc M$ and seed set $\mc S$, we define
\[
    L_\sem^{\lVert\cdot\rVert}(\xi,\zeta)
    =\max_{P_M}\lVert P_{ZSM}-P_{ZS}\otimes P_M\rVert,
\]
where the maximum ranges over all probability distributions on $\mc M$ such that $M$ is independent of $S$ and $Z$ is generated by $M$ and $S$ via $\xi U$. The analog of this criterion in terms of strong secrecy is 
\[
    L_\str^{\lVert\cdot\rVert}(\xi,\zeta)=\lVert P_{\overline{Z}S\overline{M}}-P_{\overline ZS}\otimes P_{\overline M}\rVert,
\]
where $\overline M$ is uniformly distributed on $\mc M$ and independent of $S$ and $\overline Z$ is generated by $M$ and $S$ via $\xi U$. 

\begin{lem}[\cite{RR_codingfromaplif}, Lemma 1]\label{lem:RenesRenner}
    For any seeded wiretap code $(\xi,\zeta)$ for $(T,U)$ with message set $\mc M$, there exists a subset $\mc M'$ of $\mc M$ with $\lvert\mc M'\rvert\geq\lvert\mc M\rvert/2$ such that
    \[
        L_\sem^{\lVert\cdot\rVert}(\xi,\zeta)
        \leq 4L_\str^{\lVert\cdot\rVert}(\xi\rvert_{\mc M'},\zeta\rvert_{\mc M'}).
    \]
\end{lem}

For the case where security is measured in terms of mutual information, Hayashi and Matsumoto observed the analogous relation in a special situation \cite[Section XIII]{HayMat_Multiplex}. Before we give a general result in this direction, we recall two information-theoretic inequalities. \textit{Pinsker's inequality} states that 
\[
    \lVert P-Q\rVert^2\leq 2\ln(2)D(P\Vert Q)
\]
for probability measures $P$ and $Q$. If $X$ is a random variable on the finite set $\mc X$ and $Y$ an arbitrary random variable such that $P_{XY}$ has a density, then it was shown in \cite[Lemma 1]{LLBS_polar_semsec} that
\begin{equation}\label{eq:muti<totvar}
    I(X\wedge Y)\leq-\lVert P_{XY}-P_X\otimes P_Y\rVert\log\frac{\lVert P_{XY}-P_X\otimes P_Y\rVert}{\lvert\mc X\rvert}.
\end{equation}

\begin{lem}\label{lem:semsecfromstrongsec}
	Let $(T,U)$ be a wiretap channel and $(\xi,\zeta)$ a seeded wiretap code with message set $\mc M$ satisfying $L_\str(\xi,\zeta)\leq\eta$. Then there exists a subset $\mc M'$ of $\mc M$ with $\lvert\mc M'\rvert\geq\lvert\mc M\rvert/2$ such that the restriction $(\xi\rvert_{\mc M'},\zeta\rvert_{\mc M'})$ of $(\xi,\zeta)$ to $\mc M'$ satisfies
	\begin{equation}\label{eq:semtvstrdiv}
	    L_\sem^{\lVert\cdot\rVert}(\xi\rvert_{\mc M'},\zeta\rvert_{\mc M'})
	    \leq6\sqrt{\eta}
	\end{equation}
    and
	\begin{equation}\label{eq:semdivstrdiv}
		L_\sem(\xi\rvert_{\mc M'},\zeta\rvert_{\mc M'})
		\leq-6\sqrt{\eta}\log\frac{6\sqrt{\eta}}{\lvert\mc M'\rvert}.
	\end{equation}
\end{lem}

\begin{proof}
    Pinsker's inequality implies $L_\str^{\lVert\cdot\rVert}(\xi,\zeta)^2\leq 2\ln (2) L_\str(\xi,\zeta)$. Thus by Lemma \ref{lem:RenesRenner} there exists an $\mc M'\subset\mc M$ with $\lvert\mc M'\rvert\geq\lvert\mc M\rvert/2$ such that
    \[
        L_\sem^{\lVert\cdot\rVert}(\xi\rvert_{\mc M'},\zeta\rvert_{\mc M'})\leq 4L_\str^{\lVert\cdot\rVert}(\xi,\zeta)\leq4\sqrt{2\eta\ln 2}\leq 6\sqrt\eta,
    \]
    which proves \eqref{eq:semtvstrdiv}. To obtain \eqref{eq:semdivstrdiv}, one applies \eqref{eq:muti<totvar}.
\end{proof}

If $(\xi_n,\zeta_n)_{n=1}^\infty$ is a sequence of seeded wiretap codes which is strongly secure asymptotically, then by Lemma \ref{lem:semsecfromstrongsec} a sequence $(\xi_n',\zeta_n')_{n=1}^\infty$ of subcodes exists such that for every $n$, the maximal error probability of $(\xi_n',\zeta_n')$ is no larger than that of $(\xi_n,\zeta_n)$, the semantic security information leakage of $(\xi_n',\zeta_n')$ can be upper-bounded in terms of the strong secrecy information leakage of $(\xi,\zeta)$, and the message set of $(\xi_n',\zeta_n')$ is at least half as large as that of $(\xi_n,\zeta_n)$. Thus there is no asymptotic rate loss when passing from $(\xi_n,\zeta_n)_{n=1}^\infty$ to $(\xi_n',\zeta_n')_{n=1}^\infty$. In any case, semantic security holds for $(\xi_n',\zeta_n')_{n=1}^\infty$ in terms of total variation distance due to \eqref{eq:semtvstrdiv}. For semantic security in terms of mutual information, the decrease of the strong secrecy information information leakage of $(\xi_n,\zeta_n)_{n=1}^\infty$ has to be sufficiently fast to allow for the semantic security information information leakage of $(\xi_n',\zeta_n')_{n=1}^\infty$ to vanish asymptotically as well, due to the worse bound \eqref{eq:semdivstrdiv}. A sufficient condition for semantic security to hold in terms of mutual information is that the strong secrecy leakage of $(\xi_n,\zeta_n)$ decreases to zero exponentially.

Note that the method of passing from strong secrecy to semantic security presented in this appendix is highly nonconstructive. In general, it will not be possible to find the subset of messages for which semantic security holds. Moreover, even if such a subset could be found, it would generally strongly depend on the wiretap channel. Both of these shortcomings are remedied by modular BRI schemes.

\section{Discussion of construction of \cite{HayMat_Multiplex}}\label{app:HayMat}

The paper \cite{HayMat_Multiplex} of Hayashi and Matsumoto treats a rather general message transmission problem, but one construction is also interesting in the context of seeded modular coding schemes for the wiretap channel. It was not observed in \cite{HayMat_Multiplex} that this construction also ensures semantic security directly. Instead, an expurgation argument as in Appendix \ref{app:strongsec} is applied to obtain semantic security from strong secrecy, which, as discussed, is highly nonconstructive.

We describe the type of security component of \cite{HayMat_Multiplex} in our notation and reduced to the simpler wiretap setting investigated in our paper. Let $W:\mc X\to\mc Z$ be a channel with finite input and output alphabets, where the input alphabet $\mc X$ is equipped with a group structure.  We write the group operation on $\mc X$ additively, but commutativity is not necessary. $W$ can be the concatenation of the encoder of an error-correcting code and the actual physical channel from the sender to the eavesdropper, just like in Theorem \ref{thm:EV-UB}. Let $\mc M$ be the message set and $\mc V$ a set of local randomness, both equipped with a group structure. The basic building blocks of the security components of \cite{HayMat_Multiplex} are functions $g:\mc S_1\times\mc M\times\mc V\to\mc X$ such that 
\begin{enumerate}
    \item $g(s_1,\cdot,\cdot)$ is an injective group homomorphism from the product $\mc M\times\mc V$ to $\mc X$ for every $s_1\in\mc S_1$, and 
    \item $\mbb P[g(S_1,m,v)=x]\leq(\lvert\mc X\rvert-1)^{-1}$ for every $x$ not equal to the neutral element of $\mc X$ and $S_1$ uniformly distributed on $\mc S$.
\end{enumerate}

The security component $f$ determined by $g$ has the seed set $\mc S_1\times\mc X$, hence $f:\mc S_1\times\mc X\times\mc X\to\mc M$. For given seed $s=(s_1,s_2)\in\mc S_1\times\mc X$ and message $m\in\mc M$, we define it by the random variable $X_{g,s,m}$ on $\mc X$ whose distribution is $f_s^{-1}(\cdot\vert m)$. For the seed $s=(s_1,s_2)\in\mc S_1\times\mc X$ and the message $m$, the random variable with distribution $f_s^{-1}(\cdot\vert m)$ is given by $g(s_1,m,V)+s_2$, for $V$ uniformly distributed on $\mc V$ and independent of all other random variables. Since $g(s_1,\cdot,\cdot)$ is injective, it is possible to define $f$ by the condition $f(s_1,s_2,g(s_1,m,v)+s_2)=m$ and $f(s_1,s_2,x)$ arbitrary if $x$ is not equal to $g(s_1,m,v)+s_2$ for any $m$ and $v$.

For the random seed $S$ uniformly distributed on $\mc S_1\times\mc X$, \cite[Lemma 21]{HayMat_Multiplex} implies
\begin{equation}\label{eq:haymat}
    I(M\wedge Z\vert S)
    \leq\frac{1}{\ln 2}\min_{1<\alpha\leq 2} 2^{-(\alpha-1)(\log\lvert\mc V\rvert-D_\alpha(W\Vert P_{\mc X}W\vert P_{\mc X}))},
\end{equation}
where 
\[
    D_\alpha(W\Vert Q\vert P)=\frac{1}{\alpha-1}\log\sum_{x\in\mc X}p(x)\sum_{z\in\mc Z}\frac{w(z\vert x)^\alpha}{q(z)^{\alpha-1}}
\]
is the conditional R\'enyi $\alpha$-divergence if $\mu(w(z\vert x)>0,q(z)=0)=0$ for all $x\in\mc X$. It is not hard to extend this result to allow for subnormalized channels\footnote{In \cite[Lemma 21]{HayMat_Multiplex}, $V$ is allowed to have an arbitrary distribution. However, the best upper bound is achieved by a uniformly distributed $R$.}. Then, this inequality can be applied in the same way as Theorem \ref{thm:EV-UB} in our paper.

Whereas the security components used by Hayashi and Matsumoto are algebraic in nature, biregular irreducible functions are based on graph-theoretic concepts. The precise relation of the two concepts is not yet clear. Theorem \ref{thm:EV-UB} is intimately tied to the case $\alpha=2$ due to the variational characterization of eigenvalues employed in its proof. In this respect, \eqref{eq:haymat} gives more flexibility. For the case $\alpha=2$, \eqref{eq:haymat} replaces the eigenvalue from Theorem \ref{thm:EV-UB} with the inverse of the size of the randomness necessary in the encoder. The size of the seed set $\mc S_1\times\mc X$ used by the $X_{g,s,m}$ is rather large, at least (roughly) twice the size of $\mc X$. Thus, seed transmission takes more seed reuses to compensate for the rate loss than in the case of biregular irreducible functions. The second component $s_2\in\mc X$ of the seed can be omitted if $\mc Z=\mc X$ and $W(z\vert x)=Q(z+x)$ for some probability measure $Q$ on $\mc X$ and all $x,z\in\mc X$. This condition is a very strong symmetry condition, much stronger than that required in \cite{BT_Poly_time}.

The only example of a function $g$ as above given in \cite{HayMat_Multiplex} is where $g:\mbb F_{2^\ell}^*\times\mbb F_{2^k}\times\mbb F_{2^b}\to\mbb F_{2^\ell}$, where $b=\ell-k$, and $g(s_1,m,v)=s_1\cdot(m+v)$. This is closely related to the unconstrained seeded coset function $\beta^o:\mbb F_{2^\ell}^*\times\mbb F_{2^\ell}\to\mbb F_{2^k}$ since $\beta^o(s_1,g(s_1,m,v))=m$ and $g(s_1,m,\mc V)=\{x:\beta^o(s_1,x)=m\}$. Thus, the distribution of $g(s_1,m,V)$ is given by $(\beta^o)_{s_1}^{-1}(\cdot\vert m)$. For $s=(s_1,s_2)\in\mbb F_{2^\ell}^*\times\mbb F_{2^\ell}$, the random variable $X_{g,s,m}$ equals $(\beta^o)_{s_1}^{-1}(m)+s_2$. This security component inherits from $\beta^o$ the property of being a universal hash function. In the case where $W(z\vert x)=Q(z+x)$ as above and $s_2$ is omitted, this security component is nothing other than the unconstrained seeded coset function itself, whose ability of ensuring semantic security for general symmetric channels already was established in \cite{BT_Poly_time} and \cite{TalVardy_Upgrading}. An efficient implementation of $\beta^o$ was discussed in \cite{TV_UHF_preprint}.

\section{Graphs: Definitions and Facts}\label{app:graphs}

In this appendix, some definitions and facts from graph theory are collected as a reference. 

\begin{defn}
    A \textit{graph} is a pair $G=(\mc X,\mc E)$, where $\mc X$ is a finite set and $\mc E=\mc E_1\cup\mc E_2$, where $\mc E_1$ is a subset of $\{(x,x):x\in\mc X\}$ and $\mc E_2$ is a subset of the set of 2-element subsets of $\mc X$. The elements of $\mc X$ are called \textit{vertices} (singular: \textit{vertex}) and the elements of $\mc E$ are called \textit{edges}. An element of $\mc E_1$ is also called a \textit{loop}. Two elements $x,x'\in\mc X$ are called \textit{adjacent} if either $x=x'$ and $(x,x)\in\mc E_1$ or $x\neq x'$ and $\{x,x'\}\in\mc E_2$. If $\mc E=\mc E_2$ (i.e., $G$ has no loops), then $G$ is called \textit{simple}.
\end{defn}

\textit{Interpretation:} A graph $G$ can be drawn if it is not too big. Vertices are represented by dots and edges by lines connecting these dots. Clearly, an edge which is a loop becomes a loop in the drawing. See Figure \ref{fig:examplegraph}.

\begin{figure}
    \centering
    \begin{tikzpicture}[vertex/.style={draw, circle, fill, inner sep=2pt}, edge/.style={-}
    ]
        \node[vertex] (one) {};
        \node[vertex, right = 2cm of one] (two) {};
        \node[vertex, below = 2cm of one] (three) {};
        
        \draw[edge] (one) -> (two);
        \draw[edge] (two) -> (three);
        \draw (2.7, 0) circle [radius = 0.5cm];

    \end{tikzpicture}
    \caption{A graph with three vertices and three edges, one of which is a loop.}\label{fig:examplegraph}
\end{figure}
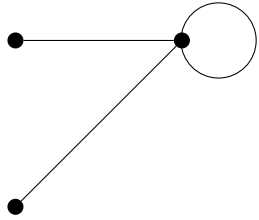

\begin{defn}
    A \textit{subgraph} of a graph $G=(\mc X,\mc E)$ is any graph $G'=(\mc X',\mc E')$ satisfying $\mc X'\subset\mc X$ and $\mc E'\subset\mc E$. (The vertices of the edges from $\mc E'$ must be from $\mc X'$.)
\end{defn}

\begin{defn}
    For every vertex $x$ of the graph $G=(\mc X,\mc E)$, its \textit{degree} $\deg(x)$ is defined as the number of vertices to which $x$ is adjacent. If $\deg(x)=d$ is constant in $x$, then $G$ is called \textit{$d$-regular}. In this case, $d$ is called the \textit{degree} of $G$.
\end{defn}

\begin{defn}
    The \textit{adjacency matrix} of the graph $G=(\mc X,\mc E)$ is an $\mc X\times\mc X$ matrix whose $(x,x')$ entry equals $1$ if $x$ and $x'$ are adjacent and $0$ else.
\end{defn}

Every adjacency matrix is symmetric. Therefore it can be diagonalized and has real eigenvalues.

\begin{defn}
    Let $G$ be a graph. 
    \begin{enumerate}
        \item A sequence $x_1,\ldots,x_n$ of vertices of $G$ is called a \textit{path} if $x_\xi$ is adjacent to $x_{\xi+1}$ for $0\leq\xi\leq n-1$. In this case, $x_1$ and $x_n$ are called the \textit{endvertices} of the path.
        \item A pair of vertices $x,x'$ is called \textit{connected} if there exists a path with endvertices $x$ and $x'$.
        \item $G$ is called \textit{connected} if every pair of vertices is connected.
        \item The \textit{distance} of two vertices $x,x'$ is the length of any shortest path connecting $x$ and $x'$. If $x,x'$ are not connected, then their distance is set to $+\infty$.
        \item The \textit{diameter} of $G$ is the maximal distance between any two vertices of $G$. 
    \end{enumerate}
\end{defn}

Connectedness is an equivalence relation on the vertex set. The equivalence classes are called \textit{connected components}. Between any two vertices contained in the same connected component, there exists a path connecting the two vertices. If the two vertices are not contained in the same connected component, no such path exists.

\section*{Acknowledgment}

H. B. would like to thank Eike Kiltz for motivating discussions on semantic security and cryptographic applications. He would also like to that Manfred Lochter of the German Federal Office for Information Security (BSI) for stimulating discussions on the security for wiretap channels, in particular for infinite alphabets, and on the operational meaning of semantic security.

\end{document}